\def\extended{true}

\PassOptionsToPackage{usenames,dvipsnames,table}{xcolor}

\documentclass[a4paper]{article}

\usepackage{currfile}

\usepackage{arxiv-a4}

\usepackage[numbers, compress]{natbib}

\usepackage[utf8]{inputenc} \usepackage[T1]{fontenc}    \usepackage{hyperref}       \usepackage{url}            \usepackage{booktabs}       \usepackage{amsfonts}       \usepackage{nicefrac}       \usepackage{microtype}      \usepackage{xcolor}         

\newcommand{\ifextended}[2]{\ifdefined\extended{#1}\else{#2}\fi}

\newcommand{\labelT}[1]{\label{#1}} 

\ifdefined\extended
\newcommand{\extref}[1]{\ref{#1}}
\newcommand{\extCref}[1]{\Cref{#1}}
\newcommand{\extCshref}[1]{\Cshref{#1}}

\else
\usepackage{xr}
\newcommand{\extref}[1]{\ref{ext-#1}}
\makeatletter
\newcommand{\extCref}[1]{\Cref{ext-#1}}
\newcommand{\extCshref}[1]{\Cshref{ext-#1}}

\makeatother
\fi

\usepackage{algorithmic}

\newcommand\refpage[1]{\ref{#1}, p.\,\pageref{#1}}

\newcommand\cshrefpage[1]{\cshref{#1}, p.\,\pageref{#1}}

\usepackage[utf8]{inputenc}

\usepackage{pgf} \usepackage{tikz}
\usetikzlibrary{arrows,shapes,decorations,automata,backgrounds,petri}
\usepackage[most]{tcolorbox}

\usepackage{booktabs}
\usepackage{multirow}

\usepackage{siunitx}

\usepackage{xspace}
\usepackage{graphicx}
\graphicspath{{./figures/}{./}{../ICML-21/figures/}{./logs2}}

\usepackage[export]{adjustbox}

\usepackage{enumitem}
\setdescription{leftmargin=\widthof{[$\beta^i_\depth$]~}}

\usepackage{tabularx}
\usepackage{rotating}
\usepackage{pifont}

\usepackage[normalem]{ulem}
\usepackage{verbatim}
\usepackage[most]{tcolorbox}

\usepackage{scalefnt}
\usepackage{amsmath} \makeatletter
\DeclareRobustCommand*\cal{\@fontswitch\relax\mathcal}
\makeatother
\usepackage{amsthm}
\usepackage{amssymb}
\usepackage{bbold} \usepackage{empheq}
\usepackage{cancel}
\newsavebox\CBox
\newcommand\hcancel[2][.5pt]{\ifmmode\sbox\CBox{$#2$}\else\sbox\CBox{#2}\fi \makebox[0pt][l]{\usebox\CBox}\rule[0.2em-#1/2]{\wd\CBox}{#1}}

\usepackage{xpatch}
\usepackage{thmtools}
\usepackage{apptools}
\makeatletter
\xpatchcmd{\thmt@restatable}{\csname #2\@xa\endcsname\ifx\@nx#1\@nx\else[{#1}]\fi}{\IfAppendix{\csname #2\@xa\endcsname}{\csname #2\@xa\endcsname\ifx\@nx#1\@nx\else[{#1}]\fi}}{}{} \makeatother

\usepackage{xfp}

\allowdisplaybreaks
\usepackage{centernot}

\DeclareMathOperator*{\argmax}{arg\,max}
\DeclareMathOperator{\E}{\mathbb{E}}
\DeclarePairedDelimiter\ceil{\lceil}{\rceil}

\DeclareRobustCommand\iff{\;\Leftrightarrow\;}
\newcommand{\twodots}{\mathinner {\ldotp \ldotp}}
\usepackage{xspace}
\def\ie{{\em i.e.}\xspace}
\def\eg{{\em e.g.}\xspace}
\def\cf{{\em cf.}\xspace}

\def\reals{{\mathbb R}}
\def\hmm{\text{\sc hmm}\xspace}
\def\cS{{\cal S}}
\def\cA{{\cal A}}
\def\cB{{\cal B}}

\def\cZ{{\cal Z}}
\def\occ{\sigma}

\def\os{{\sc os}\xspace}
\def\aoh{{\sc aoh}\xspace}

\def\Occ{{\cal O}^\sigma}
\newcommand{\PP}[4]{P_{#2}^{#4}(#3|#1)}

\def\cI{{\cal I}}

\def\l{\lambda}
\newcommand\lt[1]{\lambda_{#1}} \def\p{1} \def\WUL{W} \def\va{{\boldsymbol{a}}}

\def\vz{{\boldsymbol{z}}}
\def\vth{{\boldsymbol{\theta}}}
\def\vTh{{\boldsymbol{\Theta}}}

\def\vbeta{{\boldsymbol{\beta}}}

\def\dr{{\sc dr}\xspace}

\def\vx{{\boldsymbol{x}}}
\def\vy{{\boldsymbol{y}}}

\def\vv{{\boldsymbol{v}}}

\def\xx{x}
\def\yy{y}
\def\zz{z}
\def\nev{\textsc{nev}} \def\nes{\textsc{nes}}  \def\nxt{T}
\newcommand{\lp}[1]{\textsc{lp}#1}
\newcommand{\dlp}[1]{\textsc{dlp}#1}

\def\depth{\tau}
\def\supp{\mathit{Supp}}
\def\radius{\rho}
\def\thr{thr}

\def\t{\intercal} \def\indep{\perp \!\!\! \perp}

\def\thr{thr}
\def\nxt{T}
\def\Nxt{T} 

\def\vr{r}
\def\t{\top}
\def\Lin{Lin}

\def\ze{she} \def\zer{her}

\def\omgHSVIlc{{\sc omg}HSVI$^\text{\sc lc}$}
\def\omgHSVIlccc{{\sc omg}HSVI$^\text{\sc lc}_\text{\sc cc}$}

\usepackage{multicol}
\usepackage[ruled,vlined,linesnumbered]{algorithm2e}
\SetKwProg{Fct}{Fct}{}{}

\newtheorem{theorem}{Theorem}

\newtheorem{lemma}{Lemma}

\newtheorem{proposition}{Proposition}

\newcommand{\eqdef}     {\stackrel{{\textrm{\rm\tiny def}}}{=}}

\newcommand{\poubelle}[1]{}
\newcommand\upb[1]{\overline{#1}} \newcommand\lob[1]{\underline{#1}} \newcommand{\h}[3]{h_{#2}} 

\ifdefined\noproofs
\usepackage{environ}
\NewEnviron{killcontents}{}

\else{}\fi

\makeatletter \newcommand{\pushright}[1]{\ifmeasuring@#1\else\omit\hfill$\displaystyle#1$\fi\ignorespaces}
\newcommand{\pushleft}[1]{\ifmeasuring@#1\else\omit$\displaystyle#1$\hfill\fi\ignorespaces}
\newcommand{\specialcell}[1]{\ifmeasuring@#1\else\omit$\displaystyle#1$\ignorespaces\fi}
\makeatother

\DeclarePairedDelimiter{\abs}{\lvert}{\rvert}\DeclarePairedDelimiter{\norm}{\lVert}{\rVert}
\usepackage{esvect}
\newcommand\vabs[1]{\vv{\abs{#1}}}
\newcommand\vnorm[1]{\vv{\norm{#1}}}
\def\vleq{\ \vec\leq\ }

\def\player{}
\def\Player{}

\usepackage[nottoc,notlot,notlof]{tocbibind}

\newcommand{\citetg}[1]{\citeauthor{#1}'s \citep{#1}}

\usepackage{pgf} \usepackage{tikz}
\usetikzlibrary{trees}
\usetikzlibrary{positioning}

\definecolor{pink}{rgb}{0.858, 0.188, 0.478}

\newcommand{\persComment}[3]{
  \ifmmode
  \text{\textcolor{#3}{[#2] #1}}
  \else
  \textcolor{#3}{[#2] \em #1}
  \fi
}
\newcommand{\Jilles}[1]{\persComment{#1}{jsd}{violet}} \newcommand{\Vincent}[1]{\persComment{#1}{vt}{blue}} \newcommand{\Olivier}[1]{\persComment{#1}{ob}{teal}} \newcommand{\Abdallah}[1]{\persComment{#1}{as}{brown}} \newcommand{\Aurelien}[1]{\persComment{#1}{ad}{pink}}

\ifdefined\extended
\else
\renewcommand{\Jilles}[1]{} \renewcommand{\Vincent}[1]{} \renewcommand{\Olivier}[1]{} \renewcommand{\Abdallah}[1]{} \renewcommand{\Aurelien}[1]{} \fi

\renewcommand{\Jilles}[1]{\merde} \renewcommand{\Vincent}[1]{\merde} \renewcommand{\Olivier}[1]{\merde} \renewcommand{\Abdallah}[1]{\merde} \renewcommand{\Aurelien}[1]{\merde} 

\usepackage[noabbrev]{cleveref}
\DeclareRobustCommand{\abbrevcrefs}{\Crefname{appendix}{App.}{Apps.}\Crefname{section}{Sec.}{Secs.}\Crefname{equation}{Eq.}{Eqs.}\Crefname{figure}{Fig.}{Figs.}\Crefname{algorithm}{Alg.}{Algs.}\Crefname{tabular}{Tab.}{Tabs.}\Crefname{lemma}{Lem.}{Lems.}\Crefname{corollary}{Cor.}{Cors.}\Crefname{theorem}{Thm.}{Thms.}\Crefname{proposition}{Prop.}{Props.}\crefname{appendix}{app.}{apps.}\crefname{section}{sec.}{secs.}\crefname{equation}{eq.}{eqs.}\crefname{figure}{fig.}{figs.}\crefname{algorithm}{alg.}{algs.}\crefname{tabular}{tab.}{tabs.}\crefname{lemma}{lem.}{lems.}\crefname{corollary}{cor.}{cors.}\crefname{theorem}{thm.}{thms.}\crefname{proposition}{prop.}{props.}}
\DeclareRobustCommand{\cshref}[1]{{\abbrevcrefs\cref{#1}}}
\DeclareRobustCommand{\Cshref}[1]{{\abbrevcrefs\Cref{#1}}}

\title{
HSVI for zs-POSGs using \linebreak Concavity, Convexity and Lipschitz Properties }

\renewcommand{\headeright}{Preprint - \today}
\author{Aurélien Delage\\
  Univ. Lyon, INSA Lyon, INRIA, CITI,\\
  F-69621 Villeurbanne, France\\
  \texttt{firstname.lastname@inria.fr} \\
  \And
  Olivier Buffet\\
  Université de Lorraine, INRIA, CNRS, LORIA,\\
  F-54000 Nancy, France\\
  \texttt{firstname.lastname@inria.fr} \\
  \And
  Jilles Dibangoye\\
  Univ. Lyon, INSA Lyon, INRIA, CITI,\\
  F-69621 Villeurbanne, France\\
  \texttt{firstname.lastname@inria.fr} \\
}

\begin{document}

\maketitle

\fbox{
  \begin{minipage}{1.0\textwidth}
    \uline{Warning:}
The work presented in this paper has been improved along several
    lines in \cite{DelBufDibAbd-corr22}: readability, clarity of
    topics such as safety (global consistency), relation to the work
    of \citeauthor{WigOliRoi-ecai16} and to continual resolving
    approaches.
    \end{minipage}
}

\medskip

\begin{abstract}
  Dynamic programming and heuristic search are at the core of
  state-of-the-art solvers for sequential decision-making problems.
In partially observable or collaborative settings (\eg, POMDPs and
  Dec-POMDPs), this requires introducing an appropriate statistic that
  induces a fully observable problem as well as bounding (convex)
  approximators of the optimal value function.
This approach has succeeded in some subclasses of 2-player zero-sum
  partially observable stochastic games (zs-POSGs) as well, but failed
  in the general case despite known concavity and convexity
  properties, which only led to heuristic algorithms with poor convergence
  guarantees.
We overcome this issue, leveraging on these properties to derive
  bounding approximators and efficient update and selection operators,
before deriving a prototypical solver inspired by HSVI that
  provably converges to an $\epsilon$-optimal solution in finite time,
  and which we empirically evaluate.
This opens the door to a novel family of promising approaches
  complementing those relying on linear programming or iterative
  methods.
\end{abstract}

\message{<<< Entering \currfilename <<<}

\section{Introduction}
\label{sec|introduction}

Solving imperfect information sequential games is a challenging field
with many applications from playing Poker \citep{Kuhn-ctg50} to
security games \citep{BasNitGat-aaai16}.
We focus on finite-horizon 2-player 0-sum partially observable
stochastic games ((2p) zs-POSGs) an important class of games coming
with compact problem representations that allow for exploiting
structure (\eg, to derive relaxations).
From the viewpoint of (maximizing) Player~$1$, we aim at finding a
strategy with a worst-case expected return (\ie, whatever Player~$2$'s
strategy) within $\epsilon$ of the problem's Nash equilibrium value.

A first approach to solving a zs-POSG is to turn it into a 0-sum
extensive-form game (zs-EFG) \citep{OliVla-tr06}\footnote{Note: POSGs
  are equivalent to the large class of ``well-behaved'' EFGs as
  defined by \citet{KovSBBL-corr20}.}
addressed as a {\em sequence form} linear program 
\citep{KolMegSte-geb96,Stengel-geb96,BosEtAl-jair14}, giving rise to an exact algorithm.
A second approach is to use an iterative game solver, \ie, either a counterfactual-regret-based method (CFR)
\citep{ZinJohBowPic-nips07,BroSan-science18}, or a first-order method \citep{KroWauKKSan-mp20}, both coming with asymptotic convergence properties.
CFR-based approaches now incorporate deep reinforcement learning and
search, some of them winning against top human players at
heads-up no limit hold’em poker
\citep{MorEtAl-science17,BroSan-science18,BroBakLerQon-nips20}.

In contrast, dynamic programming and heuristic search have not been
applied to general zs-POSGs, while often at the core of
state-of-the-art solvers in other problem classes that involve
Markovian dynamics, partial observability and multiple agents (POMDP
\citep{Astrom-jmaa65,Smith-phd07}, Dec-POMDP
\citep{SzeChaZil-uai05,DibAmaBufCha-jair16}, or subclasses of
zs-POSGs with simplifying observability assumptions
\citep{GhoMcDSin-jota04,ChaDoy-tcl14,BasSte-jco15,HorBosPec-aaai17,ColKoc-jet01,HorBos-aaai19}).
They all rely on some statistic that induces a fully observable
problem whose value function ($V^*$) exhibits continuity properties
that allow deriving bounding approximations.
\citeauthor{WigOliRoi-ecai16} \citep{WigOliRoi-ecai16,WigOliRoi-corr16} contributed two
continuity properties, namely $V^*$'s concavity and convexity in two
different spaces, but which only led to heuristic algorithms with poor
convergence guarantees \citep{Wiggers-msc15}.

We here follow up on this work, successfully achieving the same 3-step
process as aforementioned approaches.
First, we obtain a fully observable game
for which Bellman's principle of optimality---the problem being made of nested subproblems---directly applies
(\Cref{sec|background}) by reasoning on the {\em occupancy state}
\citep{DibAmaBufCha-jair16} (also known for example as the {\em public
  (belief) state} \citep{MorEtAl-science17,BroBakLerQon-nips20}), \ie, the probability distribution over the players' past
action-observation histories.
Second, we exhibit novel continuity properties of optimal value
functions (not limited to $V^*$),
\ie, we extend
\citeauthor{WigOliRoi-corr16}'s
\citep{WigOliRoi-corr16,WigOliRoi-ecai16} continuity
properties (see also \citep{BroBakLerQon-nips20}), and introduce
complementary Lipschitz-continuity properties.
They allow proposing point-based upper and lower bound
approximations, and efficient update and selection operators based on
linear programming (\Cshref{sec|approxV}).
Third, we adapt \citeauthor{SmiSim-uai05}'
\citep{SmiSim-uai05} HSVI's algorithmic scheme to
$\epsilon$-optimally solve the problem in finitely many iterations
(\Cshref{sec|HSVI}).
In particular, we adopt the same changes to the algorithms, and thus
to the theoretical analysis of the finite-time convergence, as in the
work of
\citeauthor{HorBosPec-aaai17}\citep{HorBosPec-aaai17,HorBos-aaai19}.
These changes are required because, in both cases, the induced tree of
possible futures has an infinite branching factor.
\Cshref{sec|XPs} empirically validates the contributed algorithm.

\ifextended{}{
  \ifdefined\submission 
[Note: Appendices are in supplemental material here: {\small \url{https://members.loria.fr/OBuffet/papiers/jfpda21-ext.pdf} }.]
\else [Note: Appendices are in supplemental material.]
\fi
}

\section{Background}
\label{sec|background}

\uline{Note:} We may replace:
\begin{itemize}
\item  subscript ``$\depth:H-1$'' with ``$\depth:$'',
\item any function $f(\vx)$ linear in vector $\vx$ with either
$f(\cdot) \cdot \vx$ or ${\vx}^\t \! \cdot f(\cdot)$, and
\item a full tuple with its few elements of interest.
\end{itemize}

We first define zs-POSGs before recasting them into a new,
fully-observable, game (\Cshref{sec|oSG}). Then, concavity and convexity properties of this game's optimal value
function $V^*$ are presented (\Cshref{sec|CCV}), before introducing a {\em local} game and our first contributed
results, which will allow exploiting the nesting of subproblems
(\Cshref{sec|localGames}).

\ifextended{For the sake of clarity, the concepts and results of the EFG
  literature used in this work are recast in the POSG setting.
We will employ the terminology of behavioral strategies
  and strategy profiles---more convenient in our non-collaborative
  setting---instead of deterministic or stochastic policies (private
  or joint ones)---common in the collaborative setting of Dec-POMDPs.
}{}

A (2-player) zero-sum partially observable stochastic game (zs-POSG)
is defined by a tuple
$\langle \cS, \cA^1, \cA^2, \cZ^1, \cZ^2, P, r, H, \gamma, b_0 \rangle$, where
\ifdefined\extended \begin{itemize} \item \fi $\cS$ is a finite set of states;
  \ifdefined\extended \item \fi $\cA^i$ is (player) $i$'s finite set of actions;
  \ifdefined\extended \item \fi $\cZ^i$ is \player $i$'s finite set of observations;
  \ifdefined\extended \item \fi $\PP{s}{a^1,a^2}{s'}{z^1,z^2}$ is the probability to transition
  to state $s'$ and receive observations $z^1$ and $z^2$ when actions
  $a^1$ and $a^2$ are performed in state $s$;
\ifdefined\extended \item \fi $r(s,a^1,a^2)$ is a (scalar) reward function;
\ifdefined\extended \item \fi $H \in \mathbb{N}$ is a (finite) temporal horizon;
  \ifdefined\extended \item \fi $\gamma\in [0,1]$ is a discount factor; and \ifdefined\extended \item \fi $b_0$ is the initial belief state.
  \ifdefined\extended \end{itemize} \fi Player~$1$ wants to maximize the expected return, defined as the
discounted sum of future rewards, while \player $2$ wants to
minimize it, what we formalize next.

Due to the symmetric setting, many definitions and results are
given from a single player's viewpoint\ifextended{ when only obvious changes are
needed for the other}{}.

From the Dec-POMDP, POSG and EFG literature, we use the
following concepts and definitions, where $i \in \{1,2\}$:
\begin{description}[leftmargin=!,labelwidth=\widthof{$\vbeta_{\depth:\depth'}$}]
\item[\Player $-i$] is \player $i$'s opponent. Thus: $-1=2$, and $-2=1$.
\item[$\theta^i_\depth$] $ = (a^i_1, z^i_1, \dots , a^i_\depth,
  z^i_\depth) $ ($\in \Theta^i = \cup_{t=0}^{H-1} \Theta^i_t$) is a length-$\depth$ {\em action-observation history} (\aoh) for
  \player $i$.
\item[$\vth_\depth$] $ =(\theta^1_\depth,\theta^2_\depth)$ ($\in \vTh = \cup_{t=0}^{H-1} \vTh_t$) is a {\em joint \aoh} at $\depth$.
\item[{[$\beta^i_\depth$]}] A {\em (behavioral) decision rule} (\dr) at time
  $\depth$ for \player $i$ is a mapping $\beta^i_\depth$ from private
  \aoh{}s in $\Theta^i_\depth$ to {\em distributions} over private
  actions.
We note $\beta^i_\depth(\theta^i_\depth,a^i)$ the
  probability to pick \ifextended{action}{} $a^i$ when facing
  \ifextended{history}{} $\theta^i_\depth$.
\item[$\vbeta_\depth$]
  $= \langle \beta^1_\depth, \beta^2_\depth \rangle$
  ($\in \cB = \cup_{t=0}^{H-1} \cB_t$) is a {\em decision rule
    profile}.
\item[$\beta^i_{\depth:\depth'}$]
  $= (\beta^i_\depth, \dots, \beta^i_{\depth'})$ is a {\em behavioral
    strategy} (aka {\em policy}) for \player $i$ from time step $\depth$ to $\depth'$
  (included).
\item[$\vbeta_{\depth:\depth'}$]
  $= \langle \beta^1_{\depth:\depth'}, \beta^2_{\depth:\depth'}
  \rangle$ is a {\em behavioral strategy profile}  (aka {\em joint policy}).
\item[{[$V_0(\vbeta_{0:H-1}) $]}] The {\em value} of a strategy profile $\vbeta_{0:H-1}$ is \ifextended{
    \begin{align*}
      V_0(\vbeta_{0:H-1}) 
      & = E[\sum_{t=0}^{H-1} \gamma^t R_t \mid \vbeta_{0:H-1}],
    \end{align*}
  }{
    $
    V_0(\vbeta_{0:H-1}) 
    = E[\sum_{t=0}^{H-1} \gamma^t R_t \mid \vbeta_{0:H-1}],
    $
  }
  where $R_t$ is the random variable associated to the instant reward
  at $t$. 
\end{description}

The objective is here to find a Nash equilibrium strategy
(NES), \ie, a strategy profile
$\vbeta^*_{0:}=\langle\beta^{1*}_{0:},\beta^{2*}_{0:}\rangle$ such that no player has an
incentive to deviate, which can be written:
\begin{align*}
  \forall \beta^1, V_0(\beta^{1*}_{0:}, \beta^{2*}_{0:}) \geq V_0(\beta^{1}_{0:}, \beta^{2*}_{0:}),
  & \quad \text{and} \quad
  \forall \beta^2, V_0(\beta^{1*}_{0:}, \beta^{2*}_{0:}) \leq V_0(\beta^{1*}_{0:}, \beta^{2}_{0:}).
\end{align*}
In such a game, all NESs have the same Nash-equilibrium value (NEV)
$V^*_0 \eqdef V_0(\beta^{1*}_{0:}, \beta^{2*}_{0:})$.

As explained in the introduction, we aim at deriving an algorithm
based on dynamic programming or heuristic search, as in other
sequential decision-making problems.
Yet, Bellman's principle of optimality cannot be directly applied in a
game where players do not share their individual histories, and thus
do not have the same information about the current situation (except
at $\depth=0$).
To address this issue, we follow the same idea as for Dec-POMDPs or
some subclasses of zs-POSGs such as One-Sided POSGs
\cite{HorBosPec-aaai17} to consider a different game where each (new)
player controls an avatar, that interacts with the environment for
\zer{}, by publicly providing decision rules to be executed, but not
knowing which \aoh{s} are experienced.
This is what we do in the following section, demonstrating later
that one can retrieve solution strategies for the
original game which are robust to deviations.

\subsection{Re-casting POSGs as Occupancy Markov Games}
\label{sec|oSG}

Here, a different, fully observable, zero-sum game is derived
from the zs-POSG.
To that end, let us define the {\em occupancy state}
$\occ_{\vbeta_{0:\depth-1}}$ ($\in  \Occ_\depth$) as the probability
distribution over joint \aoh{}s $\vth_\depth$ given partial strategy
profile $\vbeta_{0:\depth-1}$.
This statistic exhibits the following properties (\cf also \citep[Thm.~1]{DibAmaBufCha-jair16}).

\begin{restatable}[Markov dynamics and rewards -- Proof in \extCshref{app|fromTo}]{proposition}{lemOccSufficient}
  \labelT{lem|occSufficient}
  \IfAppendix{{\em (originally stated on
      page~\pageref{lem|occSufficient})}}{}
  $\occ_{\vbeta_{0:\depth-1}}$, together with $\vbeta_\depth$, is a
  sufficient statistics to compute (i) the next \os, $\occ_{\vbeta_{0:\depth}}$, and (ii) the expected reward at $\depth$:
  $\E\left[ R_\depth \mid \vbeta_{0:\depth-1} \oplus \vbeta_\depth
  \right]$, where $\oplus$ denotes a concatenation.
\end{restatable}

These Markov properties
allow introducing an equivalent game (implicitely used by
\citet{WigOliRoi-corr16}), called a {\em zero-sum occupancy Markov
  game} (zs-OMG),\footnote{We use (i) ``Markov game'' instead of
  ``stochastic game'' because the dynamics are not stochastic, and
  (ii) ``partially observable stochastic game'' to stick with the
  literature.}
and defined by the tuple
$\langle \Occ, \cB, \nxt, r, H, \gamma \rangle$, where:
\ifdefined\extended \begin{itemize} \item \fi $\Occ (= \cup_{t=0}^{H-1} \Occ_t)$
  is the set of \os{}s induced by the zs-POSG;
\ifdefined\extended \item \fi $\cB$ is the set of \dr profiles of the zs-POSG;
\ifdefined\extended \item \fi $\nxt$ is a deterministic transition function that maps each pair
  $(\occ_\depth,\vbeta_\depth)$ to the (only) possible next \os
  $\occ_{\depth+1}$; formally (see proof of
  \extCshref{lem|occSufficient} in \extCshref{app|fromTo}),
  $\forall \theta^1_\depth, a^1, z^1, \theta^2_\depth, a^2, z^2$,
  \begin{align}
\label{eq|transition}
    \nxt(\occ_\depth,\vbeta_\depth)((\theta^1_\depth,a^1,z^1),(\theta^2_\depth,a^2,z^2))
& \eqdef Pr((\theta^1_\depth,a^1,z^1),(\theta^2_\depth,a^2,z^2) | \occ_\depth, \vbeta_\depth) \\
\nonumber
    & 
    = \beta^1_\depth(\theta^1_\depth,a^1) \beta^2_\depth(\theta^2_\depth,a^2) \occ_\depth (\vth_\depth)
    \sum_{s,s'} P^{\vz}_{\va}(s'|s) b(s|\vth_\depth)
    ,
  \end{align}
  where $b(s|\vth_\depth)$ is a belief state obtained by HMM filtering;
  \ifdefined\extended \item \fi $r$ is a reward function induced from the zs-POSG as the expected
  reward for the current \os and \dr profile:
\begin{align}
r(\occ_\depth,\vbeta_\depth) 
    \label{eq|reward}
    \eqdef E[r(S,A^1,A^2) | \occ_\depth, \vbeta_\depth ]
= \sum_{s,\vth_\depth, \va} \occ_\depth(\vth_\depth) b(s | \vth_\depth)
    \beta^1_\depth(a^1|\theta^1) \beta^2_\depth(a^2|\theta^2) r(s,\va);
  \end{align}
  we use the same notation $r$ for zs-POSGs as the context will
  indicate which one is discussed;
  \ifdefined\extended \item \fi $H$ and $\gamma$ are as in the zs-POSG
  ($b_0$ is not in the tuple but serves to define $\Nxt$ and $r$).
  \ifdefined\extended \end{itemize}  \fi The value of a strategy profile
$\vbeta_{0:}$ will be the same for both games, so that they share the
same \nev{} and \nes{}s.
We will see that, by computing the zs-OMG's \nev{}, we can obtain an
$\epsilon$-optimal zs-POSG solution strategy for $1$ or
$2$ as a by-product.

A zs-OMG is no standard finite zs Markov game since
(i) it is non-stationary, with different (continuous and of increasingly dimensionality) state and
action spaces at each time step;
(ii) at each time step, there are infinitely many actions, and a
mixture of such pure actions is equivalent to a pre-existing pure
action; and
(iii) the dynamics are deterministic (even for ``mixed'' actions).
But an important benefit of working with a zs-OMG is that both players
know the current \os,
$\occ_\depth$, thus always share the same information about the game.
This will allow finding an $\epsilon$-Nash equilibrium solution of
that game by exploring the tree of partial strategy
profiles.
This tree has an infinite branching factor due to the continuous
action (\dr) and state (\os) spaces.

We aim at using HSVI's algorithmic scheme, which relies on bounding
approximations of the optimal value function (updating them along
generated trajectories until $\epsilon$-convergence in the initial
point).
We thus now define such value functions and look at their known
structural properties, before building on them in \Cshref{sec|approxV}
to obtain the required bounding approximations.

\subsection{Concavity and Convexity (CC) Properties of $V^*$}
\label{sec|CCV}

A zs-OMG's {\em subgame} at $\occ_\depth$ is its restriction starting
from time step $\depth$ under this particular \os (thus looking for
strategies $\beta^1_{\depth:}$ and $\beta^2_{\depth:}$).
$\occ_\depth$ tells which \aoh{}s each player could be facing with
non-zero probability, and are thus relevant for planning.
We can then define the (optimal) value function in any
\os $\occ_\depth$ as follows:
\begin{align}
  \label{eq|subgame}
  V_\depth(\occ_\depth,\vbeta_{\depth:}) 
  \eqdef E[\sum_{t=\depth}^\infty \gamma^{t-\depth} r(S_t,A_t) | \occ_\depth, \vbeta_{\depth:}], 
  & \quad \text{and} \quad
  V_\depth^*(\occ_\depth) 
  \eqdef \max_{\beta_{\depth:}^1} \min_{\beta_{\depth:}^2} V_\depth ( \occ_\depth, \beta_{\depth:}^1,\beta_{\depth:}^2).
\end{align}

To study $V^*$, \citet{WigOliRoi-corr16} (whose results we extend here from $\gamma=1$
to $\gamma\leq 1$) decompose $\occ_\depth$, from $1$'s viewpoint, as a
{\em marginal} term,
$\occ_\depth^{m,1} (\theta_\depth^1) \eqdef \sum_{\theta_\depth^2} \occ_\depth(\theta_\depth^1,\theta_\depth^2) $,
and a {\em conditional} one,
$\occ_\depth^{c,1}(\theta_\depth^2 | \theta_\depth^1) \eqdef \frac{\occ_\depth (\theta^1_\depth,\theta^2_\depth) }{ \occ^{m,1}_{\depth} (\theta^1_\depth)}$,
so that $\occ_\depth = \occ_\depth^{m,1} \occ_\depth^{c,1}$.
In addition, $\nxt^1_m( \occ_{\depth}, \vbeta_\depth)$ and
$\nxt^1_c( \occ_{\depth}, \vbeta_\depth )$ here denote $1$'s marginal
and conditional terms associated to
$\nxt( \occ_{\depth}, \vbeta_\depth )$.

Given $\occ^{c,1}_\depth$ and a fixed \ifextended{strategy}{} $\beta^2_{\depth:}$, $1$ faces a POMDP, and  the optimal (POMDP) value in any \aoh
$\theta^1_\depth$  is given by
$ \nu^2_{[\occ_\depth^{c,1},\beta^{2}_{\depth:}]} (\theta^1_\depth)  \eqdef \max_{\beta_{\depth:}^1}  \mathbb{E} \left\{ \sum_{t = \depth}^{H-1} \gamma^{t-\depth} r(S_t,A_t^1,A_t^2) \mid \theta_\depth^1,\beta_{\depth:}^1, \beta_{\depth:}^{2}, \occ^{c,1}_\depth \right\}.$
\citeauthor{WigOliRoi-corr16} then demonstrate the following
concavity and convexity properties of $V^*$.

\begin{theorem}[Concavity and convexity (CC) of $V^*_\depth$ {\citep[Thm.~2]{WigOliRoi-corr16}}] \labelT{theo|ConvexConcaveV}
  For any $\depth \in \{0\twodots H-1\}$, $V_\depth^*$ is
(i) concave w.r.t. $\occ_\depth^{m,1}$ for a fixed
  $\occ_\depth^{c,1}$, and (ii) convex w.r.t. $\occ_\depth^{m,2}$ for a fixed
  $\occ_\depth^{c,2}$.
More precisely,
  \begin{align*}
    V_\depth^*(\occ_\depth)
    & = \min_{\beta_{\depth:}^2}\left[ \occ_\depth^{m,1} \cdot \nu^2_{[\occ_\depth^{c,1}, \beta_{\depth:}^2]} \right]
= \max_{\beta_{\depth:}^1}\left[ \occ_\depth^{m,2} \cdot \nu^1_{[\occ_\depth^{c,2}, \beta_{\depth:}^1]} \right].
  \end{align*}
\end{theorem}

However, this property alone allows approximating $V^*_\depth$ only
with a finite set of vectors
$\nu^{-i}_{[\occ_\depth^{c,i},\beta^{-i}_{\depth:}]}$, thus {\em only}
for finitely many conditional terms $\occ^{c,i}_\depth$, {\em not} for
the whole occupancy space.

\subsection{Introducing Local Games}
\label{sec|localGames}

{\em Sub}games (\Cshref{eq|subgame}) involve suffix strategies (over $\depth:H-1$),
while we aim at dealing with \dr{}s for one $\depth$ at a time.
Let us then introduce the {\em local} game at $\occ_\depth$, whose
payoff function is the {\em optimal action-value function} (which
assumes a known $V^*_{\depth+1}(\cdot)$) :
\begin{align}
  Q^*_\depth(\occ_\depth, \vbeta_\depth) & = r(\occ_\depth, \vbeta_\depth) + \gamma V^*_{\depth+1}( \nxt(\occ_\depth,\vbeta_\depth) ).
  \label{eq|localGame}
\end{align}
$Q_\depth^*(\occ_\depth, \cdot, \cdot)$ may not be
bilinear (cf. Prop.~\extref{prop|localGameNotBiLinear}, App.~\extref{prop|localGameNotBiLinear}), so that local games are not amenable to linear programming.
But \Cref{theo|ConvexConcaveV} (and $\Nxt$'s linearity in $\beta^1$
and $\beta^2$) leads to the following result.

\begin{restatable}[{\small New result} -- Proof in \extCshref{proofLemQcc}]{lemma}{lemQcc}
  \labelT{lem|Qcc}
  \IfAppendix{{\em (originally stated on
    page~\pageref{lem|Qcc})}}{}
$Q_\depth^*(\occ_\depth,\vbeta_\depth)$ is concave in $\beta_\depth^1$ and convex in $\beta_\depth^2$.
\end{restatable}

\citeauthor{Neu-ma28}'s minimax theorem \citep{Neu-ma28} thus applies,
and solution strategies can be obtained by respectively maximizing and
minimizing the following two intermediate value functions:
\begin{align*}
  W^{1,*}_\depth(\occ_\depth,\beta^1_\depth) 
  \eqdef  \min_{\beta^2_\depth} Q^*_\depth(\occ_\depth,\beta^1_\depth,\beta^2_\depth), 
  & \qquad \text{and} \qquad
  W^{2,*}_\depth(\occ_\depth,\beta^2_\depth) 
  \eqdef  \max_{\beta^1_\depth} Q^*_\depth(\occ_\depth,\beta^1_\depth,\beta^2_\depth).
\end{align*}
This, plus the inherent nesting of local games, opens the way to
applying Bellman's principle of optimality, building a solution of the
zs-OMG by concatening solutions of subsequent local games.

\section{Properties and Approximations of Optimal Value Functions} \label{sec|approxV}

$V^*$'s known structural properties have not been sufficient to derive
appropriate approximations for an HSVI-like algorithm.
We solve this issue by proving $V^*$'s Lipschitz continuity (LC)
(\Cshref{sec|LC|V}) before introducing approximations of $V^*$, $W^{1,*}$, and $W^{2,*}$
(\Cshref{sec|approximations}), and their related operators (\Cshref{sec|relatedOperators}).

\subsection{Lipschitz Continuity of $V^*$} \label{sec|LC|V}

Establishing $V^*$'s Lipschitz continuity starts with properties of
$\nxt$.

\begin{restatable}[Proof in \extCshref{proofLemOccLin}]{lemma}{lemOccLin}
  \labelT{lem|occ|lin}
  \IfAppendix{{\em (originally stated on
    page~\pageref{lem|occ|lin})}}{}
At depth $\depth$, $\nxt(\occ_\depth,\vbeta_\depth)$ is linear in
  $\beta^1_\depth$, $\beta^2_\depth$, and $\occ_\depth$, where
  $\vbeta_\depth=\langle \beta^1_\depth, \beta^2_\depth\rangle$.
It is more precisely $1$-Lipschitz-continuous ($1$-LC) in
  $\occ_\depth$ (in $1$-norm), \ie, for any $\occ_\depth$,
  $\occ'_\depth$:
  \begin{align*}
    \norm{\nxt(\occ'_\depth,\vbeta_\depth) - \nxt(\occ_\depth,\vbeta_\depth)}_1
    & \leq 1 \cdot \norm{\occ'_\depth - \occ_\depth}_1.
  \end{align*}
\end{restatable}

Also, the expected instant reward at any $\depth$ is linear in
$\occ_\depth$ (\cf proof of \Cref{lem|occSufficient},
\extCshref{app|fromTo}), and thus so is the expected value of a finite-horizon strategy profile
from $\depth$ onwards (\Cref{lem|V|lin|occ}).
This leads to $V^*_\depth$ being LC in $\occ_\depth$
(\Cref{cor|V|LC|occ}).

\begin{restatable}[Proof in \extCshref{proofLemVLinOcc}]{lemma}{lemVLinOcc}
  \labelT{lem|V|lin|occ}
  \IfAppendix{{\em (originally stated on
    page~\pageref{lem|V|lin|occ})}}{}
At depth $\depth$,
  $V_\depth(\occ_\depth,\vbeta_{\depth:})$ is linear w.r.t.  $\occ_\depth$.
\end{restatable}

\begin{restatable}[Proof in \extCshref{proofCorVLCOcc}]{theorem}{corVLCOcc}
  \labelT{cor|V|LC|occ}
  \IfAppendix{{\em (originally stated on
      page~\pageref{cor|V|LC|occ})}}{}
Let
  $\h{H}{\depth}{\gamma} \eqdef \frac{1-\gamma^{H-\depth}}{1-\gamma}$
(or $\h{H}{\depth}{\gamma} \eqdef H-\depth$ if $\gamma=1$).
Then $V^*_\depth(\occ_\depth)$ is $\lt{\depth}$-Lipschitz continuous in
  $\occ_\depth$ at any depth $\depth \in \{0 \twodots H-1\}$, where
  $\lt{\depth} = \frac{1}{2} \h{H}{\depth}{\gamma} \left( r_{\max} - r_{\min} \right)$. 
\end{restatable}

\subsection{Bounding Approximations of $V^*$, $W^{1,*}$ and $W^{2,*}$}
\label{sec|approximations}

We now derive bounding approximations of (i) $V^*$ to compute the gap
$\upb V_\depth(\occ_\depth) - \lob V_\depth(\occ_\depth)$, and (ii) $W^{1,*}$ and $W^{2,*}$ to efficiently solve the upper- and
lower-bounding local games at $\occ_\depth$ obtained by replacing
$V^*_{\depth+1}$ by $\upb V_{\depth+1}$ or $\lob V_{\depth+1}$ in
\Cshref{eq|localGame} when $\depth<H-1$.\footnote{At $\depth=H-1$, the
  game
  $\max_{\beta^1_{H-1}} \min_{\beta_{H-1}^2} \vr(\occ_{H-1},
  \beta^1_{H-1}, \beta^2_{H-1})$ is solved as an LP
  \citep[Sec.~3]{WigOliRoi-corr16}).}

\paragraph{Bounding $V^*$ --} Using both $V^*_\depth$'s concavity property
(\Cref{theo|ConvexConcaveV}) and its Lipschitz continuity
(\Cref{lem|occ|lin}) allows deriving the following upper bound approximation
(details in \extCshref{app|derivingApproximationsV}, including its
symmetric $\lob{V}_\depth(\occ_\depth)$) as the lower-envelope of
several upper bounds:
\begin{align*}
  \upb{V}_\depth(\occ_\depth)
  & = \min_{  \langle \tilde\occ_\depth^{c,1}, \upb\nu^2_\depth  \rangle \in \upb{bagV}_\depth } \left[ \occ_\depth^{m,1} \cdot \upb\nu^2_\depth + \lt{\depth} \norm{ \occ_\depth - \occ_\depth^{m,1}\tilde\occ_\depth^{c,1} }_1 \right],
\end{align*}
where $\upb{bagV}_\depth$ is a set of data points wherein, for each
$\tilde\occ_\depth^{c,1}$, (i) the vector $\upb\nu^2_\depth$ upper bounds
$\nu^2_{[\tilde\occ^{c,1}_\depth, \beta^2_{\depth:}]}$ for some
$\beta^2_{\depth:}$ (\Cshref{sec|CCV}), so that (ii) the scalar product gives an upper-bounding hyperplane under fixed
$\tilde\occ_\depth^{c,1}$, and (iii) the Lipschitz term allows generalizing to any $\occ_\depth^{c,1}$.

\paragraph{Upper Bounding $W^{1,*}$ --} $V^*_{\depth+1}$ being Lipschitz in $\occ_{\depth+1}$
(\Cref{cor|V|LC|occ}), and
exploiting linearity and independence properties of
$\nxt_m^1(\occ_\depth,\vbeta_\depth)$ and
$\nxt_c^1(\occ_\depth,\vbeta_\depth)$
(Lemmas~\extref{lem|T1mlin}+\extref{lem|T1cindep},
\extCshref{app|derivingApproximationsW}), we can derive an {\em upper} bound approximation $\upb{W}^1_\depth$
of $W^{1,*}_\depth$ (and conversely a {\em lower} bound
approximation $\lob{W}^2_\depth$ of $W^{2,*}_\depth$) by using finitely
many tuples $\langle \tilde\occ^{c,1}_\depth, \beta_\depth^2, \upb\nu^2_{\depth+1} \rangle$
stored in a set $\upb{bagW}^1_\depth$ (\cf
\extCshref{app|derivingApproximationsW}):
\begin{align}
  \label{eq|Wupb}
  \upb{W}^1_\depth(\occ_\depth, \beta^1_\depth) 
  &  =  \min_{  \substack{ \langle \tilde\occ^{c,1}_\depth, \beta^2_\depth, \upb\nu^2_{\depth+1} \rangle \in \upb{bagW}^1_\depth } }
  {\beta^1_\depth}^\t \cdot \Big[ {
    r(\occ_\depth, \cdot, \beta^2_\depth)
    + \gamma 
    \nxt_m^1(\occ_\depth, \cdot, \beta^2_\depth) \cdot \upb\nu^2_{\depth+1}
  } 
  \\
  & \qquad { 
    + \gamma \lt{\depth+1} \cdot \norm{ \nxt(\occ_\depth, \cdot, \beta^2_\depth) - \nxt_m^1(\occ_\depth, \cdot, \beta^2_\depth) \nxt^1_c(\tilde\occ^{c,1}_\depth, \beta^2_\depth) }_1
  } \Big].
\nonumber
\end{align}
For $\depth=H-1$, only the reward term is preserved.

We now look at the operators used to manipulate the approximations.

\subsection{Related Operators}
\label{sec|relatedOperators}

\paragraph{Selection Operator} As detailed in \extCshref{sec|getLP}, given a distribution
$\delta^2_\depth$ over tuples
$\langle \tilde\occ^{c,1}_\depth, \beta^2_\depth, \upb\nu^2_{\depth+1}
\rangle$ in $\upb{W}^1_\depth$, we can now upper bound the value of
``profile'' $\langle \beta^1_\depth, \delta^2_\depth \rangle$ when in
$\occ_\depth$ as
${\beta^1_\depth}^\t \! \cdot M^{\occ_\depth} \! \cdot \delta^2_\depth$, with $M^{\occ_\depth}$ an
$|\Theta^1_\depth \times \cA^1| \times |\upb{bagW}^1_\depth|$ matrix
(with null columns for improbable histories $\theta^1_\depth$ under
$\occ^1_\depth$).
Solving
$\max_{\beta^1_\depth} \upb{W}^1_\depth(\occ_\depth, \beta^1_\depth)$
can then be written as solving a zero-sum game where pure strategies are:
for $1$, the choice of $|\Theta^1_\depth|$ actions and,
for $2$, the choice of $1$ element of $\upb{bagW}^1_\depth$.
The corresponding \ifextended{linear program}{LP}, $\lp{\upb{W}^1_\depth}(\occ_\depth)$, is:
\begin{align}
  \label{eq|LP}
  & \begin{array}{l@{\ }l@{\ }ll}
      \displaystyle
      \max_{\beta_\depth^1,v}
      v
      \quad \text{s.t. } & \text{(i)}
      & \forall w \in \upb{bagW}^1_\depth, & v \leq {\beta_\depth^1}^\t \! \cdot M^{\occ_\depth}_{(\cdot,w)}
      \\
      & \text{(ii)}
      & \forall \theta_\depth^1 \in \Theta_\depth^1,
      & {\displaystyle \sum_{a^1}} \beta_\depth^1(a^1|\theta_\depth^1)
        = 1,
    \end{array}
    \intertext{and its dual, $\dlp{\upb{W}^1_\depth}(\occ_\depth)$, is}
\label{eq|DLP}
    & \begin{array}{l@{\ }l@{\ }ll}
        \displaystyle
        \min_{\delta^2_\depth,v}
        v
        \quad \text{s.t. }
        & \text{(i)}
        & \forall (\theta^1_\depth, a^1 ) \in \Theta^1_\depth\times \cA^1, & v \geq  M^{\occ_\depth}_{((\theta^1_\depth,a^1),\cdot)} \cdot \delta^2_\depth
        \\
        & \text{(ii)}
        &
        & {\displaystyle \sum_{w \in \upb{bagW}^1_\depth}} \! \delta^2_\depth( w )
          = 1.
      \end{array}
    \end{align}

\paragraph{Strategy Extraction} \label{sec|stratExtraction}
Any tuple $w_\depth \in \upb{bagW}^1_{\depth}$ contains \begin{itemize}
\item a default strategy for $2$ if this is an initial $w_\depth$, and
\item both (i) a decision rule $\beta^2_\depth[w_\depth]$, and
(ii) a probability distribution $\delta^2_{\depth+1}$ over tuples
  $w_{\depth+1} \in \upb{bagW}^1_{\depth+1}$ (unless $\depth=H-1$) otherwise.
\end{itemize}
As a consequence, each such tuple $w_\depth$ induces a
recursively-defined strategy for $2$.\footnote{This new space of {\em
    recursive strategies} trivially contains the space of behavioral
  strategies, which correspond to recursive strategies whose
  intermediate distributions are degenerate.}

$\delta^2_\depth$ needs to be stored as this strategy will play a key
role in the following, hence the new definitions of $\upb{V}_\depth$
and $\upb{W}^1_{\depth-1}$ (which rely essentially on the same
information (when $\depth>1$) and will be discussed together): $\upb{bagV}_\depth$ contains tuples
$\langle { \occ^{c,1}_\depth, \langle \delta^2_{\depth}, \upb\nu^2_\depth
  \rangle
} \rangle$,
and $\upb{bagW}^1_{\depth-1}$ (for $\depth\geq 1$) related tuples
$\langle { \occ^{c,1}_{\depth-1}, \beta^2_{\depth-1}, \langle \delta^2_{\depth}, \upb\nu^2_\depth
  \rangle } \rangle$.

Note: For convenience, we explain how to derive an equivalent behavioral strategy in \Cref{app|stratExtraction}.

\paragraph{Initializations}
\label{sec|bound|init}

One can look for an upper bound of \ifextended{the optimal value
  function}{} $V^*$, \ie, an optimistic bound (an admissible
heuristic) for (maximizing) player $1$, by relaxing the problem \ze{} faces.
To that end, we here solve the POMDP obtained when $2$ is assigned a
uniformly random $\beta^{2,\oslash}_{0:}$, the resulting best response
being noted $\beta^{1,\otimes}_{0:}$.
At depth $\depth$, $\langle \beta^{1,\otimes}_{0:}, \beta^{2,\oslash}_{0:} \rangle$ induces (i) an \os{} $\occ_\depth$ and (ii) a vector $\upb\nu^2_\depth$, where
$\upb\nu^2_\depth(\theta^1_\depth)$ is the value of
$\beta^{1,\otimes}_{0:}$ in $\theta^1_\depth$ (against
$\beta^{2,\oslash}_{0:}$ and under $\occ^{c,1}_\depth$).
Given these strategies, each $\upb{bagV}_\depth$ (respectively
$\upb{bagW}^1_{\depth-1}$) is initialized as $\{ \langle { \occ^{c,1}_\depth, \langle \delta^{2,\oslash}_\depth, \upb\nu^2_\depth \rangle } \rangle  \}$ (resp. $\{ \langle { \occ^{c,1}_{\depth-1}, \beta^{2,\oslash}_{\depth-1}, \langle \delta^{2,\oslash}_\depth, \upb\nu^2_\depth \rangle } \rangle \} $), where $\delta^{2,\oslash}_\depth$ is a degenerate distribution over the only
element in $\upb{bagW}^1_\depth$.

\paragraph{Updating $\upb{V}_\depth$ and $\upb{W}^1_{\depth-1}$} 

As depicted in \Cref{alg|zsPOSGwithLP+VWs+},
\crefrange{alg|UpdateFunction}{alg|UpdateFunction|end},
$\upb{V}_\depth$ and $\upb{W}^1_{\depth-1}$ are updated
simultaneously.
Given a tuple
$\langle \occ_\depth, \occ^{c,1}_{\depth-1}, \beta^2_{\depth-1}
\rangle$ (partly undefined if $\depth=0$), solving the dual LP (\ref{eq|DLP}), which relies on
$\upb{bagV}_{\depth+1}$, gives both $\delta^2_\depth$ and, as a
by-product, $\upb \nu^2_\depth$.
Indeed, assuming that (i) $2$ follows strategy $\delta^2_\depth$ and (ii) the expected return from $\depth+1$ on is given by
$\upb{V}_{\depth+1}$ ($=0$ if $\depth+1=H$), the value of $1$'s best action $a^1$ at $\theta^1_\depth$ is
upper bounded by:
\begin{align}
  \upb\nu^2_\depth(\theta^1_\depth)
  & =  \frac{1}{\occ^{1}_{\depth,m}(\theta^1_\depth)} \max_{a^1 \in \cA^1}
  M^{\occ_\depth}_{((\theta^1_\depth, a^1), . )} \cdot \delta^2_\depth
  & \text{(\cf \extCshref{prop|rec|nu}, \extCshref{sec|compute|nu}).}
  \label{eq|nu}
\end{align}
One then needs to add
$\langle \occ^{c,1}_\depth, \langle
\delta^2_\depth, \upb\nu^2_\depth \rangle \rangle$ to $\upb{bagV}_\depth$, and 
(if $\depth \geq 1$) $\langle \occ^{c,1}_{\depth-1}, \beta^2_{\depth-1}, \langle \delta^2_\depth, \upb\nu^2_\depth \rangle \rangle$ to $\upb{bagW}^1_{\depth-1}$.

\paragraph{Pruning}

Because they have different forms, $\upb{bagV}_\depth$ and
$\upb{bagW}^1_{\depth-1}$ have to be pruned independently.
Yet, to preserve the recursively defined strategies $\delta^2_\depth$,
pruned tuples should be kept in memory.

$\upb{V}_\depth$ relies on a ``$\min$-surfaces'' (rather than
``$\min$-planes'') representation, where each surface is linear in
$\occ^{m,1}_\depth$ and exploits the Lipschitz-continuity.
This allows exploiting (inverted) POMDP $\max$-planes pruning methods
so that, as explained in \extCref{lem|pruningV} (\extCshref{app|pruningV}), whether a test may induce false positives (pruning non-dominated
elements) or false negatives (not pruning dominated elements) carries
on from the $\min$-planes setting to our $\min$-surfaces setting.

For its part, $\upb{W}^1_\depth$ involving (i) a reward term that is bilinear (linear in both $\occ_\depth$ and
$\beta^1_\depth$), and (ii) a possibly non-continuous term,
deriving pruning techniques is not as
straightforward.
While solving local games may significantly benefit from pruning
$\upb{W}^1_\depth$, we leave this issue for future work.

\paragraph{About Improbable Histories}
To save on time and memory, we do not store \dr{}s and components of
vectors $\upb\nu^2_\depth$ for 0-probability \aoh{}s
$\tilde\theta^1_\depth$ in current $\occ^{m,1}_\depth$, which carry
little relevant information.
This leads to replacing, when computing $M^{\occ_\depth}$, undefined
components of vectors $\upb\nu^2_\depth$ by a heuristic overestimate
such as (\cf \extCshref{app|improbableAOHs}): $\upb\nu_{\text{init}}$
(not admissible), and $\upb\nu_{\text{bMDP}}$ (admissible).

\section{HSVI for zs-POSGs}
\label{sec|HSVI}

\subsection{Algorithm}

HSVI for zs-OMGs, which seeks $\epsilon$-optima, is described in \Cref{alg|zsPOSGwithLP+VWs+}.
As vanilla HSVI, it relies on
(i) generating trajectories while acting optimistically (lines
\ref{alg|greedP1}+\ref{alg|greedP2}), \ie, player $1$ (resp. $2$)
acting ``greedily'' w.r.t. $\upb{W}^1_\depth$ (resp. $\lob{W}^2_\depth$),
and
(ii) locally updating the upper and lower bound approximations
(lines \ref{alg|updateUpB}+\ref{alg|updateLoB}).
Both phases rely on solving the same
normal-form games described by LP~(\ref{eq|LP}).
At $\depth=H-1$, \cref{alg|oneShot} selects \dr{}s by solving the
exact game (\Cshref{sec|approximations}), and \cref{line|computeDelta}
returns a distribution reduced to the single element added in
\cref{alg|oneShotAddW}.
Note that the implementation maintains {\em full} occupancy states
$o_\depth \in \Delta(\cS \times \vTh_\depth)$, which allow easily
retrieving both ``simple'' \os{}s
$\occ_\depth \in \Occ_\depth=\Delta(\vTh_\depth)$ and ``beliefs''
$b(s|\vth_\depth)$.

A key difference with \citeauthor{SmiSim-uai05}' HSVI algorithm \citep{SmiSim-uai05} lies in
the criterion for stopping trajectories.
In vanilla HSVI (for POMDPs), the finite branching factor allows
looking at the convergence of $\upb{V}$ and $\lob{V}$ at each point reachable
under an optimal strategy.
To ensure $\epsilon$-convergence at $\occ_0$, trajectories just need
to be interrupted when the current width at $\occ_\depth$
( $\eqdef \upb{V}_\depth(\occ_\depth) - \lob{V}_\depth(\occ_\depth)$ )
is smaller than a threshold
$\gamma^{-\depth} \epsilon$.
(This happens even if $\gamma=1$ due to the approximation's width
falling to 0 beyond $H$.)
Here, dealing with an infinite branching factor, one may converge
towards an optimal solution while always visiting new points of the
occupancy space.
To counter this, we bound the width within balls around visited
points by exploiting $V^*$'s Lipschitz continuity.
This is achieved by adding a term
$- \sum_{i=1}^\depth 2 \radius \l_{\depth-i} \gamma^{-i}$ 
\citep{HorBosPec-aaai17} (even if $\gamma=1$)
to ensure that the width is below $\gamma^{_\depth} \epsilon$ within a
ball of radius $\radius$ around the current point (here
$\occ_\depth$),
hence the threshold
\begin{align}
  \label{eq|def|thr}
  \thr(\depth) & \eqdef
    \gamma^{-\depth}\epsilon - \sum_{i=1}^\depth 2 \radius \lt{\depth-i} \gamma^{-i}.
\end{align}

\begin{algorithm}
  \caption{zs-OMG-HSVI($b_0, [ \epsilon, \radius ]$) \small [here returning solution strategy $\lob\delta^1_0$ for Player $1$]}
  \label{alg|zsPOSGwithLP+VWs+}
\DontPrintSemicolon
  \SetKwFunction{FupbUpdate}{$\upb{\text{\bf Update}}$}
  \SetKwFunction{FlobUpdate}{$\lob{\text{\bf Update}}$}
  \SetKwFunction{FRecursivelyTry}{\textbf{Explore}} \SetKwFunction{FzsOMGHSVI}{\textbf{zs-OMG-HSVI}}

  {  \scalefont{.9}

    \begin{multicols}{2}

\Fct{\FzsOMGHSVI{$b_0 \simeq \occ_0$}}{ $\forall \depth \in 0\twodots H-1$, initialize  $\upb{V}_\depth$, $\lob{V}_\depth$, $\upb{W}^1_\depth$, \& $\lob{W}^2_\depth$  \;
        \While{ $\left[ \upb{V}_0(\occ_0)- \lob{V}_0(\occ_0) > \thr(0) \right]$ }{
          \FRecursivelyTry{$\occ_0, 0, -, -$}
        }
$ \delta_0^1 \gets {\displaystyle
          \argmax_{ \substack{
              \langle \occ^{c,1}_0, \langle \delta^1_0, \lob\nu^1_0 \rangle \rangle \in \lob{bagV}_0
            } } } \Big( {
          \occ^{m,2}_0 \cdot \lob\nu^1_0 +  \lt{0} \cdot 0
        } \Big)
        $ \;
        \Return{$\delta_0^1$} }

      \vfill
      
\Fct{\FupbUpdate{$\upb{V}_\depth, \upb{W}^1_{\depth-1}, \langle { \occ_\depth, \occ^{c,1}_{\depth-1}, \lob\beta^2_{\depth-1} } \rangle $ } }   {
        \label{alg|UpdateFunction}
$\langle \delta_\depth^2, \upb\nu^2_\depth \rangle \gets \dlp{\upb{W}^1_\depth}(\occ_\depth, \delta_\depth^2)$ \label{line|computeDelta} \;
$\upb{bagV}_\depth \gets \upb{bagV}_\depth \cup \{ \langle { \occ^{c,1}_\depth, \langle \delta^2_\depth, \upb\nu^2_\depth \rangle
        } \rangle \} $ \;
        $\upb{bagW}^1_{\depth-1} \gets \upb{bagW}^1_{\depth-1} \cup \{ \langle { \occ^{c,1}_{\depth-1}, \lob\beta^2_{\depth-1}, \langle \delta^2_\depth, \upb\nu^2_\depth \rangle
        } \rangle \} $
        \label{alg|UpdateFunction|end}
      }

\Fct{\FRecursivelyTry{$\occ_\depth, \depth, \occ_{\depth-1}, \vbeta_{\depth-1} $}}
      { \If{$\left[ \upb{V}_\depth(\occ) -\lob{V}_\depth(\occ) > thr(\depth) \right]$}{
\eIf{$\depth<H-1$}{
            $\upb\beta_\depth^1 \gets \lp{\upb{W}^1_\depth}(\occ, \beta_\depth^1)$ \label{alg|greedP1} \;
$\lob{\beta}_\depth^2 \gets \lp{\lob{W}^2_\depth}(\occ, \beta_\depth^2)$ \label{alg|greedP2} \;
\FRecursivelyTry{$ \nxt(\occ_\depth, \upb\beta_\depth^1, \lob\beta_\depth^2), \depth+1,\occ_\depth, \langle \upb\beta^1_\depth, \lob\beta^2_\depth \rangle
$}
          }({($\depth=H-1$)}){
            $(\upb\beta_\depth^1, \lob\beta^2_\depth) \gets \nes \left( r(\occ, \beta^1_\depth, \beta^2_\depth) \right)$ \label{alg|oneShot} \; 
$\upb{bagW}^1_{\depth} \gets \upb{bagW}^1_{\depth} \cup \{ \langle { \occ^{c,1}_{\depth}, \lob\beta^2_\depth, -
            } \rangle \} $ \label{alg|oneShotAddW} \; $\lob{bagW}^2_{\depth} \gets \lob{bagW}^2_{\depth} \cup \{ \langle { \occ^{c,2}_{\depth}, \upb\beta^1_\depth,  -
            } \rangle \} $ \; }
$\FupbUpdate ({ \upb{V}_\depth, \upb{W}^1_{\depth-1}, \langle \occ_\depth, \occ^{c,1}_{\depth-1}, \lob\beta^2_{\depth-1} \rangle }) $ \label{alg|updateUpB}\;
$\FlobUpdate ({ \lob{V}_\depth, \lob{W}^2_{\depth-1}, \langle \occ_\depth, \occ^{c,2}_{\depth-1}, \upb\beta^1_{\depth-1} \rangle }) $  \label{alg|updateLoB} \;
}
}

\end{multicols}
  }
  
  \medskip
\end{algorithm}

\paragraph{Setting $\radius$}

As can be observed, this threshold function should always return
positive values, which requires a small enough (but $>0$) $\radius$.
For a given problem (\cf \extCshref{lem|MaxRadius},
\extCshref{sec|settingRadius}), the maximum possible value
$\radius_{\max}$ depends on the Lipschitz constants at each time step,
which themselves depend on the initial upper and lower bounds of the
optimal value function.
\ifdefined\extended
But what is the effect of setting $\radius \in (0,\radius_{\max})$ to small or large values?
\ifdefined\extended \begin{itemize} \item \fi The smaller $\radius$, the larger $\thr(\depth)$, the shorter the
  trajectories, but the smaller the balls and the higher the required
  density of points around the optimal trajectory, thus the more
  trajectories needed to converge.
\ifdefined\extended \item \fi The larger $\radius$, the smaller $\thr(\depth)$, the longer the
  trajectories, but the larger the balls and the lower the required
  density of points around the optimal trajectory, thus the less
  trajectories needed to converge.
\ifdefined\extended \end{itemize} \fi So, setting $\radius$ means making a trade-off between the number of
generated trajectories and their length.
\else
Then, setting $\radius \in (0,\radius_{\max})$ means making a trade-off
between generating many trajectories (small $\radius$) and long ones
(large $\radius$).
\fi

\subsection{Finite-Time Convergence} 

\begin{restatable}[Proof in \extCshref{sec|ConvergenceProof}]{theorem}{thmTermination}
  \labelT{thm|termination}
  \IfAppendix{{\em (originally stated on page~\pageref{thm|termination})}}{}
zs-OMG-HSVI (\Cref{alg|zsPOSGwithLP+VWs+}) terminates in
  finite time with an $\epsilon$-approximation of $V^*_0(\occ_0)$.
\end{restatable}

\begin{proof}(sketch adapted from \citet{HorBos-aaai19})
  \label{sketch|theo|IntersectionHypercubeSimplexe}
Assume for the sake of contradiction that the algorithm does not
  terminate and generates an infinite number of explore trials.
Then, the number of trials of length $T$ (for some
  $0 \leq T \leq H$) must be infinite.
It is impossible to fit an infinite number of occupancy points
  $\occ_T$ satisfying $\norm{\occ_T-\occ'_T}_{\p} > \radius$ within $\Occ_T$.
There must thus be two trials of length $T$,
  $\{\occ_{\depth,1}\}_{\depth=0}^T$ and
  $\{\occ_{\depth,2}\}_{\depth=0}^T$, such that
  $\norm{\occ_{T-1,1}-\occ_{T-1,2}}_{\p} \leq \radius$,
and one can show (as detailed in \extCref{sec|ConvergenceProof})
  that the second trial should not have happened.
\end{proof}

The finite time complexity suffers from the same combinatorial
explosion as for Dec-POMDPs, and is even worse as we have to handle
"infinitely branching" trees of possible futures.
More precisely, the bound on the number of iterations depends on the
number of balls of radius $\radius$ required to cover occupancy
simplexes at each depth.

Also, the following proposition allows solving infinite horizon problems as
well (when $\gamma<1$) by bounding the length of HSVI's trajectories
using the boundedness of $\upb{V}-\lob{V}$ and the exponential growth of
$thr(\depth)$.

\begin{restatable}[Proof in \extCshref{proofLemFiniteTrials}]{proposition}{lemFiniteTrials}
  \labelT{lem|finiteTrials}
  \IfAppendix{{\em (originally stated on
      page~\pageref{lem|finiteTrials})}}{}
When $\gamma<1$, using the depth-independent Lipschitz constant
  $\l^\infty$, and with
  $\WUL \eqdef \norm{ \upb{V}^{(0)}-\lob{V}^{(0)}}_\infty$ the maximum
  width between initializations, the length of trajectories is
  upper bounded by
  $ T_{\max}
    \eqdef \ceil*{
      \log_{\gamma} \frac{
        \epsilon - \frac{2 \radius \l^\infty}{1-\gamma}
      }{
        \WUL - \frac{2 \radius \l^\infty}{1-\gamma}
      }
}. $
\end{restatable}

As in the Dec-POMDP case, the length of trajectories required to
approximate a discounted criterion is non-exponential.

\subsection{Execution}
\label{sec|Execution}

As can be noted, any strategy $\delta^2_\depth$ in a tuple $w$
guarantees at most (\ie, at worst from 2's viewpoint) expected return
$\occ^{m,1}_\depth \cdot \upb\nu^2_\depth$ if in the associated
$\occ_\depth$, whatever 1's strategy.
This holds in particular at $\depth=0$, where $\occ_0$ always
corresponds to the initial \os.
Thus, if 2 executes a strategy
$\delta^{2,\star}_0 \in \argmax_{\langle { \occ^{c,1}_0, \langle \upb\nu^2_0, \delta^2_0 \rangle }\rangle} \occ^{m,1}_0 \cdot \upb\nu_0 $,
then \zer{} expected return is at most $\upb{V}_0(\occ_0)$
($\leq V^*_0(\occ_0) + \epsilon$) (whatever 1's strategy).
Solving the derived zs-OMG therefore provides a solution strategy for
each player in the original zs-POSG, and each player can derive \zer{}
strategy on \zer{} own (no need for a coordinating central planner as
for Dec-POMDPs).
For instance, \Cshref{alg|zsPOSGwithLP+VWs+} returns a solution
strategy only for $1$.

\section{Experiments}
\label{sec|XPs}

The experiments aim at validating the proposed approach.
Additional results appear in \extCref{sec|Experiments}.

\subsection{Setup}

\paragraph{Benchmark Problems}

Four benchmark problems were used.
Mabc and Recycling Robot are well-known Dec-POMDP benchmark problems (\cf \url{http://masplan.org}) and were adapted to our competitive setting by making Player $2$ minimize (rather than maximize) the objective function.
Adversarial Tiger and Competitive Tiger were introduced by
\citet{Wiggers-msc15}.
We only consider finite horizons and $\gamma=1$.

\paragraph{Algorithms}
\label{sec|XP|algorithms}

\Cref{alg|zsPOSGwithLP+VWs+} is denoted \omgHSVIlccc, while \omgHSVIlc denotes a variant relying only on the Lipschitz
continuity (\cf \extCref{app|LipschitzOnly}), and used to highlight
the importance of exploiting the concavity and convexity properties.
They are compared against Sequence Form LP \citep{KolMegSte-geb96}, and \citeauthor{Wiggers-msc15}' two heuristic algorithms, {\em Informed}
and {\em Random} \citep{Wiggers-msc15}, which rely on the concavity
and convexity.\footnote{We use Wiggers' own (unreleased) Sequence Form
  LP solver, but could only copy the results for the two heuristic
  algorithms (based on a single run despite their randomization) from
  \citep{Wiggers-msc15}.}

\omgHSVIlc{} ran with an error $\epsilon$ specified in \Cref{tab|resExp}, and $\l_\depth = H \cdot ( r_{\max} - r_{\min} )$.
\omgHSVIlccc{} ran with an error $\epsilon = 0.01$, $\l_\depth = H \cdot ( r_{\max} - r_{\min} )$, $\rho$ the middle of its feasible interval, and the heuristic estimate for missing components of
$\upb\nu^2_\depth$ indicated in
\Cref{tab|resExp}.\footnote{The inadmissible heuristic $\upb\nu_{\text{init}}$ failed only on Adversarial Tiger.}
We also use FB-HSVI's LPE lossless compression of probabilistically equivalent \aoh{}s in \os{}s, so as to reduce their dimensionality \citep{DibAmaBufCha-jair16}.

Experiments ran on an Ubuntu machine with i7-10810U 1.10\,GHz Intel processor and 16\,GB available RAM.
We intend to make the code available within coming months under MIT license.

\subsection{Results}
A first observation is that both \omgHSVIlccc{} and \omgHSVIlc{}
maintain valid lower and upper bounds of the optimal value at
$\occ_0$, and reduce the gap progressively (\cf figures in
\extCshref{sec|Experiments}).
\Cref{tab|resExp} shows that (i) \omgHSVIlccc{} is always better than \omgHSVIlc{} and \citeauthor{Wiggers-msc15}' \citep{Wiggers-msc15} algorithms, and (ii) unless running out of memory, Sequence Form LP always outperforms \omgHSVIlccc{}. 
However, the LPE compression allows \omgHSVIlccc{} to exploit some
games' structure and thus generate trajectories even for large
horizons (\eg, Recycling Robot for $H=6$, \cf \extCshref{sec|Experiments}).
More generally, we observe that the number of
iterations performed by \omgHSVIlccc{} in 24\,\si{\hour} is highly correlated to the quality of the LPE
compression (Recycling Robot compresses the most and Adversarial Tiger
the least).
As expected, \omgHSVIlc{} turns out to be very slow, not terminating even its first iteration in most cases.

\begin{table}
  \def\tunits{time/gap}
\caption{Experiments comparing 4 solvers on various benchmark problems.
Reported values are the running times, or [gap] values [$\upb{V}_0(\occ_0)-\lob{V}_0(\occ_0)$] if the 24\,\si{\hour} timeout limit is reached.
{\small 
    ``(ni)'' indicates no improvement over the initialization after 24\,\si{\hour}.
``xx'' indicates an out-of-memory error.
``n/a'' indicates an unavailable result.
    }
}
\centering
    \label{tab|resExp}
{\begin{tabular}{lllll}
        \toprule
\multirow{1}{*}{Adversarial Tiger} & H=2 & H=3 & H=4 & H=5 \\
        \midrule
        Wiggers Random & [0.04]  & [0.38]  & [0.92] & [2.07]\\
        Wiggers Informed & [0.59]  & [1.32]  & [1.79] & [3.34]\\
        \omgHSVIlc \scriptsize{(0.1)} & 22\ \si{\min} & (ni)  & (ni) & (ni) \\
        \omgHSVIlccc \scriptsize{($\upb\nu_{\text{bMDP}}$)} & 1\ \si{\second} & 44\ \si{\second} & [1.79] & [2.27] \\
        Sequence Form LP & 0.02\ \si{\second} & 0.17\ \si{\second} & 3\ \si{\second} & 107\ \si{\second}\\
        \midrule
\multirow{1}{*}{Competitive Tiger} &  H=2 &  H=3 &  H=4 &  H=5 \\
        \midrule
        Wiggers Random & [0.56]  & [2.67]  & [5.81] &  [6.97] \\
        Wiggers Informed & [2.07]  & [2.33]  & [3.61] &  xx \\
        \omgHSVIlc \scriptsize{$(1.0)$} & [2.8] & (ni) & (ni) & (ni) \\
        \omgHSVIlccc \scriptsize{($\upb\nu_{\text{init}}$)} & 6\ \si{\second}  & [0.04] &  [2.30] & [4.92] \\
        Sequence Form LP & 0.14\ \si{\second} & 48\ \si{\second}  & 14\ \si{\min} & 2.5\ \si{\hour} \\
        \midrule
\multirow{1}{*}{Mabc} & H=2 & H=3 & H=4  & H=5 \\
        \midrule
\omgHSVIlc \scriptsize{$(0.05)$} & 2\ \si{\second}  & (ni)  & (ni) & (ni) \\
        \omgHSVIlccc \scriptsize{($\upb\nu_{\text{init}}$)} & 1\ \si{\second} &  34\ \si{\s} & [0.05] & [0.44]\\
        Sequence Form LP & 0.1\ \si{\second} & 1\ \si{\second} & 3\ \si{\second} & 181\ \si{\second} \\
        \midrule
\multirow{1}{*}{Recycling Robot} & H=3 & H=4 & H=5 & H=6 \\
        \midrule
\omgHSVIlc \scriptsize{$(0.2)$} &(ni) & (ni) & (ni) & (ni) \\
        \omgHSVIlccc \scriptsize{($\upb\nu_{\text{init}}$)} & 3\ \si{\min} & 12\ \si{\hour} & [0.77] & [2.65]\\
        Sequence Form LP & 1\ \si{\second} & 10\ \si{\second} & 1.5\ \si{\hour} & xx \\
\bottomrule
    \end{tabular}
} \end{table}

\poubelle{
\begin{figure}
\centering
\caption{Pruning does some work}
\resizebox{\linewidth}{!}{\includegraphics{figures/TailleSac.png}
}
\end{figure}
}

\section{Discussion}
\label{sec|discussion}

Inspired by state-of-the-art solution techniques for POMDPs,
Dec-POMDPs, and subclasses of zs-POSGs, we solve here zs-POSGs by turning them into zero-sum occupancy Markov
games, \ie, a fully-observable game that allows exploiting Bellman's
principle of optimality.
We expand the concavity-convexity and Lipschitz-continuity properties
of $V^*$ and $Q^*$, and build on them to propose point-based bounding
value function approximations, along with efficient selection and update operators
based on linear programming.
This allows deriving a variant of HSVI that provably converges in
finite time to an $\epsilon$-optimal solution, providing (safe)
solution strategies in a recursive form as a by-product of the solving
process.
Experiments confirm the feasibility of this approach and show
improved results compared to related heuristics (also exploiting the
concavity and convexity).

This approach paves the way for a large family of solvers as many
variants could be envisioned, \eg, using different algorithmic
schemes, approximations, selection and update operators, or pruning
techniques.
For instance, we also evaluated a variant relying only on $V^*$'s
Lipschitz-continuity.

Future work includes:
looking for better initializations, \eg, with
more advanced POMDP initializations or building on One-Sided zsPOSGs, and better Lipschitz constants, possibly through an incremental search;
proposing a pruning method for $\upb{W}^1_\depth$ and
$\lob{W}^2_\depth$;
exploiting oracle methods or other heuristics to solve local games
faster;
exploiting TPE rather than LPE compression;
and
branching on public observations (or even public information revealed by
the occupancy state's structure).

\message{>>> Leaving \currfilename >>>}

\paragraph{Acknowledgements}
  Let us thank Abdallah Saffidine, Vincent Thomas, and anonymous
  reviewers for fruitful discussions and comments that helped improve
  this work.
  
  This work was supported by the French National Research Agency through the 
  ``Planning and Learning to Act in Systems of Multiple Agents'' Project under Grant 19-CE23-0018-01.
[\url{http://perso.citi-lab.fr/jdibangoy/\#/plasma}]

\wlog{References start on page \thepage}

\bibliographystyle{plainnat}

\newpage

\ifextended{
  \newpage
  \onecolumn
  \appendix

\message{<<< Entering \currfilename <<<}

\section{Synthetic Tables}

For convenience, we provide two synthetic tables: 
\Cref{tab|PropertyTable} to sum up various theoretical properties
that are stated in this paper (assuming a finite temporal horizon), and 
\Cref{tab|NotationTable} to sum up the notations used in this
paper.

More precisely, \Cref{tab|PropertyTable} indicates, for various
functions $f$ and variables $x$, properties that $f$ is known to exhibit with
respect to $x$.
We denote by
\begin{description}[leftmargin=!,labelwidth=\widthof{$PWLCv$}]
\item[-] a function with no known (or used) property (see also comment
  below);
\item[{\sc n/a}] a non-applicable case;
\item[$\Lin$] a linear function;
\item[$LC$] a Lipschitz-continuous function;
\item[$Cv$] (resp. $Cc$) a convex (resp. concave) function;
\item[$PWLCv$] (resp. $PWLCc$) a piecewise linear and convex (resp. concave) function;
\item[$\indep$] the function being independent of the variable;
\item[$\neg P$] the negation of some property $P$ (\ie, $P$ is known not to hold).
\end{description}
Note also that, as
$\occ_\depth = \occ_\depth^{c,1} \occ_\depth^{m,1}$, the linearity or
Lipschitz-continuity properties of any function w.r.t. $\occ_\depth$
extends to both $\occ_\depth^{c,1}$ and $\occ_\depth^{m,1}$.
Reciprocally, related negative results extend from $\occ_\depth^{c,1}$ or $\occ_\depth^{m,1}$ to $\occ_\depth$.
In these three columns, we just indicate results that cannot be derived from one of the two other columns.

\begin{table}
    \caption{Known properties of various functions appearing in this work} \centering
    \label{tab|PropertyTable}
    \resizebox{\linewidth}{!}{\begin{tabular}{lr@{ }lr@{ }lr@{ }lr@{ }lr@{ }l}
      \toprule
& \multicolumn{2}{c}{$\occ_\depth$}
      & \multicolumn{2}{c}{$\occ_\depth^{m,1}$}
      & \multicolumn{2}{c}{$\occ_\depth^{c,1}$}
      & \multicolumn{2}{c}{$\beta_\depth^i$}
      & \multicolumn{2}{c}{$\beta_\depth^{-i}$} \\
      \midrule
$T(\occ_\depth,\vbeta_\depth)$
      & $\Lin$ & {\scriptsize (\cshrefpage{lem|occSufficient})}
      & \multicolumn{2}{c}{-}
      & \multicolumn{2}{c}{-}
      & $\Lin$ & {\scriptsize (\cshrefpage{lem|occSufficient})}
               & $\Lin$ & {\scriptsize (\cshrefpage{lem|occSufficient})} \\
$T^{m,i}(\occ_\depth,\vbeta_\depth)$
      & $\Lin$ & {\scriptsize (\cshrefpage{lem|T1mlin})}
      & \multicolumn{2}{c}{-}  & \multicolumn{2}{c}{-}
      & $Lin$ & {\scriptsize (\cshrefpage{lem|T1mlin})}
               & $Lin$ & {\scriptsize (\cshrefpage{lem|T1mlin})} \\
$T^{c,i}(\occ_\depth,\vbeta_\depth)$
      & \multicolumn{2}{c}{-}  & $\indep$ & {\scriptsize (\cshrefpage{lem|T1cindep})}
& $\neg LC$ & {\scriptsize (\cshrefpage{lemma|NotLipschitzOccC})} 
      & $\indep$ & {\scriptsize (\cshrefpage{lem|T1cindep})}
      & $\neg LC$ & {\scriptsize (\cshrefpage{lemma|NotLipschitzBeta})} \\
$V_\depth^{*}(\occ_\depth)$
      & $LC$ & {\scriptsize (\cshrefpage{cor|V|LC|occ})}
      & $PWLCv$ & {\scriptsize (\cshrefpage{theo|ConvexConcaveV})}
      & \multicolumn{2}{c}{-}
      & \multicolumn{2}{c}{\sc n/a} & \multicolumn{2}{c}{\sc n/a} \\
$Q_\depth^{*}(\occ_\depth,\vbeta_\depth)$
      & $LC$ & {\scriptsize (from $V^*_{\depth+1}$ LC)}
      & \multicolumn{2}{c}{-} & \multicolumn{2}{c}{-}
      & $\neg Lin$ & {\scriptsize (\cshrefpage{prop|localGameNotBiLinear})}
               & $\neg Lin$ & {\scriptsize (\cshrefpage{prop|localGameNotBiLinear})} \\
      & & & & & &
               & $Cc$ & {\scriptsize (\cshrefpage{lem|Qcc})}
      & $Cv$ & {\scriptsize (\cshrefpage{lem|Qcc})} \\
$W_\depth^{i,*}(\occ_\depth,\beta_\depth^i)$
      & $LC$ & {\scriptsize (from $Q^*_{\depth+1}$ LC)} & \multicolumn{2}{c}{-} & \multicolumn{2}{c}{-}
      & $\neg Lin$ & {\scriptsize (from $Q^*$ $\neg Lin$)} 
               & \multicolumn{2}{c}{\sc n/a} \\
      & & & & & &
               &  $Cc$ & {\scriptsize (\cshrefpage{lemma|Wconcave})}
      & \\
$\nu^2_{[\occ^{c,1}_\depth,\beta^2_{\depth:}]}$
      & \multicolumn{2}{c}{{\sc n/a}} & \multicolumn{2}{c}{{\sc n/a}}
      & $LC$ & {\scriptsize (\cshrefpage{lem|nuLC})} & \multicolumn{2}{c}{\sc n/a}
      & \multicolumn{2}{c}{-}  \\
\bottomrule
\end{tabular}
}
\end{table}

\newcommand{\multicolumnTWO}[1]{\multicolumn{2}{@{}>{\hsize=\dimexpr2\hsize+4\tabcolsep+2\arrayrulewidth\relax}X}{#1}}

\begin{table}[htbp]
  \caption{Various notations used in this work} \label{tab|NotationTable}
\centering \begin{tabularx}{1.\textwidth}{r@{ }c@{ }l@{ }X}
    \toprule
$-i$ & $\eqdef$ &
    \multicolumnTWO{$i$'s opponent. Thus: $-1=2$, and $-2=1$.}
    \medskip \\
    \multicolumn{4}{c}{\underline{Histories and occupancy states}} \medskip \\
$\theta^i_\depth$ & $\eqdef$ & 
    \multicolumnTWO{
      $(a^i_1, z^i_1, \dots , a^i_\depth, z^i_\depth) $ ($\in \Theta^i = \cup_{t=0}^{H-1} \Theta^i_t$) is a length-$\depth$ {\em action-observation history} (\aoh) for
      \player $i$.
    }
    \\
    $\vth_\depth$ & $\eqdef$ & 
    \multicolumnTWO{
      $(\theta^1_\depth,\theta^2_\depth)$ ($\in \vTh = \cup_{t=0}^{H-1} \vTh_t$) is a {\em joint \aoh} at $\depth$.
    }
    \\
    $\occ_\depth(\vth_\depth)$  & $\eqdef$ & 
    \multicolumnTWO{
      {\em Occupancy state} (\os) $\occ_\depth$ ($\in \Occ = \cup_{t=0}^{H-1} \Occ_t$, where $\Occ_\depth \eqdef \Delta(\vTh_\depth)$), \ie, probability distribution over joint \aoh{}s $\vth_\depth$
      (typically for some applied $\vbeta_{0:\depth-1}$).
    }
    \\
    $\occ_\depth^{m,i}(\theta^i_\depth)$ & $\eqdef$ & 
    \multicolumnTWO{
      {\em Marginal term} of $\occ_\depth$ from player $i$'s point of view ($\occ_\depth^{m,i} \in \Delta(\Theta_\depth^i) $).
    }
    \\
    $\occ_\depth^{c,i}(\theta^{-i}_\depth | \theta^i_\depth)$ & $\eqdef$ & 
    \multicolumnTWO{
      {\em Conditional term} of $\occ_\depth$ from $i$'s point of view ($\occ_\depth^{c,i} : \Theta_\depth^i \mapsto \Delta(\Theta_\depth^{-i}) $).
    }
    \\
    $b(s | \vth_\depth)$ & $\eqdef$ & 
    \multicolumnTWO{
      {\em Belief state}, \ie, probability distribution over states given a joint \aoh ($b(s | \vth_\depth) : \cS \times \vTh_\depth \mapsto \reals $). Can be computed by an HMM filtering process.
    }
    \\
    $o_\depth$ & $\eqdef$ &  
    \multicolumnTWO{
      {\em Full occupancy state} $o_\depth$ ($\in \Delta(\cS \times \vTh_\depth) $), \ie,
$Pr(s,\vth_\depth)$ for the current $\vbeta_{0:\depth-1}$,
and thus verifies $\occ_\depth(\vth_\depth)=\sum_{s\in\cS} o_\depth(s,\vth_\depth)$. 
Is used in the implementation to simplify computations (\eg, of $r_t$ and $\occ_{\depth+1}$ through $b$).
    }
\medskip \\
    \multicolumn{4}{c}{\underline{Decision rules and strategies}} \medskip \\
    $\beta^i_\depth$ & $\eqdef$ & 
    \multicolumnTWO{
      A {\em (behavioral) decision rule}
      (\dr) at time $\depth$ for \player $i$ is a mapping
      $\beta^i_\depth$ from private \aoh{}s in $\Theta^i_\depth$ to {\em
        distributions} over private actions.
We note $\beta^i_\depth(\theta^i_\depth,a^i)$ the probability to
      pick $a^i$ when facing $\theta^i_\depth$.
    }
    \\
    $\vbeta_\depth$ & $\eqdef$ &
    \multicolumnTWO{
      $\langle \beta^1_\depth, \beta^2_\depth \rangle$
      ($\in \cB = \cup_{t=0}^{H-1} \cB_t$) is a {\em decision rule
        profile}.
    }
    \\
    $\beta^i_{\depth:\depth'}$ & $\eqdef$ &
    \multicolumnTWO{
      $(\beta^i_\depth, \dots, \beta^i_{\depth'})$ is a {\em behavioral
        strategy} for \player $i$ from time step $\depth$ to $\depth'$
      (included).
    }
    \\
    $\vbeta_{\depth:\depth'}$ & $\eqdef$ &
    \multicolumnTWO{
      $ \langle \beta^1_{\depth:\depth'}, \beta^2_{\depth:\depth'}
      \rangle$ is a {\em behavioral strategy profile}.
    }
    \medskip \\
    \multicolumn{4}{c}{\underline{Rewards and value functions}} \medskip \\
$r_{\max}$ & $\eqdef$ & $\max_{s,\va}r(s,\va)$ & Maximum possible reward. \\
    $r_{\min}$ & $\eqdef$ & $\min_{s,\va}r(s,\va)$ & Minimum possible reward. \\
$V_\depth(\occ_\depth,\vbeta_{\depth:})$  & $\eqdef$ & 
$ E[\sum_{t=\depth}^{H-1} \gamma^t R_t \mid \occ_\depth, \vbeta_{\depth:}] $,
    &
    {\em Value} of  $\vbeta_{\depth:H-1}$ in
    \os $\occ_\depth$.
    \\
    & &
    \multicolumnTWO{
      where $R_t$ is the random var. for the reward at $t$.}
\\
    $V_\depth^*(\occ_\depth)$ & $\eqdef$ & $\max_{\beta^1_{\depth:}} \min_{\beta^2_{\depth:}} V_\depth(\occ_\depth, \vbeta_{\depth:})$ & \text{ \it Optimal value function} \\
$Q_\depth^*(\occ_\depth,\vbeta_\depth)$ & $\eqdef$  &
    $ r(\occ_\depth,\vbeta_\depth) + \gamma V_{\depth+1}^{*}(T(\occ_\depth,\vbeta_\depth))$
    & \text{ \it Opt. (joint) action-value fct.} \\
    $W_\depth^{i,*}(\occ_\depth,\beta_\depth^i)$ & $\eqdef$ 
    & $opt_{\beta^{-i}_\depth} Q_\depth^*(\occ_\depth,\vbeta_\depth)$,
    & \text{ \it Opt. (individual) action-value fct.} \\
    & &   \multicolumn{2}{>{\hsize=\dimexpr2\hsize+2\tabcolsep+\arrayrulewidth\relax}X}{
      where $opt=\max$ if $i=1$, $\min$ otherwise.
    }\\
    $\nu^2_{[\occ_\depth^{c,1}, \beta_{\depth:}^2]}$ & $\eqdef$ & 
    \multicolumnTWO{
      Vector of values (one component per \aoh{} $\theta^1_\depth$) for $1$'s best response to $\beta_{\depth:}^2$ assuming $\occ_\depth^{c,1}$.
This solution of a POMDP allows computing $V_\depth^{*}$ (see \Cref{theo|ConvexConcaveV}).
    }
    \medskip \\
    \multicolumn{4}{c}{\underline{Approximations}} \medskip \\
$\upb{V}_\depth(\occ_\depth)$ & $\eqdef$ & 
    \multicolumnTWO{
      Upper bound approximation of $V^*_\depth(\occ_\depth)$; relies on data set $\upb{bagV}_\depth$. 
    }
    \\
    $\lob{V}_\depth(\occ_\depth)$ & $\eqdef$ & 
    \multicolumnTWO{
      Lower bound approximation of $V^*_\depth(\occ_\depth)$; relies on data set $\lob{bagV}_\depth$. 
    }
    \\
    $\upb{W}^1_\depth(\occ_\depth,\beta^1_\depth)$ & $\eqdef$ & 
    \multicolumnTWO{
      Upper bound approximation of $W_\depth^{*,1}(\occ_\depth,\beta^1_\depth)$; relies on data set $\upb{bagW}^1_\depth$. 
    }
    \\
    $\lob{W}^2_\depth(\occ_\depth,\beta^2_\depth)$ & $\eqdef$ & 
    \multicolumnTWO{
      Lower bound approximation of $W_\depth^{*,2}(\occ_\depth,\beta^2_\depth)$; relies on data set $\lob{bagW}^2_\depth$. 
    }
    \\
    $\upb\nu_{\depth}^2$ & $\eqdef$ & 
    \multicolumnTWO{
      Vector (with one component per \aoh{} $\theta^1_\depth$) used in $\upb{V}_\depth$ and $\upb{W}^1_{\depth-1}$ (if $\depth\geq 1$).
    }
\medskip \\
    \multicolumn{4}{c}{\underline{Miscellaneous}} \medskip \\
$w_\depth$ & $\eqdef$ & 
    \multicolumnTWO{
      Denotes a triplet $\langle \occ^{c,1}_{\depth-1}, \beta^1_{\depth-1}, \langle \upb\nu^2_\depth,
    \delta^2_\depth \rangle \rangle \in \upb{bagW}^1_\depth$ (or a triplet in $\lob{bagW}^2_\depth$).
  }
  \\
    $\delta_{\depth}^2$ & $\eqdef$ & 
    \multicolumnTWO{
      Distribution over triplets $w_{\depth+1} \in \upb{bagW}^1_{\depth+1}$ (inducing a recursively defined strategy from $\depth$ to $H-1$).
Often denotes the strategy it induces. 
    }
    \\
    $x^{\top}$ & $\eqdef$ & 
    \multicolumnTWO{
      The transpose of a (usually column) vector $x$ of $\reals^n$.
    }
    \\ $c[y]$ & $\eqdef$ & 
    \multicolumnTWO{
      Denotes field $c$ of object/tuple $y$.
    }
    \\
    $\supp(d)$ & $\eqdef$ & 
    \multicolumnTWO{
      Support of distribution $d$, \ie, set of its non-zero probability elements.
    }
    \\
    \bottomrule
  \end{tabularx}
\end{table}

\section{Background}

\subsection{Re-casting POSGs as Occupancy Markov Games}
\label{app|fromTo}

The following result shows that the occupancy state is (i) Markovian, \ie, its value at $\depth$ only depends on its previous
value $\occ_{\depth-1}$, the system dynamics
$\PP{s}{a^1,a^2}{s'}{z^1,z^2}$, and the last behavioral decision rules
$\beta^1_{\depth-1}$ and $\beta^2_{\depth-1}$, and (ii) sufficient to estimate the expected reward.
Note that it holds for general-sum POSGs with any number of agents,
and as many reward functions; similar results have already been established, \eg, for Dec-POMDPs
(\cf \citep[Theorem~1]{DibAmaBufCha-jair16}).

\lemOccSufficient*

\begin{proof}
  \label{proof|lem|occSufficient}
  Let us first derive a recursive way of computing
  $ \occ_{ \vbeta_{0:\depth} }( \vth_\depth, \va_\depth,
  \vz_{\depth+1}) $:
  \begin{align*}
    & \occ_{ \vbeta_{0:\depth} }( \vth_\depth, \va_\depth, \vz_{\depth+1}) \eqdef Pr( \vth_\depth, \va_\depth, \vz_{\depth+1} \mid \vbeta_{0:\depth} ) \\
    & = \sum_{s_\depth, s_{\depth+1}}
    Pr( \vth_\depth, \va_\depth, \vz_{\depth+1}, s_\depth, s_{\depth+1} \mid \vbeta_{0:\depth} )
    \\
    & = \sum_{s_\depth, s_{\depth+1}}
    Pr( \vz_{\depth+1}, s_{\depth+1} \mid \vth_\depth, \va_\depth, s_\depth, \vbeta_{0:\depth} ) 
    Pr( \va_\depth \mid  \vth_\depth, s_\depth, \vbeta_{0:\depth} ) 
    Pr( s_\depth \mid \vth_\depth, \vbeta_{0:\depth} ) 
    Pr( \vth_\depth \mid \vbeta_{0:\depth} ) 
    \\
    & = \sum_{s_\depth, s_{\depth+1}}
    \underbrace{ Pr( \vz_{\depth+1}, s_{\depth+1} \mid \va_\depth, s_\depth ) }_{= \PP{s_\depth}{\va_\depth}{s_{\depth+1}}{\vz_{\depth+1}}}
    \underbrace{ Pr( \va_\depth \mid  \vth_\depth, \vbeta_{\depth} ) }_{= \vbeta(\vth_\depth, \va_\depth)}
    \underbrace{ Pr( s_\depth \mid \vth_\depth, \vbeta_{0:\depth} ) }_{= b(s_\depth \mid \vth_\depth)}
    \underbrace{ Pr( \vth_\depth \mid \vbeta_{0:\depth-1} ) }_{= \occ_{\vbeta_{0:\depth-1}}(\vth_\depth)},
    \intertext{(where $b(s \mid \vth_\depth)$ is the belief over states obtained by a usual HMM filtering process)}
    & = \sum_{s_\depth, s_{\depth+1}}
    \PP{s_\depth}{\va_\depth}{s_{\depth+1}}{\vz_{\depth+1}}
    \vbeta(\vth_\depth, \va_\depth)
    b(s_\depth \mid \vth_\depth)
    \occ_{\vbeta_{0:\depth-1}}(\vth_\depth).
  \end{align*}
  $\occ_{ \vbeta_{0:\depth} }$ can thus be computed from
  $\occ_{\vbeta_{0:\depth-1}}$ and $\vbeta_\depth$ without explicitly
  using $\vbeta_{0:\depth-1}$ or earlier occupancy states.

  Then, let us compute the expected reward at $\depth$ given
  $\vbeta_{0:\depth}$:
  \begin{align*}
     E[r(S_\depth,A^1_\depth,A^2_\depth) \mid \vbeta_{0:\depth} ] & = \sum_{s_\depth, \va_\depth} r(s_\depth, \va_\depth) Pr( s_\depth, \va_\depth \mid \vbeta_{0:\depth} )
    \\
    & = \sum_{s_\depth, \va_\depth} \sum_{\vth_\depth} r(s_\depth, \va_\depth) Pr( s_\depth, \va_\depth, \vth_\depth \mid \vbeta_{0:\depth} )
    \\
    & = \sum_{s_\depth, \va_\depth} \sum_{\vth_\depth} r(s_\depth, \va_\depth)
    Pr( s_\depth, \va_\depth \mid \vth_\depth, \vbeta_{0:\depth} )
    Pr( \vth_\depth \mid \vbeta_{0:\depth} )
    \\
    & = \sum_{s_\depth, \va_\depth} \sum_{\vth_\depth} r(s_\depth, \va_\depth)
    \underbrace{ Pr( \va_\depth \mid \vth_\depth, \vbeta_{0:\depth} ) }_{ \vbeta_\depth(\vth_\depth, \va_\depth) }
    \underbrace{ Pr( s_\depth \mid \vth_\depth, \vbeta_{0:\depth} ) }_{ b(s_\depth \mid \vth_\depth) }
    \underbrace{ Pr( \vth_\depth \mid \vbeta_{0:\depth} ) }_{\occ_{\vbeta_{0:\depth-1}}( \vth_\depth ) }
    \\
    & = \sum_{s_\depth, \va_\depth} \sum_{\vth_\depth} r(s_\depth, \va_\depth)
     \vbeta_\depth(\vth_\depth, \va_\depth)
     b(s_\depth \mid \vth_\depth) 
     \occ_{\vbeta_{0:\depth-1}}( \vth_\depth ).
  \end{align*}  
  The expected reward at $\depth$ can thus be computed from
  $\occ_{\vbeta_{0:\depth-1}}$ and $\vbeta_\depth$ without explicitly
  using $\vbeta_{0:\depth-1}$ or earlier occupancy states.
\end{proof}

\subsection{Introducing Local Games}

The first two lemmas below present properties of $\nxt^1_m$ and $\nxt^1_c$
that will be useful to demonstrate concavity and convexity properties of $Q^*$.

\begin{lemma}
  \labelT{lem|T1mlin}
  $T^1_m(\occ_\depth, \vbeta_\depth)$ is linear in $\occ_\depth$, $\beta^1_\depth$, and $\beta^2_\depth$.
\end{lemma}

\begin{proof}
  \label{proof|lem|T1mlin}
  \begin{align}
    \nxt^1_m(\occ_\depth,\vbeta_\depth)(\theta^1_\depth,a^1,z^1) & = \sum_{\theta^2_\depth, a^2, z^2} T(\occ_\depth, \vbeta_\depth)({
      ( \theta^1_\depth, a^1, z^1 ),
      ( \theta^2_\depth, a^2, z^2 )
    })
    & \text{(from \Cshref{eq|transition})}
    \nonumber
    \\
    & = \sum_{s',\theta^2_\depth,a^2,z^2} \beta^1_\depth(\theta^1_\depth,a^1) \beta^2_\depth(\theta^2_\depth,a^2)  \sum_{s} P^{z^1,z^2}_{a^1,a^2}(s'|s) 
    b(s|\theta^1_\depth,\theta^2_\depth) \occ_\depth (\theta^1_\depth,\theta^2_\depth)
    \nonumber
    \\
    & = \beta^1_\depth(\theta^1_\depth,a^1) \sum_{\theta^2_\depth,a^2} \beta^2_\depth(\theta^2_\depth,a^2)  \sum_{s,s',z^2} P^{z^1,z^2}_{a^1,a^2}(s'|s) 
    b(s|\theta^1_\depth,\theta^2_\depth) \occ_\depth (\theta^1_\depth,\theta^2_\depth).
    \label{eq|occm1}
\end{align}
\end{proof}

\begin{lemma}
  \labelT{lem|T1cindep}
  $T^1_c(\occ_\depth, \vbeta_\depth)$ is independent of $\beta^1_\depth$ and $\occ^{m,1}_\depth$.
\end{lemma}

\begin{proof}
  \label{proof|lem|T1cindep}
  \begin{align*}
\hspace{3cm}
    & \hspace{-3cm} T^1_c(\occ_\depth, \vbeta_\depth)((\theta^2_\depth,a^2,z^2) | (\theta^1_\depth, a^1, z^1)) = \frac{
      T(\occ_\depth, \vbeta_\depth)( (\theta^1_\depth, a^1, z^1), (\theta^2_\depth,a^2,z^2) )
    }{
      \sum_{\theta^2_\depth, a^2, z^2} T(\occ_\depth, \vbeta_\depth)( (\theta^1_\depth, a^1, z^1), (\theta^2_\depth,a^2,z^2) )
    } \\
    & = \frac{
      \beta^1_\depth(\theta^1_\depth,a^1) \beta^2_\depth(\theta^2_\depth,a^2)  \sum_{s,s'} P^{z^1,z^2}_{a^1,a^2}(s'|s) b(s | \theta^1_\depth, \theta^2_\depth) \occ_\depth (\theta^1_\depth, \theta^2_\depth)
    }{
      \beta^1_\depth(\theta^1_\depth,a^1) \sum_{\theta^2,a^2} \beta^2_\depth(\theta^2_\depth,a^2)  \sum_{s,s',z^2} P^{z^1,z^2}_{a^1,a^2}(s'|s) b(s | \theta^1_\depth, \theta^2_\depth) \occ_\depth (\theta^1_\depth,\theta^2_\depth)
    } \\
    & = \frac{
      \beta^2_\depth(\theta^2_\depth,a^2)  \sum_{s,s'} P^{z^1,z^2}_{a^1,a^2}(s'|s) b(s | \theta^1_\depth, \theta^2_\depth) \occ_\depth (\theta^1_\depth,\theta^2_\depth)
    }{
      \sum_{\theta^2,a^2} \beta^2_\depth(\theta^2_\depth,a^2)  \sum_{s,s',z^2} P^{z^1,z^2}_{a^1,a^2}(s'|s) b(s | \theta^1_\depth, \theta^2_\depth) \occ_\depth (\theta^1_\depth,\theta^2_\depth)
    }
    \\
& = \frac{
      \beta^2_\depth(\theta^2_\depth,a^2)  \sum_{s,s'} P^{z^1,z^2}_{a^1,a^2}(s'|s) b(s | \theta^1_\depth, \theta^2_\depth)
      \overbrace{ \occ^{c,1}_\depth(\theta^2_\depth | \theta^1_\depth) \occ^{m,1}_\depth(\theta^1_\depth) }
    }{
      \sum_{\theta^2,a^2} \beta^2_\depth(\theta^2_\depth,a^2)  \sum_{s,s',z^2} P^{z^1,z^2}_{a^1,a^2}(s'|s) b(s | \theta^1_\depth, \theta^2_\depth)
      \underbrace{ \occ^{c,1}_\depth(\theta^2_\depth | \theta^1_\depth) \occ^{m,1}_\depth(\theta^1_\depth) }
    } \\
    & = \frac{
      \left( \beta^2_\depth(\theta^2_\depth,a^2)  \sum_{s,s'} P^{z^1,z^2}_{a^1,a^2}(s'|s) b(s | \theta^1_\depth, \theta^2_\depth)
        \occ^{c,1}_\depth(\theta^2_\depth | \theta^1_\depth) \right) \occ^{m,1}_\depth(\theta^1_\depth)
    }{
      \left( \sum_{\theta^2,a^2} \beta^2_\depth(\theta^2_\depth,a^2)  \sum_{s,s',z^2} P^{z^1,z^2}_{a^1,a^2}(s'|s) b(s | \theta^1_\depth, \theta^2_\depth)
        \occ^{c,1}_\depth(\theta^2_\depth | \theta^1_\depth)  \right) \occ^{m,1}_\depth(\theta^1_\depth)
    } \\
    & = \frac{
      \beta^2_\depth(\theta^2_\depth,a^2)  \sum_{s,s'} P^{z^1,z^2}_{a^1,a^2}(s'|s) b(s | \theta^1_\depth, \theta^2_\depth)
        \occ^{c,1}_\depth(\theta^2_\depth | \theta^1_\depth)
    }{
      \sum_{\theta^2,a^2} \beta^2_\depth(\theta^2_\depth,a^2)  \sum_{s,s',z^2} P^{z^1,z^2}_{a^1,a^2}(s'|s) b(s | \theta^1_\depth, \theta^2_\depth)
      \occ^{c,1}_\depth(\theta^2_\depth | \theta^1_\depth)
    }.
    \label{eq|occc1} 
\end{align*}
\end{proof}

\begin{tcolorbox}[breakable, enhanced]
  \uline{Addendum:} The following complementary properties explain why seeking for better approximations is difficult.

  \begin{proposition}
    \labelT{lemma|NotLipschitzBeta}
    $\nxt^i_c(\occ_\depth,\vbeta_\depth)$ may be non continuous (thus
    non Lipschitz-continuous) w.r.t. $\beta_\depth^{\neg i}$.
\end{proposition}

  \begin{proof}
    \label{proof|lemma|NotLipschitzBeta}
    First, let us define
    \begin{align}
      & \begin{array}{r@{\,}c@{\,}c@{\,}c}
          f : & S_3(1) & \to & \reals  \\
              & (x,y,z) & \mapsto & \frac{\alpha x}{\alpha x + \beta y},
        \end{array}
      \end{align}
      where $(\alpha,\beta) \in (\mathbb{R}^{+,*})^2$ and $S_k(1)$ is the $k$-dimensional probability simplex.

      One can show that $f$ is not Lipschitz-Continuous.
Indeed, the sequences 
      \begin{align} 
        (u_n)_n & = \left( f\left(\frac{1}{n}, \frac{1}{n^2}, 1-\left(\frac{1}{n} + \frac{1}{n^2}\right)\right)\right)_n
        \text{ and} \\
        (v_n)_n & = \left( f\left(\frac{1}{n^2}, \frac{1}{n}, 1-\left(\frac{1}{n} + \frac{1}{n^2}\right)\right)\right)_n
      \end{align} 
      converge towards different values (respectively 1 and 0).
$f$ is thus not continuous around $(0,0,1)$, and therefore not Lipschitz continuous.
      
      This property extends to functions of the form $f(x, y_1, \dots, y_I, z_1, \dots, z_J) = \frac{\alpha x}{\alpha x + \sum_{i=1}^I \beta_i y_i}$ with \begin{itemize}
      \item $I,J\in \mathbb{N}^*$,
      \item $(x,y_1,\dots,y_I,z_1,\dots,z_J) \in S_{1+I+J}(1)$,
      \item positive scalars $\alpha$ and $\vbeta_i$
        ($i\in \{1, \dots, I\}$).
      \end{itemize}

      Note: In the following, we make plausible assumptions without
      providing a detailed example.
Let us now consider \Cshref{eq|occc1} for two tuples
      $\langle \theta^1_\depth, a^1, z^1 \rangle$ and
      $\langle \theta^2_\depth, a^2, z^2 \rangle$ such that
      $\occ_\depth(\theta^1_\depth,\theta^2_\depth)\neq 0$:
\begin{align}
        & \nxt^1_c(\occ_\depth,\vbeta_\depth)(\theta^2_{\tau},a^2,z^2|\theta^1_{\tau},a^1,z^1) \\
& \qquad =  \frac{
          \beta^2_\depth(\theta^2_\depth,a^2) \left[ \sum_{s,s'} P^{z^1,z^2}_{a^1,a^2}(s'|s) b(s | \theta^1_\depth, \theta^2_\depth) \right] \occ_\depth (\theta^1_\depth,\theta^2_\depth)
        }{
          \sum_{\hat\theta^2, \hat a^2} \beta^2_\depth(\hat \theta^2_\depth, \hat a^2) \left[ \sum_{s,s',\hat z^2} P^{z^1,\hat z^2}_{a^1, \hat a^2}(s'|s) b(s | \theta^1_\depth, \hat \theta^2_\depth) \right] \occ_\depth (\theta^1_\depth, \hat\theta^2_\depth)} \intertext{and assuming a simple case where $\occ^{c,1}(\theta^2_\depth|\theta^1_\depth)=1$ (\ie, all other \aoh{}s for $2$ being impossible):}
        & \qquad =  \frac{
          \beta^2_\depth(\theta^2_\depth,a^2)  \left[ \sum_{s,s'} P^{z^1,z^2}_{a^1,a^2}(s'|s) b(s | \theta^1_\depth, \theta^2_\depth) \right] \occ_\depth (\theta^1_\depth,\theta^2_\depth)
        }{
          \sum_{\hat a^2} \beta^2_\depth(\theta^2_\depth, \hat a^2) \left[ \sum_{s,s',\hat z^2} P^{z^1, \hat z^2}_{a^1,\hat a^2}(s'|s) b(s | \theta^1_\depth, \theta^2_\depth) \right] \occ_\depth (\theta^1_\depth,\theta^2_\depth)}.
      \end{align}
      Then, in cases where \begin{itemize}
      \item
        $\sum_{s,s'} P^{z^1,z^2}_{a^1,\tilde a^2}(s'|s) b(s |
        \theta^1_\depth, \theta^2_\depth) >0$ for action $a^2$, and \item
        $\sum_{s,s',\hat z^2} P^{z^1,\hat z^2}_{a^1,\tilde a^2}(s'|s) b(s |
        \theta^1_\depth, \theta^2_\depth) >0$
        for some, but not all, other actions $\tilde a^2$ ($=0$ typically when $z^1$ and $\tilde a^2$ are incompatible),
      \end{itemize}
      we recognize the above function
      $f(x, y_1, \dots, y_I, z_1, \dots, z_J)$, which is not
      continuous.

      Such situations where $\nxt^1_c$ is not continuous thus may
      indeed happen.
    \end{proof}

    \begin{proposition}
      \labelT{lemma|NotLipschitzOccC}
      $\nxt^i_c(\occ_\depth,\vbeta_\depth)$ may be non continuous
      (thus non Lipschitz-continuous) w.r.t. $\occ^{c,1}_\depth$.
    \end{proposition}

    A similar proof as for \Cref{lemma|NotLipschitzBeta}
    applies, the variables corresponding to parameters
    $\occ^{c,1}_\depth(\theta^2_\depth|\theta^1_\depth)$ under fixed
    $\theta^1_\depth$.

\end{tcolorbox}

This leads us to our main result here regarding $Q^*$.

\label{proofLemQcc}

\lemQcc*

\begin{proof}
  \label{proof|lem|Qcc}
  Let us rewrite $Q^*_\depth(\occ_\depth, \vbeta_\depth)$ to look at
  its properties with respect to $\beta^1_\depth$:
  \begin{align}
    Q^*_\depth(\occ_\depth, \vbeta_\depth)
    & = r(\occ_\depth, \vbeta_\depth) + \gamma V^*_{\depth+1}( \nxt(\occ_\depth,\vbeta_\depth) )
    \\
    & = r(\occ_\depth, \vbeta_\depth) + \gamma \min_{\beta_{\depth+1:}^2}\left[ \nxt^1_m(\occ_\depth,\vbeta_\depth) \cdot \nu^2_{[\nxt^1_c(\occ_\depth,\vbeta_\depth), \beta^2_{\depth+1:}]} \right]
    \\
    & = \overbrace{r(\occ_\depth, \beta^1_\depth, \beta^2_\depth)}^{
        \substack{
          \text{linear in $\beta^1_\depth$}\\
          \text{(proof of \Cref{lem|occSufficient})}
        }
      }
      + \gamma
      \underbrace{ \min_{\beta_{\depth+1:}^2} \Big[ 
        \overbrace{\nxt^1_m(\occ_\depth,\beta^1_\depth,\beta^2_\depth)}^{
          \substack{
            \text{linear in $\beta^1_\depth$}\\
            \text{(\Cref{lem|T1mlin})}
          }
        } \cdot 
        \overbrace{\nu^2_{[\nxt^1_c(\occ_\depth,\beta^2_\depth), \beta_{\depth+1:}^2]}}^{
          \substack{
            \text{independent of $\beta^1_\depth$}\\
            \text{(\Cref{lem|T1cindep})}
          }
        }
        \Big]
      }_{
        \substack{
          \text{concave in $\beta^1_\depth$}\\
          \text{(as a concave combination of linear functions)}
        }
      }.
  \end{align}
  Combining a reward that is linear in $\beta^1_\depth$ and a term
  that is concave in $\beta^1_\depth$ (as a concave combination of
  linear functions), $Q^*_\depth(\occ_\depth, \vbeta_\depth)$ is
  concave in $\beta^1_\depth$ and (symmetrically) convex in
  $\beta^2_\depth$.
\end{proof}

\begin{tcolorbox}[breakable, enhanced]
  
  \uline{Addendum:} The following complementary property highlights
  the difference in nature between usual finite normal-form games and
  the {\em local} games encountered in each occupancy state.
  
  \begin{proposition}
    \labelT{prop|localGameNotBiLinear}
    Local game $Q^*_t(\occ_t,\beta^1_t,\beta^2_t)$ may not be
    bi-linear (in $\beta^1_t$ and $\beta^2_t$).
  \end{proposition}

  \begin{proof}
    \label{proof|prop|localGameNotBiLinear}
    Let us consider a sequential version of matching pennies:
    \begin{itemize}
    \item $\cS=\{s_i,s_h,s_t\}$, the {\em initial}, {\em head}, and
      {\em tail} states;
    \item $\cA^i=\{a_h,a_t\}$ for any player $i$, the {\em head} and
      {\em tail} actions;
    \item $\cZ^i=\{z_\emptyset\}$ for any player $i$, both being blind;
    \item $\PP{s}{a^1,a^2}{s'}{z^1,z^2} = \mathbb{1}_{s'=a^1}$, the
      next state being the head (resp. tail) state if $1$'s action is
      the head (resp. tail) action;
    \item
      $r(s,a^1,a^2) = \mathbb{1}_{s=s_i} \cdot ( 1 - 2 \cdot
      \mathbb{1}_{s'=a^2} )$,
      so that player $1$ gets: (i) $0$ at $t=0$, and (ii) $+1$
      (resp. $-1$) if player $2$ has not guessed at $t$ \zer{}
      previous action (at $t-1$) at any other $t$;
    \item $\gamma=1$; $h=2$.
    \end{itemize}
    Player 1 thus has to pick head or tail first (at $t=0$) and 2
    second (at $t=1$), trying to guess 1's pick.
  
    Let us then parameterize $i$'s strategy by \zer{} probability
    $p_i\in [0,1]$ of picking action $a_h$ at $i$'s only actual
    decision point ($t=0$ for 1, and $t=1$ for 2).
Player 2's best response to some $p_1$ is for example to set
    \begin{align*}
      p_2
      & =
      \begin{cases}
        0 & \text{if } p_1 \leq 0.5, \text{ and} \\
        1 & \text{if } p_1 > 0.5.
      \end{cases}
    \end{align*}
    The value of the local game at $t=0$ can thus be written as the
    following function of $p_1$:
\begin{align*}
      Q^*_0(b_0, \beta^1_0, \beta^2_0)
      = Q^*_0(b_0, p_1)
      & = \underbrace{r(b_0, p_1)}_{=0} + \underbrace{\gamma}_{=1} \cdot \underbrace{V^*_1(T(b_0, p_1))}_{=2 \cdot |p_1-0.5|}
      = 2 \cdot |p_1-0.5|.
    \end{align*}
    $Q^*_0(b_0, \beta^1_0, \beta^2_0)$ is thus not linear in
    $\beta^1_0$, which concludes the proof.
  \end{proof}
\end{tcolorbox}

\section{Properties and Approximation of Optimal Value Functions} 

\subsection{Properties of $V^*$}

\subsubsection{Linearity and Lipschitz-continuity of $\nxt (\occ_\depth, \beta^1_\depth, \beta^2_\depth)$ }
\label{proofLemOccLin}

\lemOccLin*

\begin{proof}
  \label{proof|lem|occ|lin}
  Let $\occ$ be an occupancy state at time $\depth$ and
  $\vbeta_\depth$ be a decision rule.
Then, as seen in the proof of \Cref{lem|occSufficient}, the
  next occupancy state $\occ' = \nxt(\occ,\vbeta_\depth)$ satisfies,
  for any $s'$ and $(\vth,\va,\vz)$:
  \begin{align*}
    \occ'(\vth,\va,\vz)
    & \eqdef Pr(\vth,\va,\vz | \occ, \beta^1_\depth, \beta^2_\depth) \\
& = \beta^1_\depth(\theta^1, a^1) \beta^2_\depth(\theta^2, a^2) \left[ \sum_{s', s\in \cS} \PP{s}{\va}{s'}{\vz} b(s| \vth) \right] \occ(\vth) .
  \end{align*}
$b(s|\vth)$ depending only on the model (transition function and initial belief),
the next occupancy state $\occ'$ thus evolves linearly w.r.t.
(i) {\em private} decision rules $\beta^1_\depth$ and $\beta^2_\depth$, and
(ii) the occupancy state $\occ$.

  The $1$-Lipschitz-continuity holds because each component of vector
  $\occ_\depth$ is distributed over multiple components of $\occ'$.
Indeed, let us view two occupancy states as vectors
  $\vx,\vy \in \reals^n$, and their corresponding next states under
  $\vbeta_\depth$ as $M \vx$ and $M \vy$, where
  $M \in \reals^{m\times n}$ is the corresponding transition matrix
  (\ie, which turns $\occ$ into
  $\occ' \eqdef \nxt(\occ_\depth,\vbeta_\depth)$).
Then,
  \begin{align*}
    \norm{M\vx - M\vy}_{1}
& \eqdef \sum_{j=1}^m \ \abs{ \sum_{i=1}^n M_{i,j} (x_i - y_i) } \\
    & \leq \sum_{j=1}^m  \sum_{i=1}^n \abs{M_{i,j} (x_i - y_i) }
    & \text{(convexity of $\abs{\cdot}$)} \\
    & = \sum_{j=1}^m  \sum_{i=1}^n M_{i,j} \abs{ x_i - y_i }
    & \text{($\forall {i,j},\ M_{i,j}\geq 0$)} \\
    & = \sum_{i=1}^n \underbrace{\sum_{j=1}^m M_{i,j}}_{=1} \abs{ x_i - y_i } 
    & \text{($M$ is a transition matrix)} \\
& \eqdef \norm{\vx-\vy}_1.
& \qedhere
  \end{align*}
\end{proof}

\subsubsection{Lipschitz-Continuity of $V^*$}
\label{proofLemVLinOcc}
\label{proofCorVLCOcc}
\label{proofLemNuLC}

The next two results demonstrate that, in the finite horizon setting,
$V^*$ is Lipschitz-continuous (LC) in occupancy space, which allows
defining LC upper and lower bound approximations.

\lemVLinOcc*

Note: This result in fact applies to any reward function of a
general-sum POSG with any number of agents (here $N$), \eg, to a
Dec-POMDP.
The following proof handles the general case (with
$\vbeta_\depth \eqdef \langle \beta^1_\depth, \dots, \beta^N_\depth
\rangle$, and
$\vbeta_\depth(\va|\vth) = \prod_{i=1}^N
\beta^i_\depth(a^i,\theta^1)$).

\begin{proof}
  \label{proof|lem|V|lin|occ}
  This property trivially holds for $\depth=H-1$ because
  \begin{align*}
    V_{H-1}(\occ_{H-1},\vbeta_{H-1:}) 
    & = r(\occ_{H-1},\vbeta_{H-1}) \\
    & = \sum_{s,\va} \left( \sum_\vth Pr(s,\va|\vth) \occ_{H-1}(\vth) \right) r(s,\va) \\
    & = \sum_{s,\va} \left( \sum_\vth b(s|\vth) \vbeta_\depth(\va|\vth) \occ_{H-1}(\vth) \right) r(s,\va) \\
    & = \sum_{s,\vth}  b(s | \vth) \occ_{H-1}(\vth)  \left( \sum_{\va} \vbeta_\depth(\va|\vth) r(s,\va) \right).
  \end{align*}
Now, let us assume that the property holds for
  $\depth+1 \in \{1 \twodots H-1\}$.
Then,
\begin{align*}
    V_\depth(\occ_\depth,\vbeta_{\depth:}) 
    & = \sum_{s,\va} \Big( \sum_\vth b(s| \vth) \vbeta_\depth(\va|\vth) \occ_\depth(\vth)  \Big) r(s, \va)
    + \gamma V_{\depth+1}\left( \nxt (\occ_\depth, \vbeta_\depth), \vbeta_{\depth+1:} \right) \\
    & = \sum_{s,\vth}  b(s| \vth) \occ_\depth(\vth)  \Big( \sum_{\va} \vbeta_\depth(\va|\vth) r(s, \va) \Big)
    + \gamma V_{\depth+1}\left( \nxt (\occ_\depth, \vbeta_\depth), \vbeta_{\depth+1:} \right) .
  \end{align*}
  As
  \begin{itemize}
  \item $\nxt(\occ_\depth,\vbeta_\depth)$ is linear in $\occ_\depth$
    {\footnotesize (\Cref{lem|occ|lin})} and
  \item $V_{\depth+1}(\occ_{\depth+1}, \vbeta_{\depth+1:})$ is
    linear in $\occ_{\depth+1}$ {\footnotesize (induction hypothesis)},
  \end{itemize}
  their composition,
  $V_{\depth+1} ( \nxt (\occ_\depth,\vbeta_\depth),
  \vbeta_{\depth+1:} )$,
  is also linear in $\occ_\depth$, and so is
  $V_\depth(\occ_\depth,\vbeta_{\depth:})$.
\end{proof}

\corVLCOcc*

\begin{proof}
  \label{proof|cor|V|LC|occ}
  At depth $\depth$, the value of any behavioral strategy
  $\vbeta_{\depth:}$ is bounded, independently of $\occ_\depth$, by
  \begin{align*}
    V^{\max}_\depth & \eqdef \h{H}{\depth}{\gamma} r_{\max}, \quad \text{where } r_{\max} \eqdef \max_{s,\va}r(s,\va),
    \text{ and } \\
    V^{\min}_\depth &\eqdef \h{H}{\depth}{\gamma}  r_{\min}, \quad \text{where } r_{\min} \eqdef \min_{s,\va}r(s,\va).
  \end{align*}
Thus, $V_{\vbeta_{\depth:}}$ being a linear function defined over a
  probability simplex ($\Occ_\depth$) (\cf \Cref{lem|V|lin|occ}) and
  bounded by $[V^{\min}_\depth,V^{\max}_\depth]$, we can apply
  \citetg{Horak-phd19} Lemma~3.5 (p.~33) to establish that it is also
  $\lt{\depth}$-LC, \ie,
  \begin{align*}
    \abs{V_{\vbeta_{\depth:}}(\occ) -V_{\vbeta_{\depth:}}(\occ')}
    & \leq \lt{\depth} \norm{\occ-\occ'}_{\p} \quad (\forall \occ, \occ'), \\
    \text{with }
    \lt{\depth}
    & = \frac{V^{\max}_{\depth}-V^{\min}_{\depth}}{2}.
  \end{align*}

  Considering now optimal solutions, this means that, at depth
  $\depth$ and for any $(\occ,\occ') \in \Occ_\depth$:
  \begin{align*}
    V^*_\depth(\occ) - V^*_\depth(\occ') 
    & = \max_{\beta^1_{\depth:}} \min_{\beta^2_{\depth:}} V_\depth(\occ, \beta^1_{\depth:}, \beta^2_{\depth:}) 
    - \max_{\beta'^1_{\depth:}} \min_{\beta'^2_{\depth:}} V_\depth(\occ', \beta'^1_{\depth:}, \beta'^2_{\depth:}) \\
& \leq \max_{\beta^1_{\depth:}} \min_{\beta^2_{\depth:}} \left[ V_\depth(\occ', \beta^1_{\depth:}, \beta^2_{\depth:}) + \lt{\depth} \norm{\occ-\occ'}_{\p} \right] 
    - \max_{\beta'^1_{\depth:}} \min_{\beta'^2_{\depth:}} V_\depth(\occ', \beta'^1_{\depth:}, \beta'^2_{\depth:}) \\
& = \lt{\depth} \norm{\occ-\occ'}_{\p}.
  \end{align*}
  Symmetrically,
  $V^*_\depth(\occ) - V^*_\depth(\occ') \geq -\lt{\depth}
  \norm{\occ-\occ'}_{\p}$, hence the expected result:
  \begin{align*}
    \abs{ V^*_\depth(\occ) - V^*_\depth(\occ') } 
    & \leq \lt{\depth} \norm{\occ-\occ'}_{\p}.
\qedhere
  \end{align*}
\end{proof}

As it will be used later, let us also present the following lemma.

\begin{lemma}
  \labelT{lem|nuLC}
  Let us consider $\depth \in \{0 \twodots H-1\}$, $\theta^1_\depth$, and
  $\delta^2_\depth$.
Then
  $\nu^2_{[\occ^{c,1}_{\depth},
    \delta^2_{\depth}]}(\theta^1_{\depth})$ is $\l_\depth$-LC in
  $\occ^{c,1}_\depth(\cdot|\theta^1_\depth)$.
      
  Equivalently, we will also write that
  $\nu^2_{[\occ^{c,1}_{\depth}, \delta^2_{\depth}]}$ is $\l_\depth$-LC
  in $\occ^{c,1}_\depth$ in {\em vector-wise} 1-norm, \ie:
  \begin{align*}
    \vabs{
      \nu^2_{[\occ^{c,1}_{\depth}, \delta^2_{\depth}]} - \nu^2_{[\tilde\occ^{c,1}_{\depth}, \delta^2_{\depth}]}
    }_1
    & \vleq \lt\depth \vnorm{
      \occ^{c,1}_{\depth} - \tilde\occ^{c,1}_{\depth}
    }_1,
  \end{align*}
  where (i) the absolute value of a vector is obtained by taking the
  absolute value of each component; and (ii) the vector-wise 1-norm of a matrix is a vector made of the
  1-norm of each of its component vectors.
\end{lemma}

\begin{proof}
  \label{proof|lem|nuLC}
For any $\theta^1_\depth$, $\occ^{c,1}_{\depth}$ and $\delta^2_{\depth}$ induce a POMDP for Player $1$
  from $\depth$ on,
where (i) the state at any $t\in \{\depth \twodots H-1\}$ corresponds to a pair
  $\langle s, \theta^2_t \rangle$, and (ii) the initial belief is derived from $\occ^{c,1}_{\depth}(\cdot|\theta^1_t)$.
The belief state at $t$ thus gives:
  \begin{align*}
    b_{\theta^1_t}(s, \theta^2_t)
    & \eqdef Pr(s, \theta^2_t | \theta^1_t) = \underbrace{Pr(s | \theta^2_t, \theta^1_t)}_{b^{\hmm}_{\theta^2_t, \theta^1_t}(s)}
    \cdot \underbrace{Pr(\theta^2_t | \theta^1_t)}_{\occ^{c,1}_t(\theta^2_t|\theta^1_t)}.
  \end{align*}
So,
  \begin{itemize}
  \item the value function of any behavioral strategy
    $\beta^1_{\depth:}$ is linear at $t$ in $b_{\theta^1_t}$, thus (in
    particular) in $\occ^{c,1}_t(\cdot|\theta^1_t)$; and
  \item the optimal value function is LC at $t$ also in
    $b_{\theta^1_t}$ (with the same depth-dependent upper-bounding
    Lipschitz constant $\l_t$ as in \Cref{cor|V|LC|occ}),\footnote{The
      proof process is similar. The only difference lies in the space
      at hand, but without any impact on the resulting formulas.} thus (in particular) in $\occ^{c,1}_t(\cdot|\theta^1_t)$.
  \end{itemize}

  Using $t=\depth$, the optimal value function is
  $\nu^2_{[\occ^{c,1}_{\depth},
    \delta^2_{\depth}]}(\theta^1_{\depth})$, which is thus
  $\l_\depth$-LC in $\occ^{c,1}_{\depth}(\cdot|\theta^1_\depth)$.
\end{proof}

\subsection{Bounding Approximations of $V^*$, $W^{1,*}$ and $W^{2,*}$} \label{app|derivingApproximations}

\subsubsection{$\upb{V}_\depth$ and $\protect\lob{V}_\depth$}
\label{app|derivingApproximationsV}

To find a form that could be appropriate for an upper bound approximation of $V^*_\depth$,
let us consider an \os{} $\occ_\depth$ and a single tuple $\langle { \tilde\occ_\depth, \nu^2_{[\tilde\occ_\depth^{c,1}, \beta^2_{\depth:}]} } \rangle$,
and define
$\zeta_\depth \eqdef \occ_\depth^{m,1}\tilde\occ_\depth^{c,1}$.
Then,
\begin{align*}
  V^*(\occ_\depth)
  & \leq V^*(\zeta_\depth) + \lt{\depth}  \norm{ \occ_\depth - \zeta_\depth }_1 
  & \text{(LC, \cf \Cref{cor|V|LC|occ})} \\
  & = V^*(\occ_\depth^{m,1} \tilde\occ_\depth^{c,1})
  + \lt{\depth}  \norm{ \occ_\depth - \zeta_\depth }_1 \\
  & \label{eq|upbV}
  \leq \occ_\depth^{m,1} \cdot \nu^2_{[\tilde\occ_\depth^{c,1}, \beta^2_{\depth:}]}
  + \lt{\depth}  \norm{ \occ_\depth - \occ_\depth^{m,1}\tilde\occ_\depth^{c,1} }_1. 
  & \text{(Cvx, \cf \Cref{theo|ConvexConcaveV})} \end{align*}
Notes:
\begin{itemize}
\item $\tilde\occ^{m,1}_\depth$ does not appear in the resulting
  upper bound, thus will not need to be specified.
\item For $\depth=H-1$,
  $\nu^2_{[\tilde\occ_\depth^{c,1}, \beta_{\depth:}^2]}$ is a simple
  function of $r$, $\tilde\occ_\depth^{c,1}$, $\beta_{\depth:}^2$, and
  the dynamics of the system, as described in Eq.~(9) of
  \citet{WigOliRoi-corr16}.
\end{itemize}

From this, we can deduce the following appropriate forms of upper and
(symmetrically) lower bound function approximations for $V^*_\depth$:
\begin{align*}
  \upb{V}_\depth(\occ_\depth)
  & = \min_{   \langle \tilde\occ_\depth^{c,1}, \upb\nu^2_\depth  \rangle  \in \upb{bagV} } \left[ \occ_\depth^{m,1} \cdot \upb\nu^2_\depth + \lt{\depth} \norm{ \occ_\depth - \occ_\depth^{m,1}\tilde\occ_\depth^{c,1} }_1 \right], \text{ and} \\
\lob{V}_\depth(\occ_\depth) 
  & = \max_{ \langle \tilde\occ_\depth^{c,2}, \lob\nu^1_\depth \rangle \in \lob{bagV} } \left[ \occ_\depth^{m,2} \cdot \lob\nu^1_\depth - \lt{\depth} \norm{ \occ_\depth - \occ_\depth^{m,2}\tilde\occ_\depth^{c,2} }_1 \right],
\end{align*}
which are respectively concave in $\occ^{m,1}_\depth$ and convex in
$\occ^{m,2}_\depth$, and which both exploit the Lipschitz continuity.

\subsubsection{$\upb{W}^1_\depth$ and $\protect\lob{W}^2_\depth$}
\label{app|derivingApproximationsW}

Note: We discuss all depths from $0$ to $H-1$, even though we do not
need these approximations at $\depth=H-1$.

Let us first see how concavity-convexity properties affect $W_\depth^{*,1}$.

\begin{lemma}
  \labelT{lemma|OptimalQValues}
  Considering that vectors $\nu^2_{[\occ_{H}^{c,1},\beta_{H:}^2]}$
  are null vectors, we have, for all $\depth \in \{0\twodots H-1\}$:
  \begin{align*}
    W_\depth^{1,*}(\occ_\depth,\beta^1_\depth) & = \min_{\beta^2_\depth, \langle \beta_{\depth+1:}^2, \nu^2_{{[T^1_c(\occ_\depth, \beta^2_\depth), \beta^2_{\depth+1:}]}} \rangle} \beta^1_\depth \cdot \Big[ {
      \vr(\occ_\depth, \cdot, \beta^2_\depth) 
      + \gamma \Nxt^1_m(\occ_\depth, \cdot, \beta^2_\depth) \cdot \nu^2_{{[T^1_c(\occ_\depth, \beta^2_\depth),\beta^2_{\depth+1:}]}}
    } \Big].
  \end{align*}
\end{lemma}

\begin{proof}
  \label{proof|lemma|OptimalQValues}
  Considering that vectors $\nu^2_{[\occ_{H}^{c,1},\beta^2_{H:}]}$
  are null vectors, we have, for all $\depth \in \{0 \twodots H-1\}$:
  \begin{align*}
     W_\depth^{1,*}(\occ_\depth,\beta^1_\depth) & = \min_{\beta^2_\depth}  Q^*_\depth(\occ_\depth, \beta^1_\depth, \beta^2_\depth) = \min_{\beta^2_\depth}  \left[
      r(\occ_\depth, \vbeta_\depth)
      + \gamma V^*_{\depth+1}( \nxt(\occ_\depth, \vbeta_\depth) )
    \right] \intertext{(Line below exploits \Cref{theo|ConvexConcaveV} (p.~\pageref{theo|ConvexConcaveV}) and $\nxt^1_c$'s independence from $\beta^1_\depth$ (\Cref{lem|T1cindep}).)}
    & = \min_{\beta^2_\depth}  \left[
      r(\occ_\depth, \vbeta_\depth)
      + \gamma \min_{\langle \beta_{\depth+1:}^2, \nu^2_{{[T^1_c(\occ_\depth, \beta^2_\depth), \beta^2_{\depth+1:}]}} \rangle}  \left[
        \nxt_m^1(\occ_\depth, \vbeta_\depth) \cdot \nu^2_{{[T^1_c(\occ_\depth, \beta^2_\depth), \beta^2_{\depth+1:}]}}
      \right]
    \right] \\
    & = \min_{\beta^2_\depth, \langle \beta_{\depth+1:}^2, \nu^2_{{[T^1_c(\occ_\depth, \beta^2_\depth), \beta^2_{\depth+1:}]}} \rangle} \left[
      r(\occ_\depth, \vbeta_\depth)
      + \gamma \nxt_m^1(\occ_\depth, \vbeta_\depth) \cdot \nu^2_{{[T^1_c(\occ_\depth, \beta^2_\depth), \beta^2_{\depth+1:}]}}
    \right] \intertext{(Line below exploits $r$ and $\nxt^1_m$'s linearity in $\beta^1_\depth$ (\Cref{lem|T1mlin}).)}
    & = \min_{\beta^2_\depth, \langle \beta_{\depth+1:}^2, \nu^2_{{[T^1_c(\occ_\depth, \beta^2_\depth), \beta^2_{\depth+1:}]}} \rangle} {\beta^1_\depth}^\t \! \cdot \Big[ {
      \vr(\occ_\depth, \cdot, \beta^2_\depth) 
      + \gamma \Nxt^1_m(\occ_\depth, \cdot, \beta^2_\depth) \cdot \nu^2_{{[T^1_c(\occ_\depth, \beta^2_\depth),\beta^2_{\depth+1:}]}}
    } \Big].
\qedhere
  \end{align*}
\end{proof}

Note that, since $V^*_H=0$, $\depth=H-1$ is a particular case which can be simply re-written:
\begin{align*}
   W_\depth^{1,*}(\occ_\depth,\beta^1_\depth) & = \min_{\beta^2_\depth} {\beta^1_\depth}^\t \! \cdot
   \vr(\occ_\depth, \cdot, \beta^2_\depth).
 \end{align*}

\begin{tcolorbox}[breakable, enhanced]

  \uline{Addendum:} The following complementary property is not
  directly used in the present work, but makes for a more complete
  table of properties (\Cref{tab|PropertyTable}).

    \begin{proposition}
      \labelT{lemma|Wconcave}
      $W^{1,*}_\depth(\occ_\depth,\beta^1_\depth)$ is concave in $\beta^1_\depth$.
    \end{proposition}

    \begin{proof}
      \label{proof|lemma|Wconcave}
      Let $X$ and $Y$ be two convex domains, $f: X\times Y \mapsto \reals$ be a concave-convex function, and $g(x)\eqdef \min_{y\in Y} f(x,y)$.
Then, for any $x_1, x_2 \in X$, and any $\alpha \in [0,1]$,
      \begin{align}
        g(\alpha x_1 + (1-\alpha) x_2)
        & = min_y \underbrace{ f( \alpha x_1 + (1-\alpha) x_2, y ) }_{
          \geq \alpha f( x_1, y ) + (1-\alpha) f( x_2, y )
        }
        & \text{(concavity in $x$)}
        \\
        & \geq min_y \left[ \alpha f( x_1, y ) + (1-\alpha) f( x_2, y ) \right]
        \\
        & \geq min_y \alpha f( x_1, y ) +  min_y (1-\alpha) f( x_2, y )
        \\
        & = \alpha g( x_1 ) +  (1-\alpha) g( x_2 ).
        & \text{($\alpha\geq 0$)}
      \end{align}
      $g$ is thus concave in $x$.

      This result directly applies to the function at hand, proving its concavity in $\beta^1_\depth$.
    \end{proof}
  \end{tcolorbox}

To find a form that could be appropriate for an upper bound approximation of $W^{*,1}_\depth$,
let us now consider an \os{} $\occ_\depth$ and a single tuple $\langle { \tilde\occ_\depth, \tilde\beta^2_\depth, \nu^2_{[T^1_c(\tilde\occ_\depth, \tilde\beta^2_\depth), \tilde\beta_{\depth+1:}^2]} } \rangle$.
Then,
\begin{align}
  W_\depth^{1,*}(\occ_\depth,\beta^1_\depth) & = \min_{\beta^2_\depth}  \left[
    r(\occ_\depth, \vbeta_\depth)
    + \gamma V^*_{\depth+1}( \nxt(\occ_\depth, \vbeta_\depth) )
  \right]
  \nonumber \\
& \leq   
  r(\occ_\depth, \beta^1_\depth, \tilde\beta^2_\depth)
  + \gamma  V^{BR,1}_{\depth+1}( \nxt(\occ_\depth, \beta^1_\depth, \tilde\beta^2_\depth) | \tilde\beta_{\depth+1:}^2)
  \nonumber
  \qquad \text{\small (Use 
    $\tilde\beta^2_\depth$ \& $\tilde\beta^2_{\depth+1:}$ instead of mins)} \intertext{\small (where $V^{BR,1}_{\depth+1}( \nxt(\occ_\depth, \beta^1_\depth, \tilde\beta^2_\depth) | \tilde\beta_{\depth+1:}^2)$ is the value of 1's best response to $\tilde\beta^2_{\depth+1:}$ if in $\nxt(\occ_\depth, \beta^1_\depth, \tilde\beta^2_\depth)$)}
& = r(\occ_\depth, \beta^1_\depth, \tilde\beta^2_\depth)
  + \gamma 
  \nxt_m^1(\occ_\depth, \beta^1_\depth, \tilde\beta^2_\depth) \cdot \underbrace{ \nu^2_{[\nxt^1_c(\occ^{c,1}_\depth, \tilde\beta^2_\depth), \tilde\beta_{\depth+1:}^2]} } \nonumber
  \qquad \text{\small (Lem.~3 of \citet{WigOliRoi-corr16})} \\
  & \leq r(\occ_\depth, \beta^1_\depth, \tilde\beta^2_\depth)
  + \gamma 
  \nxt_m^1(\occ_\depth, \beta^1_\depth, \tilde\beta^2_\depth) \cdot \Big( {
    \nu^2_{[\nxt^1_c(\tilde\occ^{c,1}_\depth, \tilde\beta^2_\depth), \tilde\beta_{\depth+1:}^2]} }
  \qquad  \text{\small (\Cshref{lem|nuLC}: $\lt{\depth+1}$-LC}
  \\
  & \qquad \qquad {
    + \lt{\depth+1} \cdot \vnorm{ \nxt^1_c(\occ^{c,1}_\depth, \tilde\beta^2_\depth) - \nxt^1_c(\tilde\occ^{c,1}_\depth, \tilde\beta^2_\depth) }_1
  } \Big)
  \qquad \qquad \text{\small  of $\nu^2_{[\nxt^1_c(\occ^{c,1}_\depth, \tilde\beta^2_\depth), \tilde\beta_{\depth+1:}^2]}$)}
  \nonumber \\
  & \label{eq|WsingleTerm}
  = {\beta^1_\depth}^\t \! \cdot \Big[
  r(\occ_\depth, \cdot, \tilde\beta^2_\depth)
  + \gamma 
  \nxt_m^1(\occ_\depth, \cdot, \tilde\beta^2_\depth) \cdot \Big( {
    \nu^2_{[\nxt^1_c(\tilde\occ^{c,1}_\depth, \tilde\beta^2_\depth), \tilde\beta_{\depth+1:}^2]} }
  \qquad  \text{\small (Linearity in $\beta^1_\depth$)}
  \\
  & \qquad \qquad {
    + \lt{\depth+1} \cdot \vnorm{ \nxt^1_c(\occ^{c,1}_\depth, \tilde\beta^2_\depth) - \nxt^1_c(\tilde\occ^{c,1}_\depth, \tilde\beta^2_\depth) }_1
  } \Big)
  \Big]
  \nonumber \\
  & = {\beta^1_\depth}^\t \! \cdot \Big[
  r(\occ_\depth, \cdot, \tilde\beta^2_\depth)
  + \gamma 
  \nxt_m^1(\occ_\depth, \cdot, \tilde\beta^2_\depth) \cdot {
    \nu^2_{[\nxt^1_c(\tilde\occ^{c,1}_\depth, \tilde\beta^2_\depth), \tilde\beta_{\depth+1:}^2]} }
  \qquad \text{\small (Alternative writing)}
  \\
  & \qquad \qquad {
    + \gamma \lt{\depth+1} \cdot \norm{ \nxt(\occ_\depth, \cdot, \tilde\beta^2_\depth) - 
      \nxt_m^1(\occ_\depth, \cdot, \tilde\beta^2_\depth) \nxt^1_c(\tilde\occ^{c,1}_\depth, \tilde\beta^2_\depth) }_1
  }
  \Big]
  \nonumber
\end{align}

From this, we can deduce the following appropriate forms of (i) upper bounding approximation for $W^{1,*}_\depth$ and (ii) (symmetrically) of lower bound approximation for
$W^{2,*}_\depth$:
\begin{align*}
  \upb{W}^1_\depth(\occ_\depth, \beta^1_\depth) 
  & = \min_{ \langle \tilde\occ^{c,1}_\depth, \beta^2_\depth, \upb\nu^2_{\depth+1} \rangle \in \upb{bagW}^1_\depth
    }{\beta^1_\depth}^\t \cdot \Big[ {
      \vr(\occ_\depth, \cdot, \beta^2_\depth)
      + \gamma \Nxt^1_m(\occ_\depth, \cdot, \beta^2_\depth) \cdot \upb\nu^2_{\depth+1}
    } \\
    & \qquad \qquad
    + \gamma \lt{\depth+1} \cdot \norm{ \nxt(\occ_\depth, \cdot, \beta^2_\depth) - \Nxt^1_m(\occ_\depth, \cdot, \beta^2_\depth) \nxt^1_c(\tilde\occ^{c,1}_\depth, \beta^2_\depth) }_1
    \Big],
    \text{ and} \\
\lob{W}^2_\depth(\occ_\depth, \beta^2_\depth) 
  & = \max_{ \langle \tilde\occ^{c,2}_\depth, \beta^1_\depth, \lob\nu^1_{\depth+1} \rangle \in \lob{bagW}^2_\depth
    }{\beta^2_\depth}^\t \cdot \Big[ {
      \vr(\occ_\depth, \beta^1_\depth, \cdot)
      + \gamma \Nxt^2_m(\occ_\depth, \beta^1_\depth, \cdot) \cdot \lob\nu^1_{\depth+1}
    } \\
    & \qquad \qquad
    - \gamma \lt{\depth+1} \cdot \norm{ \nxt(\occ_\depth, \beta^1_\depth, \cdot) - \Nxt^2_m(\occ_\depth, \beta^1_\depth, \cdot) \nxt^2_c(\tilde\occ^{c,2}_\depth, \beta^1_\depth) }_1
    \Big],
\end{align*}
where $\upb\nu^2_{\depth+1}$ and $\lob\nu^1_{\depth+1}$ respectively upper and lower bound the actual vectors associated to the players' future strategies (resp. of $2$ and $1$).

Again, $\depth=H-1$ is a particular case where only the reward term is preserved.

\subsection{Related Operators}

\subsubsection{Selection Operator: Solving for $\beta^1_\depth$ as an LP} \label{sec|getLP}

\begin{proposition}
  \labelT{prop|bilinearValue}
  Using now a distribution $\delta^2_\depth$ over tuples 
  $w = \langle \tilde\occ^{c,1}_\depth, \beta^2_\depth, \upb\nu^2_{\depth+1} \rangle \in
  \upb{bagW}^1_\depth$,
  the corresponding upper-bounding value for ``profile''
  $\langle \beta^1_\depth, \delta^2_\depth \rangle$ when in
  $\occ_\depth$ can be written as an expectancy:
  \begin{align*}
    {\beta^1_\depth}^\t \cdot M^{\occ_\depth} \cdot \delta^2_\depth,
  \end{align*}
  where $M^{\occ_\depth}$ is an $|\Theta^1_\depth \times \cA^1| \times |\upb{bagW}^1_\depth|$ matrix.
\end{proposition}

\begin{proof}
  \label{proof|prop|bilinearValue}
  From the right-hand side term in (\ref{eq|WsingleTerm}), the
  upper-bounding value associated to $\occ_\depth$, $\beta^1_\depth$ and
  a tuple
  $\langle \tilde\occ^{c,1}_\depth, \beta^2_\depth, \upb\nu^2_{\depth+1} \rangle \in
  \upb{bagW}^1_\depth$ can be written:
  \begin{align*}
    & {\beta^1_\depth}^\t \! \cdot \Big[
  r(\occ_\depth, \cdot, \beta^2_\depth)
  + \gamma 
  \nxt_m^1(\occ_\depth, \cdot, \beta^2_\depth) \cdot \Big( \upb\nu^2_{\depth+1}
{
    + \lt{\depth+1} \cdot \vnorm{ \nxt^1_c(\occ^{c,1}_\depth, \tilde\beta^2_\depth) - \nxt^1_c(\tilde\occ^{c,1}_\depth, \tilde\beta^2_\depth) }_1
  } \Big)
  \Big].
  \end{align*}
  Using now a distribution $\delta^2_\depth$ over tuples 
  $w = \langle \tilde\occ^{c,1}_\depth, \beta^2_\depth, \upb\nu^2_{\depth+1} \rangle \in
  \upb{bagW}^1_\depth$,
  the corresponding upper-bounding value for ``profile''
  $\langle \beta^1_\depth, \delta^2_\depth \rangle$ when in
  $\occ_\depth$ can be written as an expectancy:
  {\scalefont{.92}
  \begin{align*}
    & \sum_{w \in \upb{W}^1_\depth}
    {\beta^1_\depth}^\t \cdot  \Big[
    \vr(\occ_\depth, \cdot, \beta^2_\depth[w])
    + \gamma \Nxt_m^1(\occ_\depth, \cdot, \beta^2_\depth[w]) \cdot \Big( \upb\nu^2_{\depth+1}[w]
+ \lt{\depth+1} 
    \cdot \vnorm{ \nxt^1_c(\occ^{c,1}_\depth, \beta^2_\depth[w]) - \nxt^1_c(\tilde\occ^{c,1}_\depth[w], \beta^2_\depth[w]) }_1
    \Big)
    \Big] \cdot \delta^2_\depth(w)
    \intertext{(where $x[w]$ denotes the field $x$ of tuple $w$)}
    & = {\beta^1_\depth}^\t \cdot M^{\occ_\depth} \cdot \delta^2_\depth,
  \end{align*}
  }
  where $M^{\occ_\depth}$ is an $|\Theta^1_\depth \times \cA^1| \times |\upb{bagW}^1_\depth|$ matrix.
\end{proof}

For implementation purposes, using \Cshref{eq|reward,eq|occm1} (to develop respectively $r(\cdot,\cdot,\cdot)$ and
$\Nxt^1_m(\cdot,\cdot,\cdot)$), we can derive the expression of a
component, \ie, the upper-bounding value if $a^1$ is applied in
$\theta^1_\depth$ while $w$ is chosen:
\begin{align*}
  M^{\occ_\depth}_{(\langle \theta^1_\depth, a^1\rangle, w)}  & \eqdef 
  \vr(\occ_\depth, \cdot, \beta^2_\depth[w])
  + \gamma \Nxt_m^1(\occ_\depth, \cdot, \beta^2_\depth[w]) \cdot \Big( \upb\nu^2_{\depth+1}[w] +  \lt{\depth+1} \cdot
  \vnorm{  \nxt^1_c(\occ^{c,1}_\depth, \beta^2_\depth[w]) - \nxt^1_c(\tilde\occ^{c,1}_\depth[w], \beta^2_\depth[w]) }_1
  \Big) \\
  & =
  { 
    \sum_{s,\theta^2_\depth, a^2} \occ_\depth(\vth_\depth) b(s | \vth_\depth)
    \beta^2_\depth[w](a^2|\theta^2) r(s,\va)
  } \\
  & \qquad + \gamma \sum_{z^1} {
    \left[
      \sum_{\theta^2_\depth,a^2} \beta^2_\depth[w](a^2 | \theta^2_\depth)  \sum_{s,s',z^2} P^{\vz}_{\va}(s'|s) 
      b(s | \vth_\depth) \occ_\depth (\vth_\depth)
    \right]
  } \cdot \Big( \upb\nu^2_{\depth+1}[w](\theta^1_\depth, a^1, z^1) \\ & \qquad + 
  \lt{\depth+1} \cdot
  \vnorm{  \nxt^1_c(\occ^{c,1}_\depth, \beta^2_\depth[w]) - \nxt^1_c(\tilde\occ^{c,1}_\depth[w], \beta^2_\depth[w]) }_1(\theta^1_\depth, a^1, z^1)
  \Big) \\
  & = \sum_{\theta^2_\depth} \occ_\depth(\vth_\depth)
  \sum_{a^2} \beta^2_\depth[w](a^2|\theta^2_\depth) \\
  & \qquad \cdot \Bigg(
  { 
    \sum_{s}  b(s | \vth_\depth) r(s,\va)
  }
  + \gamma \sum_{z^1} {
    \left[
      \sum_{s,s',z^2} P^{\vz}_{\va}(s'|s) 
      b(s | \vth_\depth)
    \right]
  }
  \cdot \Big( \upb\nu^2_{\depth+1}[w](\theta^1_\depth, a^1, z^1) \\ & \qquad + \lt{\depth+1} \cdot
  \vnorm{  \nxt^1_c(\occ^{c,1}_\depth, \beta^2_\depth[w]) - \nxt^1_c(\tilde\occ^{c,1}_\depth[w], \beta^2_\depth[w]) }_1(\theta^1_\depth, a^1, z^1)
  \Big) \Bigg).
\end{align*}

Then, solving
$\max_{\beta^1_\depth} \upb{W}^1_\depth(\occ_\depth, \beta^1_\depth)$
can be rewritten as solving a zero-sum game where pure strategies are:
\begin{itemize}
\item for Player $1$, the choice of not $1$, but $|\Theta^1_\depth|$ actions (among $|\cA^1|$) and,
\item for Player $2$, the choice of $1$ element of $\upb{bagW}^1_\depth$.
\end{itemize}
One can view it as a Bayesian game with one type per history
$\theta^1_\depth$ for $1$, and a single type for $2$.

With our upper bound approximation, $\max_{\beta^1_\depth} \upb{W}^1_\depth(\occ_\depth,\beta^1_\depth)$
can thus be solved as the \ifextended{linear program}{LP}:
\begin{align*}
& \begin{array}{l@{\ }l@{\ }ll}
      \displaystyle
      \max_{\beta_\depth^1,v}
      v
      \quad \text{s.t. } & \text{(i)}
      & \forall w \in \upb{bagW}^1_\depth, & v \leq {\beta_\depth^1}^\t \! \cdot M^{\occ_\depth}_{(\cdot,w)}
      \\
      & \text{(ii)}
      & \forall \theta_\depth^1 \in \Theta_\depth^1,
      & {\displaystyle \sum_{a^1}} \beta_\depth^1(a^1|\theta_\depth^1)
        = 1,
    \end{array}
    \intertext{whose dual LP is given by}
& \begin{array}{l@{\ }l@{\ }ll}
        \displaystyle
        \min_{\delta^2_\depth,v}
        v
        \quad \text{s.t. }
        & \text{(i)}
        & \forall (\theta^1_\depth, a^1 ) \in \Theta^1_\depth\times \cA^1, & v \geq  M^{\occ_\depth}_{((\theta^1_\depth,a^1),\cdot)} \cdot \delta^2_\depth
        \\
        & \text{(ii)}
        &
        & {\displaystyle \sum_{w \in \upb{bagW}^1_\depth}} \! \delta^2_\depth( w )
          = 1.
      \end{array}
    \end{align*}
As can be noted, $M^{\occ_\depth}$'s columns corresponding to
$0$-probability histories $\theta^1_\depth$ in $\occ^{m,1}_\depth$ are
empty, so that the corresponding decision rules (for these histories)
are not relevant and can be set arbitrarily.
The actual implementation thus ignores these histories, whose
corresponding decision rules also do not need to be stored.

\subsubsection{Strategy Induced by $\delta^2_\depth$}

$\delta^2_{\depth}$, as a distribution over tuples in
$\upb{bagW}^1_{\depth}$, induces a recursively-defined strategy for
$2$ as
(left) a mixture of behavioral \dr{}s at $\depth$,
and (right) a mixture of other mixture strategies for $\depth+1$ on:
\begin{align*}
  \beta^2_\depth[\delta^2_\depth] \eqdef \sum_{\tilde\beta^2_\depth \in \upb{bagW}^1_{\depth}} \delta^2_\depth(\tilde\beta^2_\depth) \cdot \tilde\beta^2_\depth,
  & \qquad \text{and} \qquad
  \delta^2_{\depth+1}[\delta^2_\depth] \eqdef \sum_{\tilde\delta^2_{\depth+1}  \in \upb{bagW}^1_\depth} \delta^2_\depth(\tilde\delta^2_{\depth+1}) \cdot \tilde\delta^2_{\depth+1},
\end{align*}
until reaching the horizon.
$\delta^2_\depth$ needs to be stored as this strategy will play a key
role in the following.

For $\depth\geq 1$, both $\upb{V}_\depth$ and $\upb{W}^1_{\depth-1}$ rely
essentially on the same information and are strongly related, so that
we will discuss them together.
$\upb{bagV}_\depth$ contains tuples
$\langle { \occ^{c,1}_\depth, \langle \delta^2_{\depth}, \upb\nu^2_\depth
  \rangle
} \rangle$,
and $\upb{bagW}^1_{\depth-1}$ (for $\depth\geq 1$) related tuples
$\langle { \occ^{c,1}_{\depth-1}, \beta^2_{\depth-1}, \langle \delta^2_{\depth}, \upb\nu^2_\depth
  \rangle } \rangle$.

\Cref{fig|dag} represents (in rectangular nodes) the elements of
$\upb{bagW}^1$ reachable from a given element of $\upb{bagV}_0$ (the
ellipsoid root node).
The children of any internal node at level/depth $\depth$ (including
the root) are the nodes corresponding to the elements $w_{\depth+1}$
in the support of $\delta^2_\depth$ (\ie, the set $\supp(\delta^2_\depth)$ of
  elements with non-zero probability in distribution $\delta^2_\depth$).
Level $\depth=H-1$ corresponds to the leaves of this graph.

As can be observed, it is directed and acyclic.
On can thus extract a behavioral strategy from some $\delta^2_\depth$
through a recursive process or, better, dynamic programming (to avoid
repeating the same computations when the same internal node is reached
through various branches).

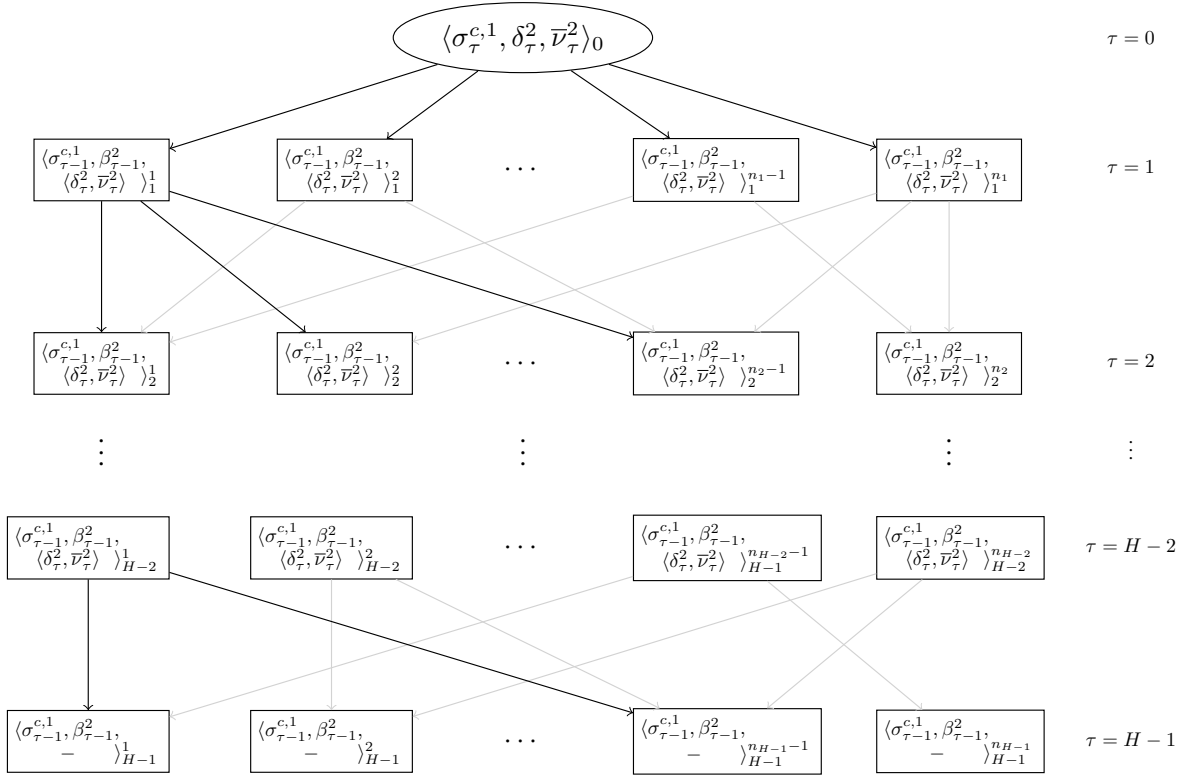
\begin{figure}
  \centerline{
    \resizebox{1.\linewidth}{!}{
      \def\sScript{.7} \def\sTiny{0.1} 

\newcommand{\nodeW}[2]{
  $\begin{array}{@{}c@{}c@{}l@{}}
     \langle
     & \occ^{c,1}_{\depth-1}, \beta^2_{\depth-1}
     &,  \\ & \langle \delta^2_\depth, \upb\nu^2_\depth \rangle
     & \rangle_{#2}^{#1}
   \end{array}
   $
}
\newcommand{\nodeWLast}[2]{
  $\begin{array}{@{}c@{}c@{}l@{}}
     \langle
     & \occ^{c,1}_{\depth-1}, \beta^2_{\depth-1}
     &,  \\ & -
     & \rangle_{#2}^{#1}
   \end{array}
   $
}

\colorlet{lightgray}{gray!35}

\begin{tikzpicture}

  \def\radA{1cm} \def\radB{4cm} \def\radC{7.5cm}
  \def\myshape{rectangle} 

\node [ellipse, draw, scale=1] (n0) at (0,0) {$\langle \occ^{c,1}_\depth, \delta^2_\depth, \upb\nu^2_\depth \rangle_0$};
  \node [scale=\sScript] at (\radC,0) {$\depth=0$};
  
\node [below = 1cm of n0] (n1dots) { $\cdots$ };
  \node [\myshape, draw, scale=\sScript, left = \radB of n1dots] (n11) { \nodeW{1}{1} };
  \node [\myshape, draw, scale=\sScript, left = \radA of n1dots] (n12) { \nodeW{2}{1} };
  \node [\myshape, draw, scale=\sScript, right = \radB of n1dots] (n1n) { \nodeW{n_1}{1} };
  \node [\myshape, draw, scale=\sScript, right = \radA of n1dots] (n1n1) { \nodeW{n_1-1}{1} };
\node [below = 1cm of n0, xshift=\radC, scale=\sScript] {$\depth=1$};
  
  \draw[->] (n0) -- (n11);
  \draw[->] (n0) -- (n12);
  \draw[->] (n0) -- (n1n);
  \draw[->] (n0) -- (n1n1);
  
\node [below = 2cm of n1dots] (n2dots) { $\cdots$ };
\node [\myshape, draw, scale=\sScript, left = \radB of n2dots] (n21) { \nodeW{1}{2} };
  \node [\myshape, draw, scale=\sScript, left = \radA of n2dots] (n22) { \nodeW{2}{2} };
  \node [\myshape, draw, scale=\sScript, right = \radB of n2dots] (n2n) { \nodeW{n_2}{2} };
  \node [\myshape, draw, scale=\sScript, right = \radA of n2dots] (n2n1) { \nodeW{n_2-1}{2} };
\node [below = 2cm of n1dots, xshift=\radC, scale=\sScript] {$\depth=2$};
  
  \draw[->, lightgray] (n12) -- (n21);
\draw[->, lightgray] (n12) -- (n2n1);

  \draw[->, lightgray] (n1n1) -- (n21);
\draw[->, lightgray] (n1n1) -- (n2n);

\draw[->, lightgray] (n1n) -- (n22);
  \draw[->, lightgray] (n1n) -- (n2n);
  \draw[->, lightgray] (n1n) -- (n2n1);
  
  \draw[->] (n11) -- (n21);
  \draw[->] (n11) -- (n22);
  \draw[->] (n11) -- (n2n1);

  \node [below of = n21] (n21vdots) { $\vdots$ };
  \node [below of = n2dots] (n2dotsvdots) { $\vdots$ };
  \node [below of = n2n] (n2nvdots) { $\vdots$ };
  \node [below of = n2dots, xshift=\radC, scale=\sScript] {$\vdots$};
  
\node [below = .7cm of n2dotsvdots] (nH2dots) { $\cdots$ };
  \node [\myshape, draw, scale=\sScript, left = \radB of nH2dots] (nH21) { \nodeW{1}{H-2} };
  \node [\myshape, draw, scale=\sScript, left = \radA of nH2dots] (nH22) { \nodeW{2}{H-2} };
  \node [\myshape, draw, scale=\sScript, right = \radB of nH2dots] (nH2n) { \nodeW{n_{H-2}}{H-2} };
  \node [\myshape, draw, scale=\sScript, right = \radA of nH2dots] (nH2n1) { \nodeW{n_{H-2}-1}{H-1} };
\node [below = .7cm of n2dotsvdots, xshift=\radC, scale=\sScript] {$\depth=H-2$};
  
\node [below = 2cm of nH2dots] (nH1dots) { $\cdots$ };
  \node [\myshape, draw, scale=\sScript, left = \radB of nH1dots] (nH11) { \nodeWLast{1}{H-1} };
  \node [\myshape, draw, scale=\sScript, left = \radA of nH1dots] (nH12) { \nodeWLast{2}{H-1} };
  \node [\myshape, draw, scale=\sScript, right = \radB of nH1dots] (nH1n) { \nodeWLast{n_{H-1}}{H-1} };
  \node [\myshape, draw, scale=\sScript, right = \radA of nH1dots] (nH1n1) { \nodeWLast{n_{H-1}-1}{H-1} };
\node [below = 2cm of nH2dots, xshift=\radC, scale=\sScript] {$\depth=H-1$};
  
\draw[->, lightgray] (nH22) -- (nH12);
  \draw[->, lightgray] (nH22) -- (nH1n1);

  \draw[->, lightgray] (nH2n1) -- (nH11);
\draw[->, lightgray] (nH2n1) -- (nH1n);

\draw[->, lightgray] (nH2n) -- (nH12);
\draw[->, lightgray] (nH2n) -- (nH1n1);
  
  \draw[->] (nH21) -- (nH11);
\draw[->] (nH21) -- (nH1n1);

\end{tikzpicture}
     }
  }
  \caption{DAG structure of the recursively defined strategy induced
    by $\delta^2_0$. Each $\delta^2_\depth$ (part of a tuple $w_\depth$) is a
    probability distribution over elements $w_{\depth+1}$, thus
    inducing 2 other probability distributions: (i) one over decision rules $\beta^2_\depth$, and (ii) the other over probability distributions
    $\delta^2_{\depth+1}$.}
  \label{fig|dag}
\end{figure}

\subsubsection{Upper Bounding $\nu^2_{[\occ^{c,1}_\depth,\delta^2_\depth]}$}
\label{sec|compute|nu}

Adding a new complete tuple to $\upb{bagW}^1_\depth$ requires a new
vector $\upb\nu^2_{\depth}$ that upper bounds the vector
$\nu^2_{[\occ^{c,1}_\depth,\delta^2_\depth]}$ associated to the
strategy induced by $\delta^2_\depth$.
We can obtain one in a recursive manner (not solving the induced
POMDP).

\begin{restatable}{proposition}{propRecNu}
  \labelT{prop|rec|nu}
For each $\delta^2_\depth$ obtained as the solution of the
  aforementioned (dual) LP in $\occ_\depth$, and each
  $\theta^1_\depth$,
  $\nu^2_{[\occ^{c,1}_{\depth},
    \delta^2_{\depth}]}(\theta^1_\depth)$ is upper bounded by a value
  $\upb\nu^2_{\depth}(\theta^1_\depth)$ that depends on vectors
  $\upb\nu^2_{\depth+1}$ in the support of $\delta^2_\depth$.
In particular, if $\theta^1_\depth \in \supp(\occ^{m,1}_\depth)$, we
  have:
  \begin{align*}
    \upb\nu^2_{\depth}(\theta^1_\depth)
    & \eqdef \frac{1}{\occ^{1}_{\depth,m}(\theta^1_\depth)} \max_{a^1 \in \cA^1} M^{\occ_\depth}_{((\theta^1_\depth, a^1), . )} \cdot \delta^2_\depth.
  \end{align*}
\end{restatable}

\begin{proof}
  \label{proof|prop|rec|nu}

For a newly derived $\delta^2_\depth$, as
  $\nu^2_{[\occ^{c,1}_\depth, \delta^2_\depth]}(\theta^1_\depth)$
  is the value of $1$'s best action ($\in \cA^1$) if $1$ (i) observes $\theta^1_\depth$ while in $\occ^{c,1}_\depth$ and (ii) $2$ plays $\delta^2_\depth$, we have:
\begin{align}
    \hspace{1cm}
    & \hspace{-1cm} \nu^2_{[\occ^{c,1}_\depth, \delta^2_\depth]}(\theta^1_\depth) \eqdef V^\star_{[\occ^{c,1}_\depth, \delta^2_\depth]}(\theta^1_\depth) \qquad \text{(optimal POMDP value function)}
    \nonumber \\
& = \max_{\beta^1_{\depth:}} \E\left[
      \sum_{t=\depth}^H \gamma^{t-\depth} R_t
      \mid \beta^1_{\depth:}, \theta^1_\depth, \occ^{c,1}_\depth, \delta^2_\depth
    \right]
    \nonumber \\
    & = \max_{a^1} \E\left[
      R_\depth
      +
      \gamma \max_{\beta^1_{\depth+1:}}
      \E\left[
        \sum_{t=\depth+1}^H \gamma^{t-(\depth+1)} R_t
        \mid \beta^1_{\depth+1:}, \langle \theta^1_\depth, a^1, Z^1 \rangle, \occ^{c,1}_{\depth+1}, \delta^2_{\depth+1}
      \right] 
      \mid a^1, \theta^1_\depth, \occ^{c,1}_\depth, \delta^2_\depth
    \right] \\
    & = \max_{a^1} \E\left[
      R_\depth
      + \gamma V^\star_{[\occ^{c,1}_{\depth+1}, \delta^2_{\depth+1}]}( \theta^1_\depth, a^1, Z^1 )
      \mid a^1, \theta^1_\depth, \occ^{c,1}_\depth, \delta^2_\depth
    \right] \\
    & = \max_{a^1} \sum_{w, \theta^2_\depth, a^2, z^1}
    \underbrace{Pr(w, \theta^2_\depth, z^1, a^2 \mid a^1, \theta^1_\depth, \occ^{c,1}_\depth, \delta^2_\depth)} \\
    & \qquad \cdot
    \left(
      r(\vth_\depth,\va_\depth)
      + \gamma \nu^2_{[\occ^{c,1}_{\depth+1}, \delta^2_{\depth+1}[w]]}( \theta^1_\depth, a^1, z^1 )
    \right)
    \qquad \text{\small (where $\occ^{c,1}_{\depth+1}=T^1_c(\occ^{c,1}_\depth, \beta^2_\depth[w])$
      {\scriptsize (lem.~\refpage{lem|T1cindep})})}\\
    & = \max_{a_1}  
    \sum_{w, \theta^2_\depth, a^2, z^1}
    \underbrace{Pr(w | \delta^2_\depth)}
    \cdot \underbrace{Pr(\theta^2_\depth | \theta^1_\depth, \occ^{c,1}_\depth)}
    \cdot \underbrace{Pr( a^2 | \beta^2_\depth[w], \theta^2_\depth)}
    \cdot \underbrace{Pr( z^1 | \vth_\depth, \va_\depth)} \\
    & \qquad \cdot
    \left(
      r(\vth_\depth, \va) 
      + \gamma \nu^2_{[\occ^{c,1}_{\depth+1}, \delta^2_{\depth+1}[w]]}(\theta^1_\depth, a^1, z^1)
    \right) \\
& = \max_{a_1}  \sum_w \delta^2_\depth(w)
    \sum_{\theta^2_\depth} \occ^{c,1}_\depth(\theta^2_\depth|\theta^1_\depth)
    \sum_{a^2} \beta^2_\depth[w](a^2|\theta^2_\depth) \\
    & \qquad \cdot
    \left(
      r(\vth_\depth, \va) 
      + \gamma \sum_{z^1} Pr(z^1|\vth_\depth,\va) \underbrace{\nu^2_{[\occ^{c,1}_{\depth+1}, \delta^2_{\depth+1}[w]]}(\theta^1_\depth, a^1, z^1)}
    \right) \\
    \intertext{then, as $\nu^2_{[\occ^{c,1}_{\depth+1}, \delta^2_{\depth+1}[w]]}$ is $\l_{\depth+1}$-LC in (any) $\occ^{c,1}_{\depth+1}$ (\Cref{lem|nuLC}),}
    & \leq \max_{a_1} \sum_w \delta^2_\depth(w)
    \sum_{\theta^2_\depth} \occ^{c,1}_\depth(\theta^2_\depth|\theta^1_\depth)
    \sum_{a^2} \beta^2_\depth[w](a^2|\theta^2_\depth) \cdot
    \Bigg(
    r(\vth_\depth, \va)     
    \\
    & + \gamma \sum_{z^1} Pr(z^1|\vth_\depth,\va)
      \left[
        \overbrace{
          \underbrace{\nu^2_{[\tilde\occ^{c,1}_{\depth+1}[w], \delta^2_{\depth+1}[w]]}(\theta^1_\depth, a^1, z^1)}
          + \l_{\depth+1} \vnorm{\occ^{c,1}_{\depth+1} - \tilde\occ^{c,1}_{\depth+1}[w]}_1(\theta^1_\depth, a^1, z^1)
        }
      \right]
    \Bigg) \\
    & \leq \max_{a_1} \sum_w \delta^2_\depth(w)
    \sum_{\theta^2_\depth} \occ^{c,1}_\depth(\theta^2_\depth|\theta^1_\depth)
    \sum_{a^2} \beta^2_\depth[w](a^2|\theta^2_\depth) \cdot
    \Bigg(
    \underbrace{ r(\vth_\depth, \va) }
    \\
    & \quad
    + \gamma \sum_{z^1} \underbrace{Pr(z^1|\vth_\depth,\va)}
    \left[
      \overbrace{ \upb\nu^2_{\depth+1}[w](\theta^1_\depth, a^1, z^1) }
      +  \l_{\depth+1} \vnorm{\occ^{c,1}_{\depth+1} - \tilde\occ^{c,1}_{\depth+1}[w]}_1(\theta^1_\depth, a^1, z^1)
    \right]
    \Bigg) 
    \\
    & = \max_{a_1} \sum_w \delta^2_\depth(w)
    \sum_{\theta^2_\depth} \occ^{c,1}_\depth(\theta^2_\depth|\theta^1_\depth)
    \sum_{a^2} \beta^2_\depth[w](a^2|\theta^2_\depth) 
    \cdot
    \Bigg(
      \overbrace{\sum_s b(s|\vth_\depth) r(s, \va)}
    \\
    & \quad
      + \gamma \sum_{z^1} \left( \overbrace{ \sum_s b(s|\vth_\depth) \underbrace{Pr(z^1|s,\va)} } \right)
      \cdot \Big[ \upb\nu^2_{\depth+1}[w](\theta^1_\depth, a^1, z^1)
     \\
    & \quad
    + \l_{\depth+1} \underbrace{ \vnorm{\occ^{c,1}_{\depth+1} - \tilde\occ^{c,1}_{\depth+1}[w]}_1(\theta^1_\depth, a^1, z^1) }
    \Big]
    \Bigg)
    \\
    & = \max_{a_1} \sum_w \delta^2_\depth(w)
    \sum_{\theta^2_\depth} \occ^{c,1}_\depth(\theta^2_\depth|\theta^1_\depth)
    \sum_{a^2} \beta^2_\depth[w](a^2|\theta^2_\depth)
    \label{eq|nuNotInSupport} \\
    & \quad
    \cdot
    \Bigg(
    \sum_s b(s|\vth_\depth) r(s, \va)     
    + \gamma \sum_{z^1} \left( \overbrace{ \sum_{s, s', z^2} b(s|\vth_\depth) P^{\vz}_{\va}(s'|s) } \right)
    \cdot \Big[ \upb\nu^2_{\depth+1}[w](\theta^1_\depth, a^1, z^1)
    \nonumber \\
    & \quad 
    + \l_{\depth+1}
    \overbrace{ \vnorm{ \nxt^1_c(\occ^{c,1}_\depth, \beta^2_\depth[w]) - \nxt^1_c(\tilde\occ^{c,1}_\depth[w], \beta^2_\depth[w]) }_1(\theta^1_\depth, a^1, z^1) }
    \Big] \Bigg)
    \nonumber \\
    & = \frac{1}{\occ^{1}_{\depth,m}(\theta^1_\depth)} \max_{a^1 \in \cA^1} M^{\occ_\depth}_{((\theta^1_\depth, a^1), . )} \cdot \delta^2_\depth.
    \nonumber
\end{align}

\end{proof}

\subsubsection{Pruning $\upb{V}_\depth$}
\label{app|pruningV}

The following key theorem allows reusing usual POMDP $\max$-planes
pruning techniques in our setting (reverting them to handle
$\min$-planes upper bound approximations).

\begin{restatable}[Proof in \extCshref{app|pruningV}]{theorem}{lemPruningV}
  \labelT{lem|pruningV}
  \IfAppendix{{\em (originally stated on
      page~\pageref{lem|pruningV})}}{}
Let $P$ be a $\min$-planes pruning operator (inverse of $\max$-planes pruning for POMDPs), and
$\langle \occ^{c,1}_\depth, \upb\nu^2_\depth \rangle \in \upb{bagV}_\depth$.
If $P$ correctly identifies $\upb\nu^2_\depth$ as non-dominated (or
  resp. dominated) under fixed $\occ^{c,1}_\depth$, then
  $\langle \occ^{c,1}_\depth, \upb\nu^2_\depth \rangle$ is non-dominated
  (or resp. dominated) in $\Occ_\depth$.
\end{restatable}

\begin{proof}
  \label{proof|lem|pruningV}
  We will demonstrate that:
  \begin{itemize}
  \item if $P$ shows that a vector $\nu^2_\depth$ (associated to
    $\occ_\depth$) is dominated {\em under fixed $\occ_\depth^{c,1}$}
    by a $\min$-planes upper bound relying only on other vectors
    $\tilde \nu^2_\depth$, then this vector is dominated in the whole
    space $\Occ$;
\item else, the vector $\nu^2_\depth$ is useful at least around
    $\xi_\depth = (\xi_\depth^{m,1}, \occ_\depth^{c,1})$, where
    $\xi_\depth^{m,1}$ is the domination point returned by $P$.
  \end{itemize}
Note: The following is simply showing that, if the linear part is dominated by a $\min$-planes approximation for a given conditional term $\occ_\depth^{c,1}$, then the Lipschitz generalization in the space of conditional terms is also dominated since $\l$ is constant.

Given a matrix $M=(m_{i,j})$, let $\vnorm{M}_1$ denote the column vector whose $i$th component is $\norm{m_{i,\cdot}}_1$.
Here, such matrices will correspond to conditional terms,
$\vnorm{ \occ_\depth^{c,1} - \tilde \occ_\depth^{c,1} }_1$ denoting
the vector whose component for \aoh{} $\theta^1_\depth$ is
$\norm{ \occ_\depth^{c,1}(\cdot | \theta_{\depth}^1) - \tilde
\occ_\depth^{c,1}( \cdot | \theta_{\depth}^1) }_1$ (where $\occ_\depth^{c,1}( \cdot | \theta_{\depth}^1)$ may also be denoted $\occ_\depth^{c,1}(\theta_{\depth}^1)$ for brevity).

Let us assume that the vector $\nu^2_\depth$ (associated to $\occ_\depth^{c,1}$) is dominated under $\occ_\depth^{c,1}$, \ie, $\forall \xi_\depth^{m,1}$,
\begin{align*}
  (\xi_\depth^{m,1})^\top \cdot ( \nu^2_\depth + \lt{\depth} \overbrace{\vnorm{ \occ_\depth^{c,1} - \occ_\depth^{c,1} }_1}^\text{$0$})
  & \geq \min_{\tilde \nu^2_\depth, \tilde \occ_\depth^{c,1}}  \left[ (\xi_\depth^{m,1})^{\top} \cdot ( \tilde \nu^2_\depth + \lt{\depth} \vnorm{ \occ_\depth^{c,1} - \tilde \occ_\depth^{c,1} }_1 ) \right] .
\intertext{We will show that, $\forall \xi_\depth = (\xi_\depth^{m,1},\xi_\depth^{c,1})$,}
(\xi_\depth^{m,1})^{\top} \cdot ( \nu^2_\depth + \lt{\depth} \vnorm{ \xi_\depth^{c,1} - \occ_\depth^{c,1} }_1 )
  & \geq \min_{\tilde \nu^2_\depth, \tilde \occ_\depth^{c,1}}  \left[ (\xi_\depth^{m,1})^{\top} \cdot ( \tilde \nu^2_\depth + \lt{\depth} \vnorm{ \xi_\depth^{c,1} - \tilde \occ_\depth^{c,1} }_1 ) \right].
\intertext{Let $\xi_\depth$ be an occupancy state. First, remark that $\exists \langle \tilde \nu^2_\depth, \tilde \occ_\depth^{c,1} \rangle$  such that}
(\xi_\depth^{m,1})^{\top} \cdot ( \nu^2_\depth + \lt{\depth} \overbrace{\vnorm{ \occ_\depth^{c,1} - \occ_\depth^{c,1} }_1 }^\text{$0$})
  & \geq (\xi_\depth^{m,1})^{\top} \cdot ( \tilde \nu^2_\depth + \lt{\depth} \vnorm{ \occ_\depth^{c,1} - \tilde \occ_\depth^{c,1} }_1 ).
\end{align*}

For the sake of clarity, let us introduce the following functions (where $\xx$, $\yy$, and $\zz$ will denote conditional terms for player $1$):
\begin{align}
g(x)
  & \eqdef \sum_{\theta^1_\depth} \xi_\depth^{m,1}(\theta^1_\depth) \cdot (\nu^2_{\yy}(\theta^1_\depth) + \lt{\depth} \norm{ \yy(\theta^1_\depth) - \xx(\theta^1_\depth)}_1)
  \nonumber \\
  & = g(\yy) + \lt{\depth} (\xi_\depth^{m,1})^{\top} \cdot \vnorm{\yy - \xx}_1, \intertext{and}
  h(\xx) 
  & \eqdef \sum_{\theta^1_\depth} \xi_\depth^{m,1}(\theta^1_\depth) \cdot (\nu^2_{\zz}(\theta^1_\depth) + \lt{\depth} \norm{\zz(\theta^1_\depth) - \xx(\theta^1_\depth)}_1)
  \nonumber \\
  & = h(\zz) + \lt{\depth} (\xi_\depth^{m,1})^{\top} \cdot \vnorm{\zz - \xx}_1.
  \nonumber
\intertext{Let us assume that $g(\yy) \geq h(\yy)$, and show that $g \geq h$. First,}
g(\xx) &= g(\yy) + \lt{\depth} (\xi_\depth^{m,1})^{\top} \cdot \vnorm{\xx-\yy}_1
  \nonumber \\
  & \geq h(\yy) + \lt{\depth} (\xi_\depth^{m,1})^{\top} \cdot \vnorm{\xx-\yy}_1
  \nonumber \\
  & = h(\zz) + \lt{\depth} (\xi_\depth^{m,1})^{\top} \cdot ( \vnorm{\yy-\zz}_1 + \vnorm{\xx-\yy}_1 )
  \nonumber \\
  &\geq h(\zz) + \lt{\depth} (\xi_\depth^{m,1})^{\top} \cdot \left( \vnorm{\yy-\zz}_1 + \abs{ \vnorm{\xx-\zz}_1 - \vnorm{\zz-\yy}_1  } \right).
  \label{eq|azerty}
\end{align}
Now, $\forall \theta^1_\depth$, if $\norm{\xx(\theta^1_\depth) - \zz(\theta^1_\depth)}_1 - \norm{\zz(\theta^1_\depth) - \yy(\theta^1_\depth)}_1 \geq 0$, then
\begin{align}
  & \norm{\yy(\theta^1_\depth)  - \zz(\theta^1_\depth) }_1 + \left|\norm{\xx(\theta^1_\depth) - \zz(\theta^1_\depth)}_1 - \norm{\zz(\theta^1_\depth) - \yy(\theta^1_\depth)}_1 \right|
  \nonumber \\
  & =  \cancel{\norm{\yy(\theta^1_\depth) - \zz(\theta^1_\depth)}_1} + \norm{\xx(\theta^1_\depth) - \zz(\theta^1_\depth) }_1 - \cancel{\norm{\zz(\theta^1_\depth) - \yy(\theta^1_\depth) }_1}
  \nonumber \\
  & = \norm{\xx(\theta^1_\depth) - \zz(\theta^1_\depth)}_1,
  \label{eq|uiop}
\intertext{else,}
& \norm{\yy(\theta^1_\depth)  - \zz(\theta^1_\depth) }_1 + \left|\norm{\xx(\theta^1_\depth) - \zz(\theta^1_\depth)}_1 - \norm{\zz(\theta^1_\depth) - \yy(\theta^1_\depth)}_1 \right| \\
  & =  2 \norm{\yy(\theta^1_\depth) - \zz(\theta^1_\depth)}_1 - \norm{\xx(\theta^1_\depth) - \zz(\theta^1_\depth)}_1
  \nonumber \\
  & \geq  \norm{\xx(\theta^1_\depth) - \zz(\theta^1_\depth)}_1.
  \label{eq|qsdf}
\intertext{Finally, coming back to (\ref{eq|azerty}):} g(\xx)
  & \geq h(\zz) + \lt{\depth} (\xi_\depth^{m,1})^{\top} \cdot \left( \vnorm{\yy-\zz}_1 + \abs{ \vnorm{\xx-\zz}_1 - \vnorm{\zz-\yy}_1 } \right)
  \nonumber \\
  & \geq h(\zz) + \lt{\depth} (\xi_\depth^{m,1})^{\top} \cdot \norm{\xx(\theta^1_\depth) - \zz(\theta^1_\depth)}_1
  \qquad \qquad \text{(from (\ref{eq|uiop}+\ref{eq|qsdf}))}
  \nonumber \\
  & \geq h(\xx).
  \nonumber
\intertext{With $x=\xi_\depth^{c,1}$, $y=\occ_\depth^{c,1}$ and $z=\tilde\occ_\depth^{c,1}$, this gives:}
g(\xi_\depth^{c,1}) & = \sum_{\theta^1_\depth} \xi_\depth^{m,1}(\theta^1_\depth) (\nu^2_\depth(\theta^1_\depth) + \lt{\depth} \norm{ \occ_\depth^{c,1}(\theta^1_\depth) - \xi_\depth^{c,1}(\theta^1_\depth)}_1) \\
  \geq h(\xi_\depth^{c,1}) & = \sum_{\theta^1_\depth} \xi_\depth^{m,1}(\theta^1_\depth) (\tilde \nu^2_\depth(\theta^1_\depth) + \lt{\depth} \norm{\tilde \occ_\depth^{c,1}(\theta^1_\depth) - \xi_\depth^{c,1}(\theta^1_\depth)}_1).
\end{align}
This shows that $\nu^2_\depth$ is dominated for every $(\xi_\depth^{m,1},\xi_\depth^{c,1})$, where both  $\xi_\depth^{m,1}$ and $\xi_\depth^{c,1}$ are arbitrary.
Therefore, one can prune a vector $\nu^2_\depth$ using $P$ applied in the space where $\occ_\depth^{c,1}$ is fixed.

As a consequence, some properties of $P$ are preserved in its extension to zsPOSGs:
\begin{itemize}
\item If $P$ correctly identifies $\nu^2_\depth$ as non-dominated at
  $\occ^{c,1}_\depth$, then
  $\langle \nu^2_\depth, \occ^{c,1}_\depth \rangle$ is
  non-dominated in $\Occ_\depth$. \\
  That is, if $P$ does not induce false negatives, neither does its
  extension to zsPOSGs.
\item If $P$ correctly identifies $\nu^2_\depth$ as dominated at
  $\occ^{c,1}_\depth$, then
  $\langle \nu^2_\depth, \occ^{c,1}_\depth \rangle$ is dominated
  in $\Occ_\depth$. \\
  That is, if $P$ does not induce false positives, neither does its
  extension to zsPOSGs.
\end{itemize}

\end{proof}

\subsubsection{About Improbable Histories} \label{app|improbableAOHs}

When solving the LP for $\beta^1_\depth$ in some \os $\occ_\depth$,
the resulting \dr{} is optimized for the \aoh{}s
$\theta^1_\depth \in \supp(\occ^{m,1}_\depth)$ only (and otherwise
random).
Also, as mentioned in \Cref{prop|rec|nu}, the corresponding vector
$\upb\nu^2_\depth$ can be obtained as a by-product of the LP if
restricted to the same \aoh{}s, and has not very relevant values for
other \aoh{}s because of the non-optimized decisions.

As a consequence, and in a view to save on time and memory, we prefer
not computing and storing \dr{}s $\beta^1_\depth$ and vectors
$\upb\nu^2_\depth$ outside the support of the $\occ^{m,1}_\depth$.
The missing values for some \aoh{}s $\tilde\theta^1_\depth$ outside
$\supp(\occ^{m,1}_\depth)$ may be required when computing new LPs, but
can then be replaced by \begin{itemize}
\item any probability distribution over actions for
  $\beta^1_\depth(\cdot|\tilde\theta^1_\depth)$ (because solving the
  LP would have led to a random choice anyway), and \item a generic upper bound such as $V^{\max}_\depth$ (\cf Proof of
  \extCref{cor|V|LC|occ}).
\end{itemize}
$V^{\max}_\depth$ is a gross (but conservative) overestimation, thus
far from informative, which impedes the convergence of the algorithm.
We now present the two approaches we considered as a replacement.

\paragraph{Initialization-based upper bound [$\upb\nu_{\text{init}}$] }
This first approach uses the vectors computed when initializing
$\upb{W}^1_\depth$ by solving the POMDP relaxation of the zs-POSG
obtained by making the opponent always select actions uniformly at
random.
This is in fact not a valid upper bound for a given $\delta^2_\depth$,
because this opponent's strategy used for the initialization may, at
least considering some \aoh{}s $\theta^1_\depth$, be better than the
current strategy defined by $\delta^2_\depth$. Yet, this upper bound turns out to give satisfying results in most of
our experiments.
We denote this heuristic $\upb\nu_{\text{init}}$.

\paragraph{bMDP upper bound  [$\upb\nu_{\text{bMDP}}$] }
This second approach is a ``belief MDP'' heuristic approximation based
on computing \begin{enumerate}
\item the optimal value function $V^*_{\textsc{mdp}}$ for the
  (finite-horizon) MDP relaxation of the POMDP obtained for
  $\delta^2_\depth$, \\ \hspace*{-\leftmargin}
  then, for a given \aoh{}
  $ \theta_{\depth+1}^1 \eqdef \langle \theta^1_\depth, a^1, z^1
  \rangle$, $\beta^2_\depth$ and $\occ^{c,1}_\depth$,
\item $b_{|\theta^1_{\depth+1}}$, the probability distribution over
  states given $ \theta_{\depth+1}^1$, $\beta^2_\depth$ and $\occ^{c,1}_\depth$, and \item the weighted sum
  $\sum_s b_{|\theta^1_{\depth+1}}(s) \cdot V^*_{\textsc{mdp}}(s)$.
\end{enumerate}
The above-mentioned belief is obtained with:
\begin{align*}
  \hspace{2cm}
  & \hspace{-2cm} Pr(s_{\depth+1} \mid \occ^{c,1}_\depth, \langle \theta^1_\depth, a^1, z^1 \rangle, \beta^2_\depth) \\
  & = \sum_{\theta^2_\depth,a^2,z^2} \sum_{s_\depth}
  Pr(s_{\depth+1}, s_\depth, \langle \theta^2_\depth, a^2, z^2 \rangle \mid \occ^{c,1}_\depth, \langle \theta^1_\depth, a^1, z^1 \rangle, \beta^2_\depth) \intertext{(by Law of total probability)}
  & = \sum_{\theta^2_\depth,a^2,z^2} \sum_{s_\depth}
  \frac{
    \overbrace{Pr(s_{\depth+1}, s_\depth, \langle \theta^2_\depth, a^2, z^2 \rangle, z^1 \mid \occ^{c,1}_\depth, \langle \theta^1_\depth, a^1 \rangle, \beta^2_\depth)}^{X}
  }{
    \sum_{s_{\depth+1}, s_\depth, \langle \theta^2_\depth, a^2, z^2 \rangle} Pr(s_{\depth+1}, s_\depth, \langle \theta^2_\depth, a^2, z^2 \rangle, z^1 \mid \occ^{c,1}_\depth, \langle \theta^1_\depth, a^1 \rangle, \beta^2_\depth)
  } \intertext{(by Bayes' Theorem), where:} X
  & = Pr(s_{\depth+1}, s_\depth, \langle \theta^2_\depth, a^2, z^2 \rangle, z^1 \mid \occ^{c,1}_\depth, \langle \theta^1_\depth, a^1 \rangle, \beta^2_\depth) \\
  & = Pr(s_{\depth+1}, z^1, z^2 \mid s_\depth, a^1, a^2)
  \cdot Pr(s_\depth, \langle \theta^2_\depth, a^2 \rangle \mid \occ^{c,1}_\depth, \langle \theta^1_\depth, a^1 \rangle, \beta^2_\depth) \\
  & = Pr(s_{\depth+1}, z^1, z^2 \mid s_\depth, a^1, a^2)
  \cdot Pr(a^2\mid \theta^2_\depth, \beta^2_\depth)
  \cdot Pr(s_\depth, \langle \theta^2_\depth \rangle \mid \occ^{c,1}_\depth, \langle \theta^1_\depth, a^1 \rangle) \\
  & = Pr(s_{\depth+1}, z^1, z^2 \mid s_\depth, a^1, a^2)
  \cdot Pr(a^2\mid \theta^2_\depth, \beta^2_\depth)
  \cdot Pr(s_\depth \mid \theta^1_\depth, \theta^2_\depth ) 
  \cdot Pr( \theta^2_\depth \mid \occ^{c,1}_\depth, \langle \theta^1_\depth, a^1 \rangle) \\
  & = \underbrace{ Pr(s_{\depth+1}, z^1, z^2 \mid s_\depth, a^1, a^2) }_{ \PP{s_\depth}{a^1,a^2}{s_{\depth+1}}{z^1,z^2} }
  \cdot \underbrace{Pr(a^2\mid \theta^2_\depth, \beta^2_\depth)}_{ \beta^2_\depth(a^2 \mid \theta^2_\depth) }
  \cdot \underbrace{ Pr(s_\depth \mid \theta^1_\depth, \theta^2_\depth ) }_{ Pr(s_\depth \mid \theta^1_\depth, \theta^2_\depth ) }
  \cdot \underbrace{ Pr( \theta^2_\depth \mid \occ^{c,1}_\depth,  \theta^1_\depth) }_{ \occ^{c,1}_\depth( \theta^2_\depth \mid  \theta^1_\depth) } \\
  & = \underbrace{\PP{s_\depth}{a^1,a^2}{s_{\depth+1}}{z^1,z^2}}_{O(z^1,z^2 \mid s_{\depth+1}, a^1, a^2) \cdot T(s_{\depth+1}|s_\depth, a^1, a^2)}
  \cdot \beta^2_\depth(a^2 \mid \theta^2_\depth)
  \cdot Pr(s_\depth \mid \theta^1_\depth, \theta^2_\depth ) 
  \cdot \occ^{c,1}_\depth( \theta^2_\depth \mid  \theta^1_\depth).
\end{align*}
We denote this heuristic $\upb\nu_{\text{bMDP}}$.
  
As a consequence, for some $\occ^{c,1}_\depth$ and $\delta^2_\depth$,
the computed $\upb\nu^2_\depth$ remains a valid upper bound of the true
vector $\nu^2_{[\occ^{c,1}_\depth, \delta^2_\depth]}$, even if
the strategy extracted from $\delta^2_\depth$ replaces unspecified
\dr{}s by any probability distribution.
In particular, the computed upper bound at $\depth=0$ remains valid.

\subsubsection{Strategy Conversion}
\label{app|stratExtraction}

As discussed in \Cref{sec|stratExtraction}, no effort is required to
extract a solution strategy for a player from the lower bound (for
$1$) or the upper bound (for $2$), but that strategy is in an unusual
recursive form.
We will here see (in the finite horizon setting) how to derive a
(unique) equivalent behavioral strategy $\beta^i_{0:}$ using
realization weights \citep{KolMegSte-stoc94} in intermediate steps.
To that end, we first define these realization weights in the case of
a behavioral strategy (rather than for a mixed strategy
as done by \citeauthor{KolMegSte-stoc94}) and present some useful
properties.

\paragraph{About Realization Weights}

Let us denote $rw^i(a^i_0, z^i_1, a^i_1, \dots, a^i_\depth)$ the {\em
  realization weight} (RW) of sequence
$a^i_0, z^i_1, a^i_1, \dots, a^i_\depth$ under strategy
$\beta^i_{0:}$, defined as
\begin{align}
  rw^i(a^i_0, z^i_1, a^i_1, \dots, a^i_\depth)
  & \eqdef \prod_{t=0}^\depth  \beta^i_{0:}(a^i_t | a^i_0, z^i_1, a^i_1, \dots, z^i_t) \\
  & = rw^i(a^i_0, z^i_1, a^i_1, \dots, a^i_{\depth-1}) \cdot \beta^i_{0:}(a^i_\depth | \underbrace{a^i_0, z^i_1, a^i_1, \dots, z^i_\depth}_{\theta^i_\depth}).
  \intertext{This definition already leads to useful results such as:}
  \beta^i_{0:}(a^i_\depth | \theta^i_\depth)
  & = \frac{
    rw^i(\theta^i_{\depth-1}, a^i_{\depth-1}, z^i_\depth, a^i_\depth)
  }{
    rw^i(\theta^i_{\depth-1}, a^i_{\depth-1})
  },
  \label{eq|betaFromRw}
\intertext{and}
  \forall z^i_\depth, \quad
  rw^i(\theta^i_{\depth-1}, a^i_{\depth-1})
  & = rw^i(\theta^i_{\depth-1}, a^i_{\depth-1}) \cdot \underbrace{\sum_{a^i_\depth} \beta(a^i_\depth | \theta^i_{\depth-1}, a^i_{\depth-1}, z^i_\depth )}_{=1} \\
  & = \sum_{a^i_\depth} rw^i(\theta^i_{\depth-1}, a^i_{\depth-1}) \cdot \beta(a^i_\depth | \theta^i_{\depth-1}, a^i_{\depth-1}, z^i_\depth ) \\
  & = \sum_{a^i_\depth} rw^i(\theta^i_{\depth-1}, a^i_{\depth-1}, z^i_\depth, a^i_\depth).
  \label{eq|RWsFromFullLengthRWs}
\end{align}

We now extend \citeauthor{KolMegSte-stoc94}'s definition by introducing
{\em conditional realization weights}, where the realization weight of a
{\em suffix sequence} is ``conditioned'' on a {\em prefix sequence}:
\begin{align}
  rw^i(\underbrace{a^i_\depth, \dots, a^i_{\depth'}}_{\text{suffix seq.}} | \underbrace{a^i_0, \dots, z^i_\depth}_{\text{prefix seq.}})
  & \eqdef \prod_{t=\depth}^{\depth'}  \beta^i_{0:}(a^i_t | a^i_0, \dots, z^i_\depth, a^i_\depth, \dots, z^i_t)
  \label{eq|rw|def} \\
  & = \beta^i_{0:}(a^i_{\depth} | a^i_0, \dots, z^i_\depth)   \cdot rw^i(a^i_{\depth+1}, \dots, a^i_{\depth'} | a^i_0, \dots, z^i_{\depth+1}).
  \label{eq|rw|rec|beta}
\end{align}
As can be noted, this definition only requires the knowledge of a
partial strategy $\beta^i_{\depth:}$ rather than a complete strategy
$\beta^i_{0:}$.

\paragraph{Mixing Realization Weights}

Let $\depth'\geq \depth+1$, and $rw^i[w]$ denote the realization
weights of some element $w$ at $\depth+1$.
Then, for some $\delta^i_\depth$, we have
\begin{align}
  rw[\delta^i_\depth](a^i_{\depth+1}, \dots, a^i_{\depth'} | a^i_0, \dots, z^i_{\depth+1})
  & = \sum_w \delta^i_\depth(w) \cdot rw[w](a^i_{\depth+1}, \dots, a^i_{\depth'} | a^i_0, \dots, z^i_{\depth+1}).
  \label{eq|rw|rec|delta}
\end{align}

\paragraph{From $w^i_0$ to $\beta^i_{0:}$}

\SetKwFunction{FExtract}{${\text{\bf Extract}}$}
\SetKwFunction{FRecGetRWMix}{${\text{\bf RecGetRWMix}}$}
\SetKwFunction{FRecGetRWCat}{${\text{\bf RecGetRWCat}}$}

\begin{algorithm}\caption{Extracting $\beta^i_{0:}$ from $w^i_0$}
  \label{alg|extractingBeta}
  
  \DontPrintSemicolon

  \Fct{\FExtract{$w^i_0$}}{
    \tcc{Step 1., keeping only $rw(\theta^i_{0:H-1})$ for all $\theta^i_{0:H-1}$}
    $\left( rw(\theta^i_{0:H-1}) \right)_{\theta^i_{0:H-1}} \gets$ \FRecGetRWMix{$0 , w^i_0$}\;
    
    \tcc{Step 2.}
    \For{$t=H-2, \dots, 0$}{
      \ForAll{$\theta^i_{0:t}, a^i_{t}$}{
        $z^i_{t+1} \gets z^i$ s.t. $\beta_t(\cdot|\theta^i_{0:t}, a^i_t, z^i)$ is defined\;
        $rw(\theta^1_{0:t}, a^i_t) \gets \sum_{a^i_{t+1}} rw(\theta^i_{0:t}, a^i_t, z^i_{t+1},a^i_{t+1} | - )$ \label{line|RWsFromFullLengthRWs}
      }
    }
    
    \tcc{Step 3.}
    \For{$t=H-1, \dots, 0$}{
      \ForAll{$\theta^i_{0:t}, a^i_t$}{
        $\beta^i_t(a^i_t | \theta^i_{0:t}) \gets  \frac{
          rw^i(\theta^i_{0:t-1}, a^i_{t-1}, z^i_t, a^i_t)
        }{
          rw^i(\theta^i_{0:t-1}, a^i_{t-1})
        }
        $
        \label{line|betaFromRw}
      }
    }
    \Return{$\beta^i_{0:}$}
  }
  
  \Fct{\FRecGetRWMix{$t , w = \langle \beta^i_t, \delta^i_t \rangle $}}{
    \For{$w'$ s.t. $\delta^i_t(w')>0$}{
      $rwCat[w'] \gets$ \FRecGetRWCat{$t,w'$} }
    \ForAll{$(a^i_0, \dots, a^i_{H-1})$}{
      $ rwMix[w](a^i_{t}, \dots, a^i_{H-1} | a^i_0, \dots, z^i_{t}) \gets \sum_{w'} {
        \delta^i_t(w')
        \cdot rwCat[w'](a^i_{t+1}, \dots, a^i_{H-1} | a^i_0, \dots, z^i_{t+1})
      }
      $
      \label{line|rw|rec|beta}
    }
    \Return{$rwMix[w]$}
  }

  \Fct{\FRecGetRWCat{$t , w = \langle \beta^i_t, \delta^i_t \rangle $}}{
    
    \eIf{$t=H-1$}{
      \ForAll{$(a^i_0, \dots, a^i_{H-1})$}{
        $ rwCat[w](a^i_{H-1} | a^i_0, \dots, z^i_{H-1})
        \gets  \beta^i_t(a^i_{H-1} | a^i_0, \dots, z^i_{H-1}) $
      }
    }{
      $rwMix[w]
      \gets $ \FRecGetRWMix{$t,w$}\;   \ForAll{$(a^i_0, \dots, a^i_{H-1})$}{
        $ rwCat[w](a^i_{t}, \dots, a^i_{H-1} | a^i_0, \dots, z^i_{t}) \gets  \beta^i_t(a^i_t|a^i_0, \dots, z^i_t)
        \cdot rwMix[w](a^i_{t+1}, \dots, a^i_{H-1} | a^i_0, \dots, z^i_{t+1})
        $
        \label{line|rw|rec|delta}
      }
    }        
    \Return{$rwCat[w]$}
  }
\end{algorithm}

Using the above results, function \FExtract in
\Cref{alg|extractingBeta} derives a behavioral strategy $\beta^i_{0:}$
equivalent to the recursive strategy induced by some tuple $w^i_0$ in
3 steps as follows:
\begin{enumerate}
\item{\bf From $w^i_0$ to $rw(\theta^i_{0:H-1}, a^i_{H-1})$
    ($\forall (\theta^i_{0:H-1}, a^i_{H-1})$) ---}
These (classical) realization weights are obtained by recursively
  going through the directed acyclic graph describing the recursive
  strategy, computing {\em full length} (conditional) realization
  weights $rw(\theta^i_{t:H-1}, a^i_{H-1} | \theta^i_{0:t})$ (for $t=H-1$ down
  to $0$).
  
  When in a leaf node, at depth $H-1$, the initialization is given by
  \Cref{eq|rw|def} when $\depth=\depth'=H-1$:
  \begin{align*}
    rw^i(a^i_{H-1} | a^i_0, \dots, z^i_{H-1})
    & \eqdef \prod_{t={H-1}}^{{H-1}}  \beta^i(a^i_t | a^i_0, \dots, z^i_t) \\
    & =  \beta^i(a^i_{H-1} | a^i_0, \dots, z^i_{H-1}).
  \end{align*}
  Then, in the backward phase, we can compute full length realization weights
  $rw(\theta^i_{t+1:H-1}, a^i_{H-1} | \theta^i_{0:t})$ with
  increasingly longer suffixes (thus shorter prefixes) using (i) \Cref{eq|rw|rec|delta} (in function \FRecGetRWMix,
  \cref{line|rw|rec|delta}) to ``mix'' several strategies using the
  distribution $\delta^i_t$ attached to the current $w$, and
(ii) \Cref{eq|rw|rec|beta}, with $\depth'=H-1$, (in function
  \FRecGetRWCat, \cref{line|rw|rec|beta}) to concatenate the
  behavioral decision rule $\beta^i_t$ attached to the current $w$ in
  front of the strategy induced by the distribution $\delta^i_t$ also
  attached to $w$.
Note: Memoization can here be used to avoid repeating the same
  computations.
\item{\bf Retrieving (classical) realization weights $rw(\theta^i_{0:t}, a^i_t | -)$ ($\forall t$) ---}
We can now compute realization weights
  $rw(\theta^i_{0:t}, a^i_t | -)$ for all $t$'s using
  \Cref{eq|RWsFromFullLengthRWs} (\cref{line|RWsFromFullLengthRWs}).
\item{\bf Retrieving behavioral decision rules $\beta^i_t$ ---}
Applying \Cref{eq|betaFromRw} (\cref{line|betaFromRw}) then provides
  the expected behavioral decision rules.
\end{enumerate}

In practice, lossless compressions are used to reduce the
dimensionality of the occupancy state (\cf \Cref{sec|XP|algorithms}),
which are currently lost in the current implementation of the conversion.
Ideally, one would like to preserve compressions whenever possible or
at least retrieve them afterwards, and possibly identify further
compressions in the solution strategy.

\section{HSVI for zs-POSGs}
\label{proofLemMaxRadius}
\label{proofLemThr}

This section presents
results that help
(i) tune zs-OMG-HSVI's radius parameter $\rho$, ensuring that
trajectories will always stop,
and (ii) then demonstrate the finite time convergence of this
algorithm.

\subsection{Algorithm}

\subsubsection{Setting $\radius$}
\label{sec|settingRadius}

\begin{restatable}[Proof in \extCshref{proofLemMaxRadius}]{proposition}{lemMaxRadius}
  \labelT{lem|MaxRadius}
  \IfAppendix{{\em (originally stated on
      page~\pageref{lem|MaxRadius})}}{}
Bounding $\lt{\depth}$ by $\l^{\infty} = \frac{1}{2} \frac{1}{1-\gamma} \left[ r_{\max} - r_{\min} \right]$
  when $\gamma<1$, and noting that
\begin{align}
    \label{eq|thr}
    \thr(\depth)
    & = 
      \gamma^{-\depth}\epsilon - 2 \radius \l^\infty \frac{\gamma^{-\depth}-1}{1-\gamma}
      \qquad (\text{ or } \epsilon - \radius (r_{\max}-r_{\min})  (2H + 1 - \depth) \depth \quad \text{ if } \gamma=1\ ),
  \end{align}
  one can ensure positivity of the threshold at any
  $\depth \in 1 \twodots H-1$ by enforcing $0  < \radius < \frac{1-\gamma}{2\l^\infty}\epsilon$ (or $0  < \radius <\frac{\epsilon}{(r_{\max}-r_{\min}) (H + 1)H}$  if $\gamma=1$).
\end{restatable}

\begin{proof}
  \label{proof|lem|MaxRadius}
  \uline{Let us first consider the case $\gamma<1$.}\\
We have (for $\depth \in \{ 1 \twodots H-1 \}$):
  \begin{align*}
    \thr(\depth)
    & = \gamma^{-\depth}\epsilon - \sum_{i=1}^\depth 2 \radius \l^\infty \gamma^{-i} \\
    & = \gamma^{-\depth}\epsilon - 2 \radius \l^\infty \sum_{i=1}^\depth \gamma^{-i} \\
    & = \gamma^{-\depth}\epsilon - 2 \radius \l^\infty \left( \gamma^{-1} + \gamma^{-2} + \cdots + \gamma^{-\depth} \right) \\
    & = \gamma^{-\depth}\epsilon - 2 \radius \l^\infty \gamma^{-1} \left( \gamma^{0} + \gamma^{-1} + \cdots + \gamma^{-(\depth-1)} \right) \\
    & = \gamma^{-\depth}\epsilon - 2 \radius \l^\infty \gamma^{-1} \frac{\gamma^{-\depth}-1}{\gamma^{-1}-1} \\
    & = \gamma^{-\depth}\epsilon - 2 \radius \l^\infty \frac{\gamma^{-\depth}-1}{1-\gamma}.
  \end{align*}
  Then, let us derive the following equivalent inequalities:
  \begin{align*}
    0
    & < \thr(\depth) \\ 2 \radius \l^\infty \frac{\gamma^{-\depth}-1}{1-\gamma}
    & < \gamma^{-\depth}\epsilon \\
\radius
    & < \frac{1}{2\l^\infty} \frac{1-\gamma}{\gamma^{-\depth}-1} \gamma^{-\depth} \epsilon \\
\radius
    & < \frac{1}{2\l^\infty} \frac{1-\gamma}{1-\gamma^\depth}  \epsilon.
\end{align*}
  To ensure positivity of the threshold for any $\depth \geq 1$, one
  thus just needs to set $\radius$ as a positive value smaller than
  $\frac{1-\gamma}{2\l^\infty}\epsilon$.

  \uline{Let us now consider the case $\gamma=1$.}\\
We have (for $\depth\in \{1,\dots,H-1\}$):
  \begin{align*}
    \thr(\depth)
    & \eqdef \epsilon - \sum_{i=1}^\depth 2 \radius \l_{\depth-i} \\
    & = \epsilon - \sum_{i=1}^\depth 2 \radius (H-(\depth-i))\cdot(r_{\max}-r_{\min}) \\
    & = \epsilon - 2 \radius (r_{\max}-r_{\min}) \left[ \depth(H-\depth) + \sum_{i=1}^\depth i \right]  \\
    & = \epsilon - 2 \radius (r_{\max}-r_{\min}) \left[ \depth H - \depth^2 + \frac{1}{2}\depth (\depth+1) \right]  \\
    & = \epsilon - 2 \radius (r_{\max}-r_{\min}) \left[  (H+\frac{1}{2}) \depth - \frac{1}{2} \depth^2 \right]  \\
    & = \epsilon - \radius (r_{\max}-r_{\min}) \left[ (2H+1) \depth - \depth^2 \right]  \\
    & = \epsilon - \radius (r_{\max}-r_{\min}) \left[ (2H + 1 - \depth) \depth \right].
  \end{align*}
  Then, let us derive the following equivalent inequalities:
  \begin{align*}
    0
    & < \thr(\depth) \\ 
    \radius (r_{\max}-r_{\min}) (2H + 1 - \depth) \depth
    & < \epsilon
    \qquad \qquad \qquad \text{(holds when $\depth=0$ and $\depth=H+1$)}
    \\
    \radius 
    & < \frac{\epsilon}{(r_{\max}-r_{\min}) (2H + 1 - \depth) \depth}
    \quad \text{(when $\depth \in \{0 \twodots H+1\}$).}
  \end{align*}
  The function
  $f: \depth \mapsto \frac{\epsilon}{(r_{\max}-r_{\min}) (2H + 1 -
    \depth) \depth}$
  reaches its minimum (for $\depth \in (0,H+1)$) when
  $\depth=H+\frac{1}{2}$.
To ensure positivity of the threshold for any $\depth \in \{1 \twodots H-1 \}$, one
  thus just needs to set $\radius$ as a positive value smaller than
  $\frac{\epsilon}{(r_{\max}-r_{\min}) (H + 1)H}$.
\end{proof}

\subsection{Finite-Time Convergence}

\subsubsection{Convergence Proof}
\label{sec|ConvergenceProof}

Proving the finite-time convergence of zs-OMG-HSVI to an error-bounded
solution requires some preliminary lemmas.

\begin{lemma}
  \labelT{lemma|OMG-HSVIContraction}
  Let $(\occ_0,\dots,\occ_{\depth+1})$ be a full trajectory generated by zs-OMG-HSVI and $\vbeta_\depth$ the joint behavioral \dr that induced the last transition, \ie, $\occ_{\depth+1}= T(\occ_\depth, \vbeta_\depth)$.
Then, after updating $\upb{W}^1_\depth$ and $\lob{W}^2_\depth$, we have that $\upb{W}^1_\depth(\occ_\depth,\beta^1_\depth) - \lob{W}^2_\depth(\occ_\depth,\beta^2_\depth) \leq \gamma \thr(\depth+1)$.
\end{lemma}

\begin{proof}
  \label{proof|lemma|OMG-HSVIContraction}
  By definition,
  \begin{align*}
    \upb{W}^1_\depth (\occ_\depth, \beta^1_\depth)
    & = \min_{ \substack{
        \langle \tilde\occ^{c,1}_\depth, \tilde{\beta}^2_\depth, \upb\nu^2_{\depth+1} \rangle \\
        \in \upb{bagW}^1_\depth
      }
    } 
      \beta^1_\depth \cdot \Big(
        r(\occ_\depth, \cdot, \tilde{\beta}^2_\depth)
        + \gamma \Nxt^1_m(\occ_\depth,\cdot, \tilde{\beta}^2_\depth) \cdot \Big[  \upb\nu^2_{\depth+1} 
+  \lt{\depth+1} \vnorm{ \nxt^1_c(\occ^{c,1}_\depth, \beta^2_\depth) - \nxt^1_c(\tilde\occ^{c,1}_\depth, \beta^2_\depth) }_1
        \Big]
    \Big).
  \end{align*}
  Therefore, after the update ($\beta^2_\depth$ and $\beta^1_\depth$
  being added to their respective bags ($\upb{bagW}^1_\depth$ and
  $\lob{bagW}^2_\depth$) along with vectors $\upb\nu^2_{\depth+1}$ and
  $\lob\nu^1_{\depth+1}$),
  \begin{align*}
    \upb{W}^1_\depth (\occ_\depth, \beta^1_\depth)
    & \leq
    \beta^1_\depth \cdot \left[
      r(\occ_\depth, \cdot, \beta^2_\depth)
      + \gamma \Nxt^1_m(\occ_\depth,\cdot, \beta^2_\depth) \cdot \upb\nu^2_{\depth+1} 
    \right], \text{ and} \\
    \lob{W}^2_\depth (\occ_\depth, \beta^2_\depth)
    & \geq
    \beta^2_\depth \cdot \left[
      r(\occ_\depth, \beta^1_\depth, \cdot)
      + \gamma \Nxt^2_m(\occ_\depth, \beta^1_\depth, \cdot) \cdot \lob\nu^1_{\depth+1} 
    \right].
\intertext{Then,}
\upb{W}^1_\depth (\occ_\depth, \beta^1_\depth) - \lob{W}^2_\depth(\occ_\depth,\beta^2_\depth)
    & \leq \left[ \cancel{r(\occ_\depth, \beta^1_\depth, \beta^2_\depth)} + \gamma T_m^1(\occ_\depth, \vbeta_\depth) \cdot \upb\nu^2_{\depth+1} \right] - \left[ \cancel{r(\occ_\depth, \beta^1_\depth, \beta^2_\depth)} +  \gamma T_m^2(\occ_\depth,\vbeta_\depth) \cdot \lob\nu^1_{\depth+1} \right]\\
    & = \gamma \left[ \upb{V}(T(\occ_\depth,\vbeta_\depth)) - \lob{V}(T(\occ_\depth,\vbeta_\depth)) \right] \\
    & \leq \gamma \thr(\depth+1)
    \qquad \qquad \text{(Holds at the end of any trajectory.)}
    \qedhere
  \end{align*}
\end{proof}

\begin{lemma}[Monotonic evolution of $\upb{W}^1_\depth$ and $\lob{W}^2_\depth$]
  \labelT{lemma|DecreaseFunctions}
Let $K\upb{W}^1_\depth$ and $K\lob{W}^2_\depth$ be the approximations
  after an update at $\occ_\depth$ with behavioral \dr
  $\langle \upb\beta^1_\depth, \lob\beta^2_\depth \rangle$
  (respectively associated to vectors $\upb\nu^2_{\depth+1}$ and
  $\lob\nu^1_{\depth+1}$).
Let also $K^{(n+1)}\upb{W}^1_\depth$ and $K^{(n+1)}\lob{W}^2_\depth$ be
  the same approximations after $n$ other updates (in various \os{}s).
Then,
\begin{align*}
\max_{\beta^1_\depth} K^{(n+1)} \upb{W}^1_\depth(\occ_\depth,\beta^1_\depth)
    & \leq \max_{\beta^1_\depth} K\upb{W}^1_\depth(\occ_\depth,\beta^1_\depth)
      \leq \upb{W}^1_\depth(\occ_\depth, \upb\beta^1_\depth) \quad \text{ and} \\
    \min_{\beta^2_\depth} K^{(n+1)} \lob{W}^2_\depth(\occ_\depth,\beta^2_\depth)
    & \geq \min_{\beta^2_\depth} K\lob{W}^2_\depth(\occ_\depth, \beta^2_\depth)
      \geq \lob{W}^2_\depth(\occ_\depth, \lob\beta^2_\depth).
  \end{align*}
\end{lemma}

\begin{proof}
  \label{proof|lemma|DecreaseFunctions}
  Starting from the definition,
\begin{align*} 
    & \max_{\beta^1_\depth} K\upb{W}^1_\depth(\occ_\depth,\beta^1_\depth) \\
    & = \max_{\beta^1_\depth} \hspace{-.5cm} \min_{ \substack{ 
        \langle \tilde\occ^{c,1}_\depth, \beta^2_\depth, \upb\nu^2_{\depth+1} \rangle \in \\
        \upb{bagW}^1_{\depth} \cup \{ \langle \occ^{c,1}_\depth, \lob\beta^2_\depth, \upb\nu^2_{\depth+1} \rangle \}
      }}
    \beta^1_\depth \cdot \bigg[
    r(\occ_\depth,\cdot,\beta^2_\depth)
    + \gamma T_m^1(\occ_\depth,\cdot,\beta^2_\depth) \cdot \Big( 
    \upb\nu^2_{\depth+1}
    +  \lt{\depth+1} \vnorm{ \nxt^1_c(\occ^{c,1}_\depth, \beta^2_\depth) - \nxt^1_c(\tilde\occ^{c,1}_\depth, \beta^2_\depth) }_1
    \Big)
    \bigg] \\
    & \leq \max_{\beta^1_\depth} \min_{\langle \tilde\occ^{c,1}_\depth, \beta^2_\depth, \upb\nu^2_{\depth+1} \rangle \in \upb{bagW}^1_{\depth}}
    \beta^1_\depth \cdot \bigg[
      r(\occ_\depth,\cdot,\beta^2_\depth)
      + \gamma T_m^1(\occ_\depth,\cdot,\beta^2_\depth) \cdot \Big(
    \upb\nu^2_{\depth+1}
    + \lt{\depth+1} \vnorm{ \nxt^1_c(\occ^{c,1}_\depth, \beta^2_\depth) - \nxt^1_c(\tilde\occ^{c,1}_\depth, \beta^2_\depth) }_1
    \Big)
    \bigg] \\
    & = \max_{\beta^1_\depth} \upb{W}^1_\depth(\occ_\depth,\beta^1_\depth) \\
    & = \upb{W}^1_\depth(\occ_\depth, \upb\beta^1_\depth).
  \end{align*}

  Then, this upper bound approximation can only be refined, so that,
  for any $n \in {\mathbf N}$,
  \begin{align*} 
    \forall \beta^1_\depth, \quad
    K^{(n+1)}\upb{W}^1_\depth(\occ_\depth,\beta^1_\depth)
    & \leq K\upb{W}^1_\depth(\occ_\depth,\beta^1_\depth), \\
    \text{thus, }
    \min_{\beta^1_\depth} K^{(n+1)}\upb{W}^1_\depth(\occ_\depth,\beta^1_\depth)
    & \leq \min_{\beta^1_\depth} K\upb{W}^1_\depth(\occ_\depth,\beta^1_\depth).
  \end{align*}

  The expected result thus holds for $\upb{W}^1_\depth$, and symmetrically for $\lob{W}^2_\depth$.
\end{proof}

\begin{lemma}
  \labelT{lemma|ComparisonVAndW}
  After updating, in order, $\upb{W}^1_\depth$ and $\upb{V}_\depth$, we have $K\upb{V}_\depth(\occ_\depth) \leq \max_{\beta^1_\depth} K\upb{W}^1_\depth(\occ_\depth,\beta^1_\depth)$.

  After updating, in order, $\lob{W}^2_\depth$ and $\lob{V}_\depth$, we have $K\lob{V}_\depth(\occ_\depth) \geq \min_{\beta^2_\depth} K\lob{W}^2_\depth(\occ_\depth,\beta^2_\depth)$.
\end{lemma}

\begin{proof}
  \label{proof|lemma|ComparisonVAndW}
  After updating $\upb{bagW}^1_\depth$, the algorithm computes
  (\Cref{alg|zsPOSGwithLP+VWs+},
  \Cref{line|computeDelta}) a new solution
  $\upb\delta^2_\depth$ of the dual LP (at $\occ^1_\depth$) and the
  associated vector $\upb\nu^2_\depth$, so that
  \begin{align*}
    \max_{\beta^1_\depth} K \upb{W}^1_\depth(\occ^1_\depth, \beta^1_\depth)
    & = \occ^{m,1}_\depth \cdot \upb\nu^2_\depth.
\intertext{This vector will feed $\upb{bagV}_\depth$ along with
      $\occ^1_\depth$, so that}
    K \upb{V}_\depth(\occ_\depth)
    & \leq \occ^{m,1}_\depth \cdot \upb\nu^2_\depth. 
\intertext{As a consequence,}
    K \upb{V}_\depth(\occ_\depth)
    & \leq \max_{\beta^1_\depth} K \upb{W}^1_\depth(\occ_\depth, \beta^1_\depth).
    \end{align*}

    The symmetric property holds for $K\lob{V}_\depth$ and $K\lob{W}^2_\depth$, which concludes the proof.
\end{proof}

\thmTermination*

\begin{proof}
  \label{proof|thm|termination}
We will prove by induction from $\depth=H$ to $0$, that the
  algorithm stops expanding \os{}s at depth $\depth$ after finitely
  many iterations (/trajectories).

First, by definition of horizon $H$, no \os $\occ_H$ is
  ever expanded.
The property thus holds at $\depth=H$.

  Let us now assume that the property holds at depth $\depth+1$ after $N_{\depth+1}$ iterations.
By contradiction, let us assume that the algorithm generates infinitely many trajectories of length $\depth+1$.
Then, because $\Occ_\depth \times {\cal B}_\depth$ is compact, after
  some time the algorithm will have visited
  $\langle \occ_\depth, \vbeta_\depth \rangle$, then, some iterations
  later, $\langle \occ_\depth^{'}, \vbeta_\depth^{'} \rangle$, such
  that $\norm{\occ_\depth - \occ_\depth^{'} }_1 \leq \rho$.
Let us also note the corresponding terminal \os{s} (because
  trajectories beyond iteration $N_{\depth+1}$ do not go further)
  $\occ_{\depth+1} = T(\occ_\depth, \beta_\depth)$ and
  $\occ_{\depth+1}^{'} = T(\occ_\depth^{'}, \vbeta_\depth^{'})$.

Now, we show that the second trajectory should not have happened, \ie, $\upb{V}(\occ_\depth^{'}) - \lob{V}(\occ_\depth^{'}) \leq \thr(\depth)$.

  Combining the previous lemmas,
  \begin{align*}
    \upb{V}(\occ_\depth^{'})
    & \leq \upb{V}(\occ_\depth) + \lt{\depth} \norm{\occ_\depth - \occ_\depth^{'}}_1
    \qquad \qquad \text{(By Lipschitz-Continuity)}\\
    & \leq \max_{\tilde{\beta}^1_\depth} \upb{W}^1_\depth(\occ_\depth,\tilde{\beta}^1_\depth) + \lt{\depth} \norm{\occ_\depth - \occ_\depth^{'}}_1
    \qquad \qquad \text{(\Cref{lemma|ComparisonVAndW})} \\
    & \leq \upb{W}^1_\depth(\occ_\depth,\beta^1_\depth) + \lt{\depth} \norm{\occ_\depth - \occ_\depth^{'}}_1
    \qquad \qquad \text{(\Cref{lemma|DecreaseFunctions})} \\
    & = \upb{W}^1_\depth(\occ_\depth,\beta^1_\depth) + \lt{\depth} \rho.
\intertext{Symmetrically, we also have}
    \lob{V}(\occ_\depth^{'})
    & \geq \lob{W}^2_\depth(\occ_\depth,\beta^2_\depth) - \lt{\depth} \rho.
\intertext{Hence,}
    \upb{V}(\occ_\depth^{'}) - \lob{V}(\occ_\depth^{'})
    & \leq \left( \upb{W}^1_\depth (\occ_\depth, \beta^1_\depth) + \lt{\depth} \rho \right)
    - \left( \lob{W}^2_\depth (\occ_\depth,\beta^2_\depth) - \lt{\depth} \rho \right) \\
    & = \left( \upb{W}^1_\depth (\occ_\depth, \beta^1_\depth) - \lob{W}^2_\depth (\occ_\depth,\beta^2_\depth) \right)
    + 2 \lt{\depth} \rho  \\
    & \leq \gamma \thr(\depth+1) + 2 \lt{\depth} \rho
    \qquad \qquad \text{(\Cref{lemma|OMG-HSVIContraction})} \\
    & = \gamma \left( \gamma^{-(\depth+1)} \epsilon - \sum_{i=1}^{\depth+1} 2 \radius \l_{\depth+1-i} \gamma^{-i}  \right) + 2 \lt{\depth} \rho \\
    & = \gamma^{-\depth} \epsilon - \sum_{i=1}^{\depth+1} 2 \radius \l_{\depth+1-i} \gamma^{-i+1} + 2 \lt{\depth} \rho \\
    & = \gamma^{-\depth} \epsilon - \sum_{j=0}^{\depth} 2 \radius \l_{\depth-j} \gamma^{-j} + 2 \lt{\depth} \rho \\
    & = \gamma^{-\depth} \epsilon - \cancel{ 2 \radius \l_{\depth-0} \gamma^{-0} } - \sum_{j=1}^{\depth} 2 \radius \l_{\depth-j} \gamma^{-j} + \cancel{ 2 \lt{\depth} \rho } = \thr(\depth).
  \end{align*}
  Therefore, $\occ^{'}_\depth$ should not have been expanded. This shows that the algorithm will generate only a finite number of
  trajectories of length $\depth$.
\end{proof}

\subsubsection{Handling Infinite Horizons}
\label{proofLemFiniteTrials}

\lemFiniteTrials*

\begin{proof}{(detailed version)}
  \label{proof|lem|finiteTrials}
  Since $W$ is the largest possible width, any trajectory stops in the
  worst case at depth $\depth$ such that
  \begin{align*}
\thr(\depth) & < \WUL \\
\gamma^{-\depth}\epsilon - 2 \radius \l^\infty \frac{\gamma^{-\depth}-1}{1-\gamma}
    & < \WUL 
    & \text{(from \Cshref{eq|thr})} \\
\gamma^{-\depth} \epsilon - 2 \radius \l^\infty \frac{\gamma^{-\depth}}{1-\gamma} 
    - 2 \radius \l^\infty \frac{-1}{1-\gamma}
    & < \WUL \\
\gamma^{-\depth} \underbrace{\left(\epsilon - \frac{2 \radius \l^\infty}{1-\gamma} \right)}_{>0 \quad \text{(\Cshref{lem|MaxRadius})}}
    & < \WUL - \frac{2 \radius \l^\infty }{1-\gamma}  \\
\gamma^{-\depth} 
    & < \frac{
      \WUL - \frac{2 \radius \l^\infty }{1-\gamma}
    }{
      \epsilon - \frac{2 \radius \l^\infty}{1-\gamma}
    }\\
\exp(-\depth\ln(\gamma)) 
    & < \exp\left(\ln\left(\frac{
          \WUL - \frac{2 \radius \l^\infty }{1-\gamma}
        }{
          \epsilon - \frac{2 \radius \l^\infty}{1-\gamma}
        }\right)\right)\\
-\depth\ln(\gamma)
    & < \ln\left(\frac{
          \WUL - \frac{2 \radius \l^\infty }{1-\gamma}
        }{
          \epsilon - \frac{2 \radius \l^\infty}{1-\gamma}
        }\right)\\
\depth\ln(\gamma)
    & > \ln\left(\frac{
          \epsilon - \frac{2 \radius \l^\infty}{1-\gamma}
        }{
          \WUL - \frac{2 \radius \l^\infty }{1-\gamma}
        }\right)\\
\depth 
    & <
    \log_{\gamma}\left(\frac{
        \epsilon - \frac{2 \radius \l^\infty}{1-\gamma}
      }{
        \WUL - \frac{2 \radius \l^\infty }{1-\gamma}
      }\right).
\qedhere
  \end{align*}
\end{proof}

Even if the problem horizon is infinite, trajectories will thus have bounded length.
Then, everything beyond this {\em effective} horizon will rely on the
upper- and lower-bound initializations and the corresponding
strategies.

\section{Experiments}
\label{sec|Experiments}

This section provides (i) information regarding the benchmark problems at hand in
\Cref{tab|BenchmarkNumberStatesActionsObservations} and (ii) supplemental experimental results.

\begin{table}
    \caption{Number of states/actions/observations for each benchmark problem}
    \centering
    \label{tab|BenchmarkNumberStatesActionsObservations}
\begin{tabular}{l|cc@{ }cc@{ }c}
      \toprule
& $\mathcal{S}$ &  $\mathcal{A}^1$ &  $\mathcal{A}^2$ & $\mathcal{O}^1$ & $\mathcal{O}^2$ \\
\midrule
      Competitive Tiger & 2 & 4 & 4 & 3 & 3  \\
      Adversarial Tiger & 2 & 3 & 2 & 2 & 2 \\
      Recycling Robot & 4 & 3 & 3 & 2 & 2 \\
      Mabc & 4 & 2 & 2 & 2 & 2 \\
      \bottomrule
\end{tabular}
\end{table}

Graphs in \Cref{fig|ConvergenceGraphs} show how the upper- and
lower-bounding values at $\occ_0$ (\ie, $\upb{V}_0(\occ_0)$ (red
curve) and $\lob{V}_0(\occ_0)$ (blue curve)) evolve as a function of
the number of iterations, here considering the same benchmarks and
time horizons as in \Cref{sec|XPs} (except for $H=2$).

As expected, these bounds monotonically converge toward the optimal
value (here provided by Sequence Form in all cases but Recycling Robot
for $H=6$).
This convergence would be symmetric in Competitive Tiger, the only
symmetric game, if it were not for the algorithm breaking ties in a
biased manner when multiple equivalent solutions are possible.

The dotted red (respectively blue) curve is obtained by simply
removing from $\upb{bagV}_0$ and (resp. $\lob{bagV}_0$) its initial
element.
This curve somehow allows better observing the actual (hidden)
progress of the convergence at the beginning since the resulting value
is always obtained (even in first iterations) from updates.
A solid curve is thus constant until joined by the corresponding
dotted curve.
In the case of Adversarial Tiger, the ``merge'' takes place earlier for the upper bound than for the lower bound for $H=5$, while the opposite is observed for $H=4$.
Such a difference might be due to the optimal strategies being
possibly very different when the horizon is extended (\eg, as also
observed in the Tiger POMDP, or the Multi-Agent Tiger Dec-POMDP), but
needs to be further investigated.

Note also that iteration durations are expected to increase with the
time horizon, so that it is surprising to observe similar numbers of
iterations in 24 hours in both Competitive Tiger ($H=4$ vs $5$) and
Recycling Robot ($H=5$ vs $6$).
This phenomenon is currently under investigation, but probably linked
to the very efficient compression in these two problems.

\newcommand\evolGraphCaption[1]{Evolution of the upper- and
  lower-bound value $\upb{V}_0(\occ_0)$ (in red) and
  $\lob{V}_0(\occ_0)$ (in blue) of \omgHSVIlccc{} for the #1 problem
  as a function of the number of iterations (generated trajectories).
Optimal value found by Sequence Form LP in green for reference (when available). 
}
\newcommand\evolGraphSubCaption[1]{#1}
\def\evolGraphScale{0.33} 

\newcommand{\dotvfill}{\par\leaders\hbox{$\cdot$}\vfill}

\begin{figure}\centering
    \includegraphics[width=\evolGraphScale\linewidth]{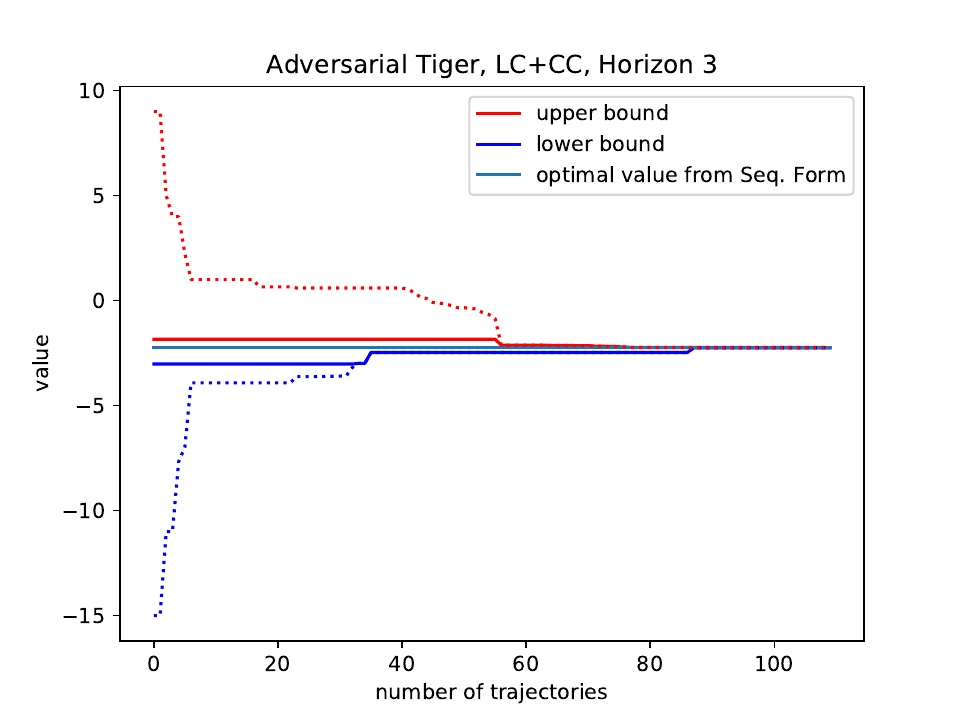}\hfill
    \includegraphics[width=\evolGraphScale\linewidth]{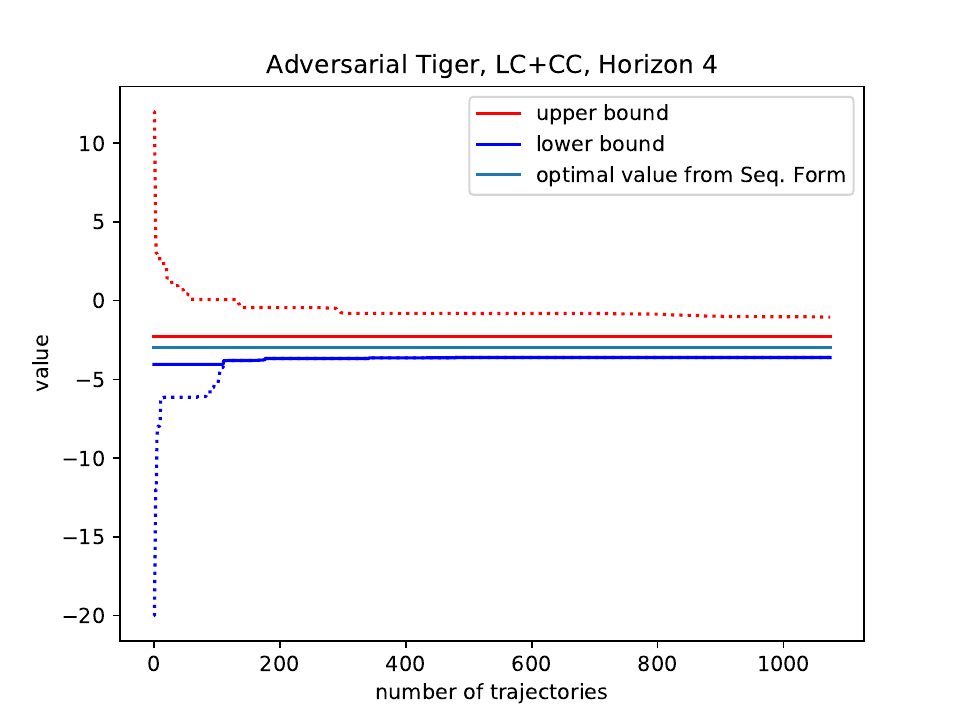}\hfill
    \includegraphics[width=\evolGraphScale\linewidth]{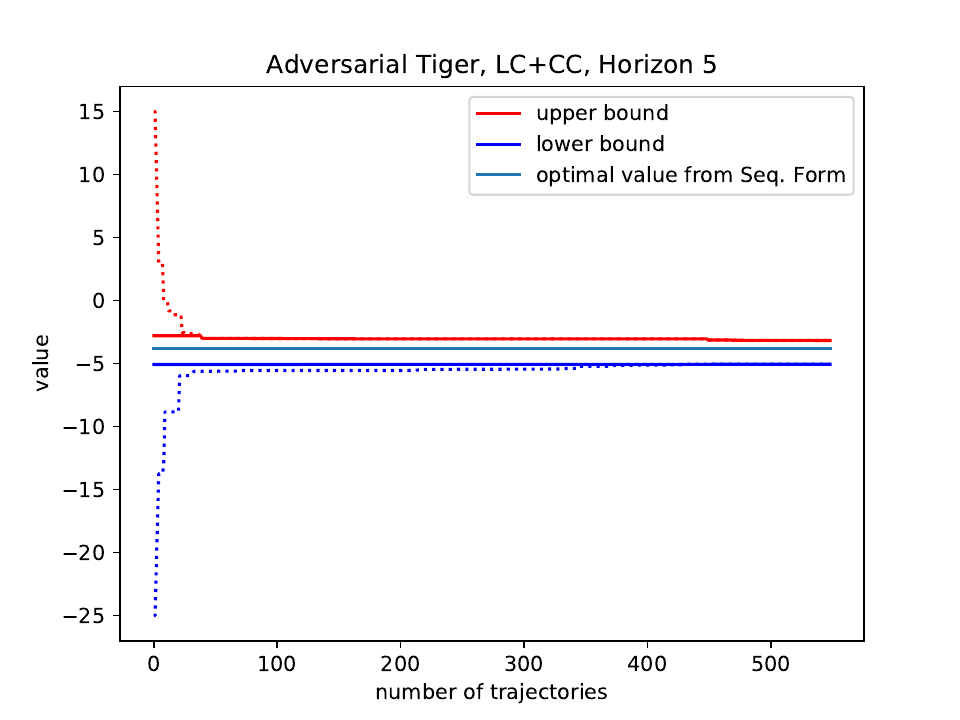}

    \dotfill (a) \evolGraphSubCaption{Adversarial Tiger}\dotfill 
\label{fig|ConvergenceAdvTiger}

    \centering
    \includegraphics[width=\evolGraphScale\linewidth]{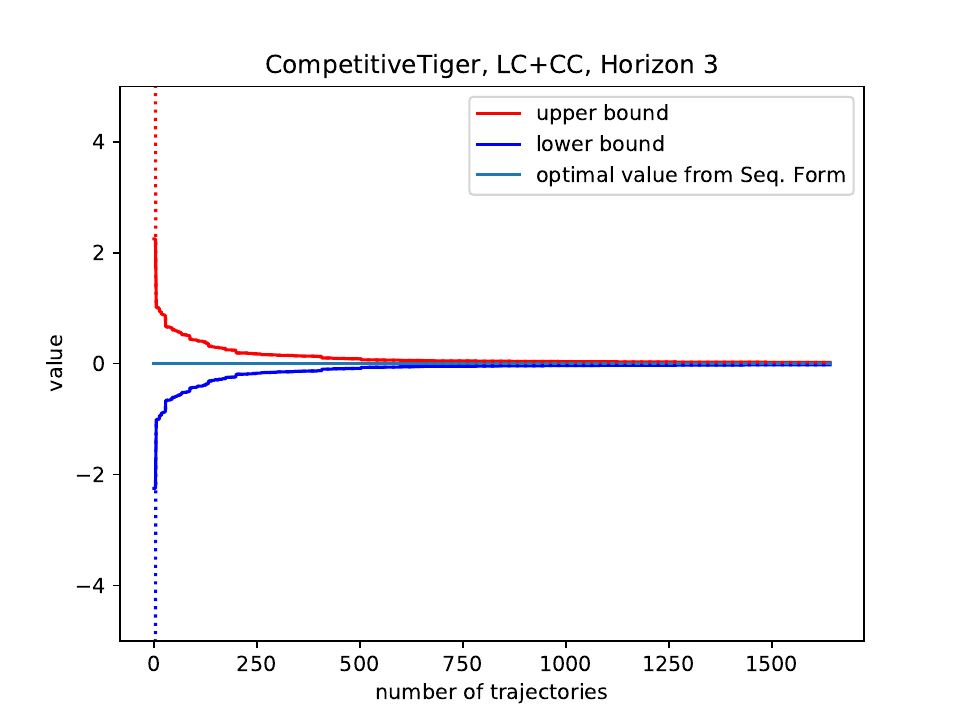}\hfill
    \includegraphics[width=\evolGraphScale\linewidth]{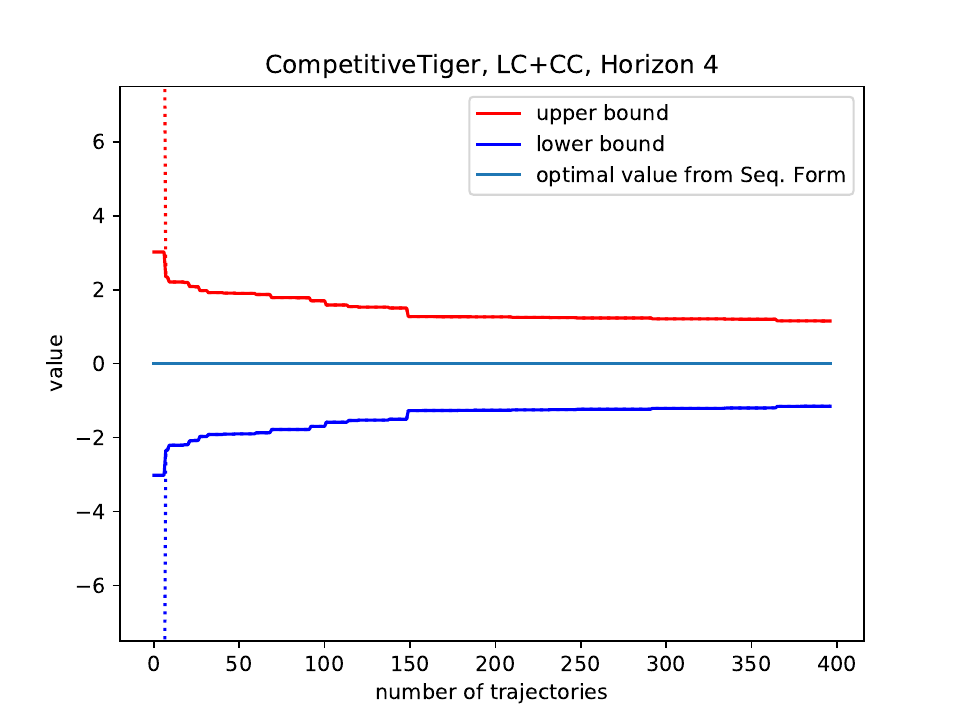}\hfill
    \includegraphics[width=\evolGraphScale\linewidth]{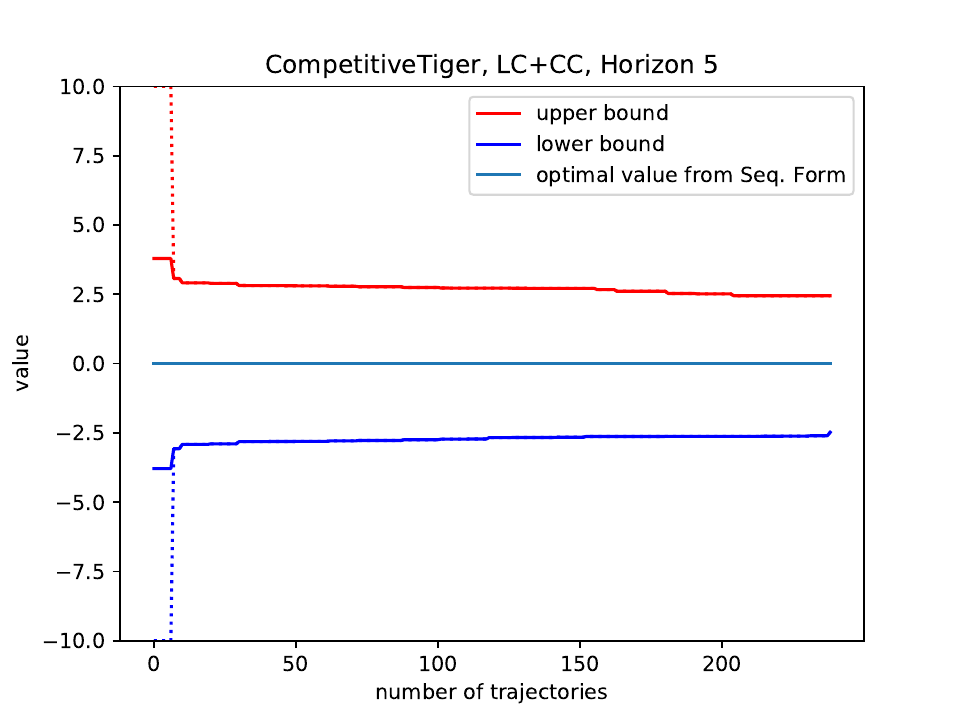}

    \dotfill (b) \evolGraphSubCaption{Competitive Tiger}\dotfill 
\label{fig|ConvergenceCompTiger}

    \centering
    \includegraphics[width=\evolGraphScale\linewidth]{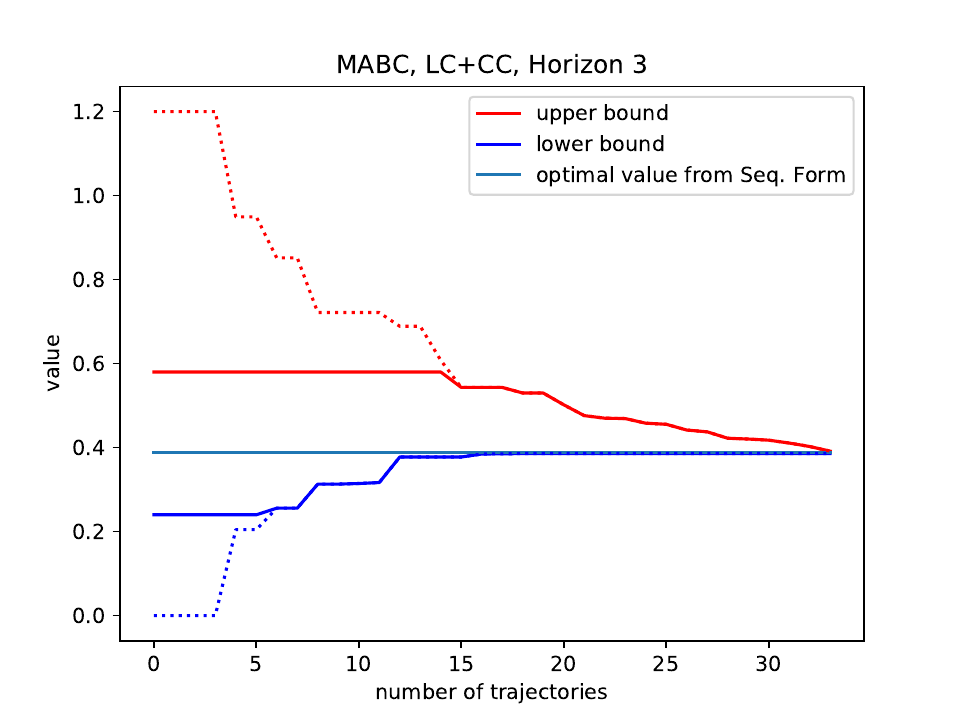}\hfill
    \includegraphics[width=\evolGraphScale\linewidth]{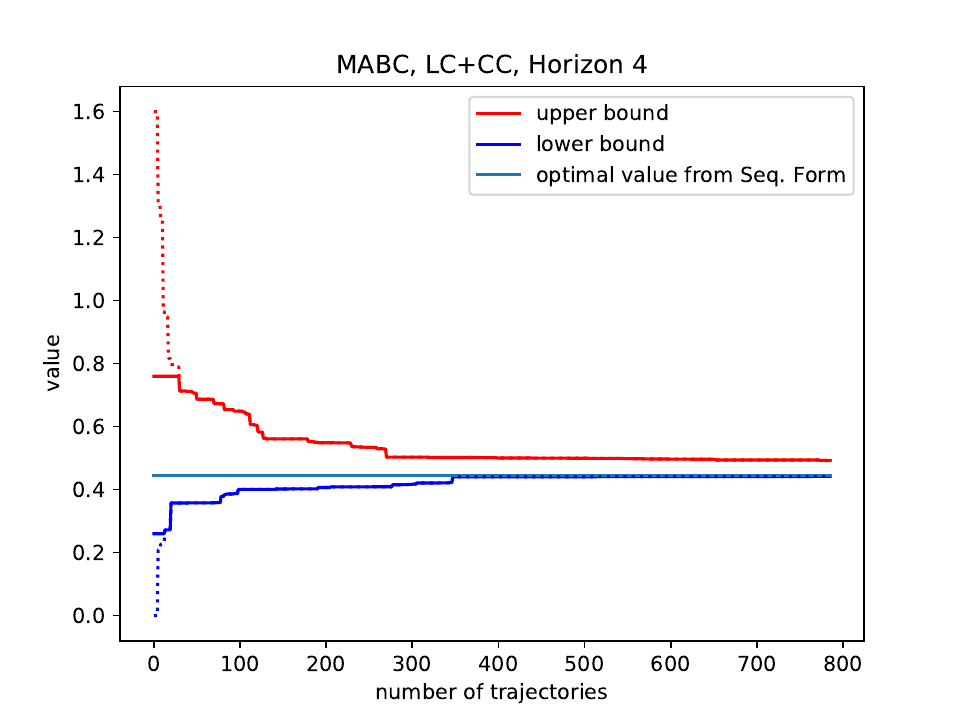}\hfill
    \includegraphics[width=\evolGraphScale\linewidth]{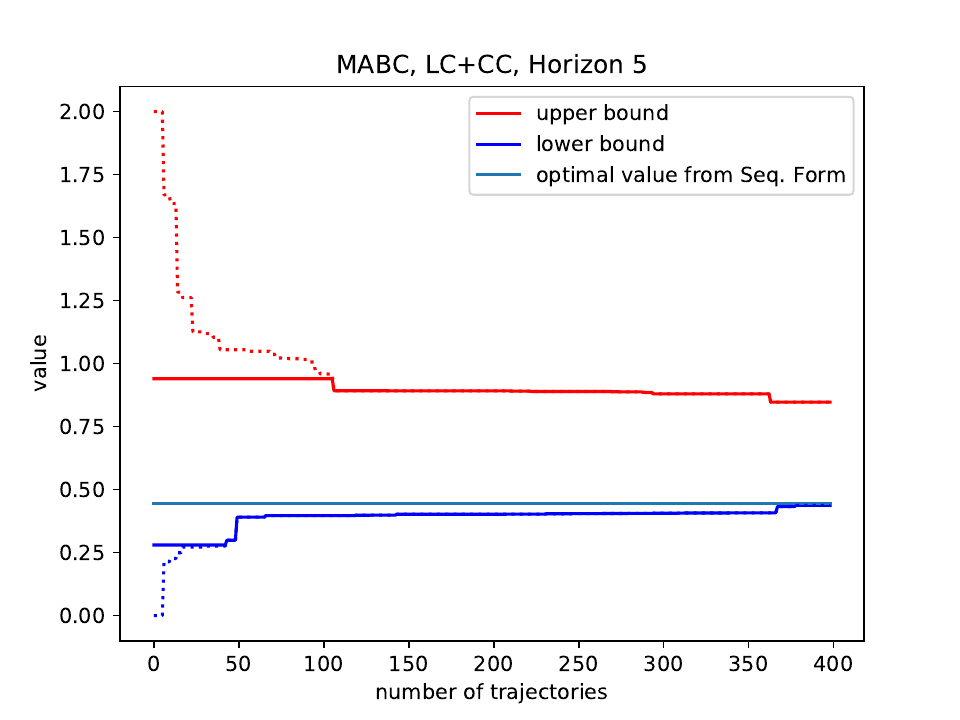}

    \dotfill (c) \evolGraphSubCaption{Mabc}      \dotfill 
    
    \includegraphics[width=\evolGraphScale\linewidth]{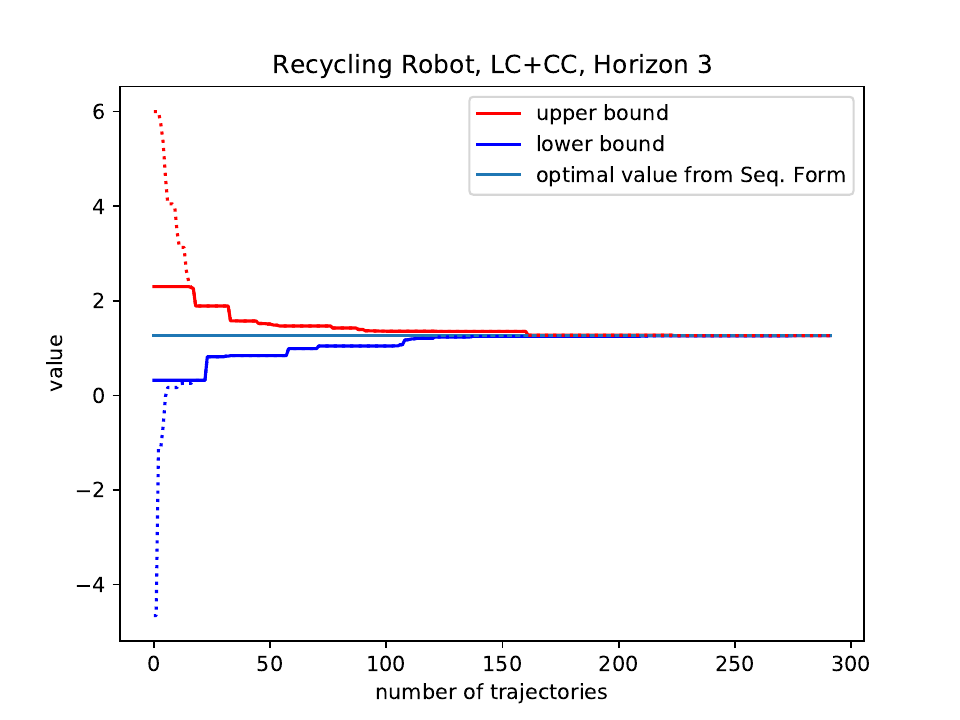}\hfill
    \includegraphics[width=\evolGraphScale\linewidth]{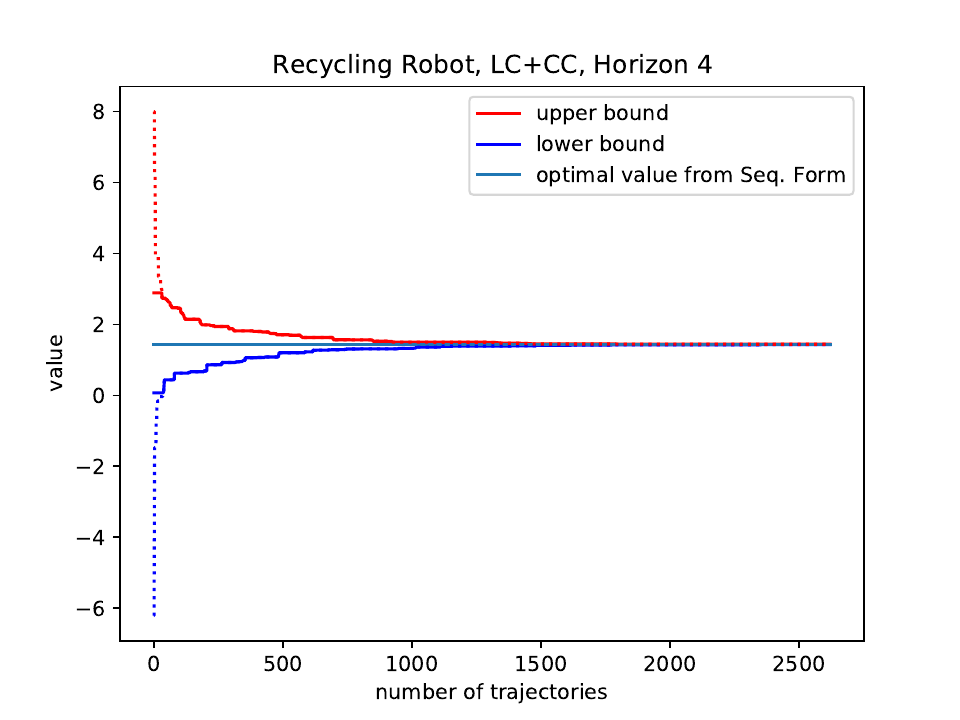}\hfill
    \includegraphics[width=\evolGraphScale\linewidth]{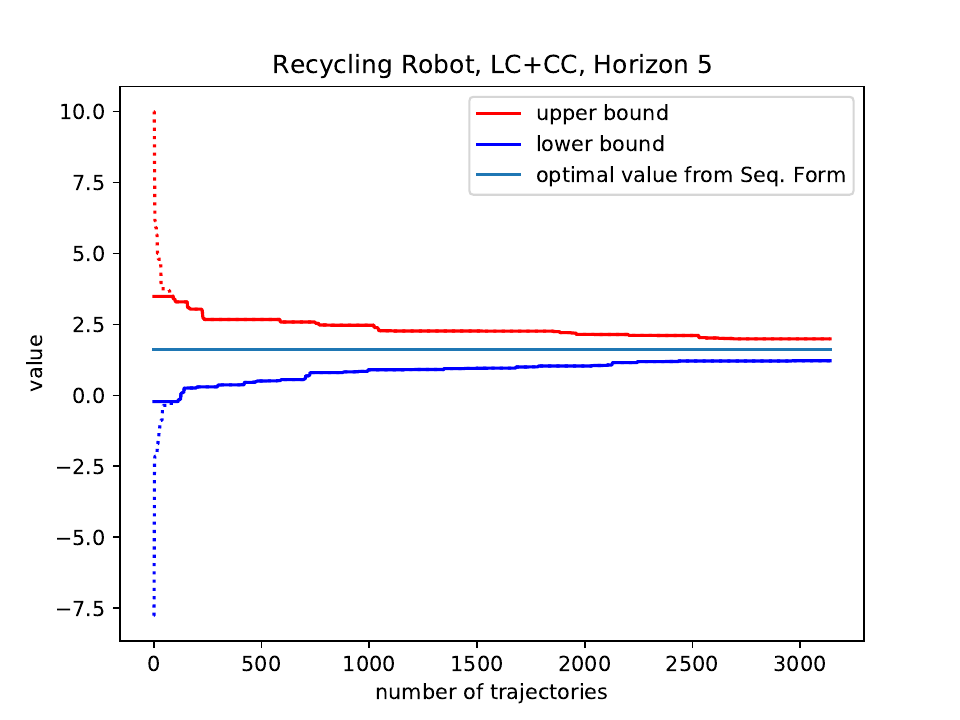}

    \includegraphics[width=\evolGraphScale\linewidth]{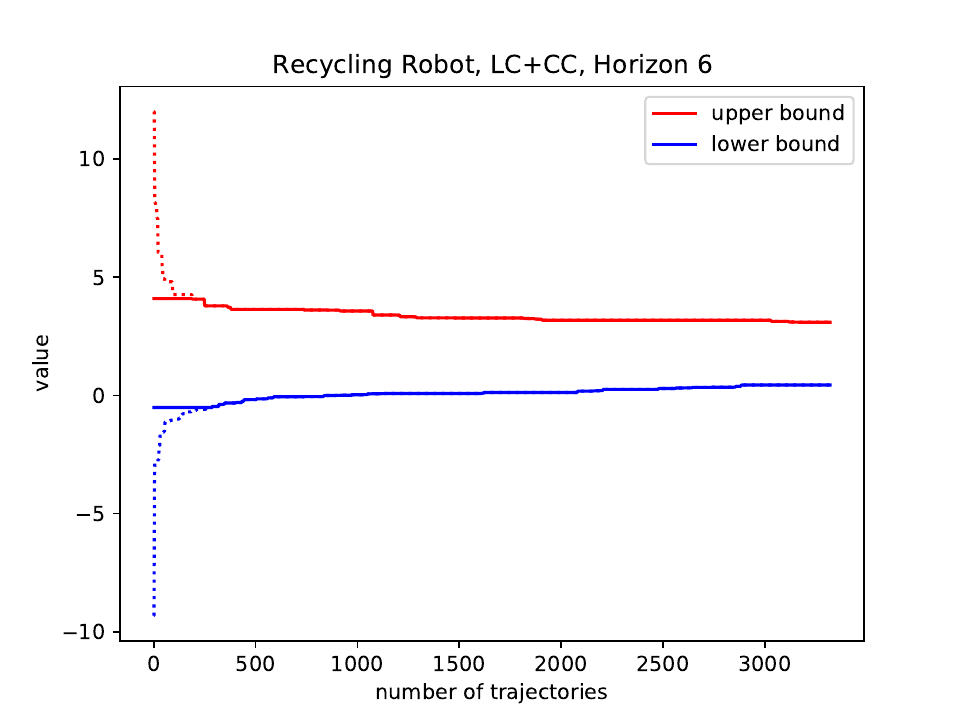}

    \dotfill (d) \evolGraphSubCaption{Recycling Robot}       \dotfill

\label{fig|ConvergenceRecyclingRobot}

    \centering
    \caption{
      Evolution of the upper- and
      lower-bound value $\upb{V}_0(\occ_0)$ (in red) and
      $\lob{V}_0(\occ_0)$ (in blue) of \omgHSVIlccc{} for the various benchmark problems
      as a function of the number of iterations (generated trajectories).
Dotted curves: same bounding approximations when removing each bag's ($\upb{bagV}_0$ and $\lob{bagV}_0$) initial elements.
In green: reference optimal value found by Sequence Form LP
      (when available).
}
\label{fig|ConvergenceGraphs}\end{figure}

\section{Relying on Lipschitz-Continuity Alone}
\label{app|LipschitzOnly}

This appendix demonstrates that one can derive an other version of
HSVI for solving zs-POSGs using only $V^*$ Lipschitz-continuity
(\Cref{cor|V|LC|occ}, p.~\pageref{cor|V|LC|occ}).
To that end, we
\begin{enumerate}
\item first describe appropriate LC upper and lower bound
  approximations; \item then discuss the various operators (initialization, update and
  pruning) required by these approximations; and
\item finally explain how the local game faced in each visited
  $\occ_\depth$ is solved as bi-level optimization problem with an
  error-bounded algorithm based on \citeauthor{Munos-ftml14}'
  Deterministic Optimistic Optimization (DOO) algorithm
  (\citeyear{Munos-ftml14}).
\end{enumerate}

We focus here on upper bounding $V^*$, as lower bounding is a
symmetrical problem.

\subsection{LC-only Approximation of $V^*$}

$V^*_\depth$ being LC, we define an upper bound approximation
$\upb{V}_\depth$ at depth $\depth$ as a lower envelope of (i) an initial upper bound $\upb{V}^{(0)}_\depth(\occ_\depth)$ and (ii) downward-pointing L1-cones, where an upper-bounding cone
$\upb{c} = \langle \upb{\occ}_\depth, \upb{v}_\depth \rangle $---located at
$\upb \occ_\depth$, with ``summit'' value $\upb{v}_\depth$, and slope
$\l_{(H-\depth)}$---induces a function
$\upb{V}^{(\upb c)}_\depth(\occ_\depth) \eqdef \upb{v}_\depth + \l_{(H-\depth)}
\norm{\upb\occ_\depth-\occ_\depth}_{\p}$.
Formally, the set of cones being denoted
$\upb{bagC}_\depth$,
\begin{align*}
  \upb{V}_\depth(\occ_\depth)
  & = \min\{ \upb{V}^{(0)}_\depth(\occ_\depth), \min_{\upb c \in \upb{bagC}_\depth} \upb{V}^{(\upb c)}_\depth(\occ_\depth) \}.
\end{align*}

\subsection{Related Operators}

\paragraph{Initialization}
With such an approximation, the same relaxations proposed in
\Cshref{sec|relatedOperators} can be used to initialize $\upb{V}$.
Note that that the initialization is required to
\begin{itemize}
\item be a LC function (because the LC property is needed to solve
  local games); and
\item come with a default decision rule at each time step for $2$
  (that shall be used as a default (safe) strategy to execute).
\end{itemize}

\paragraph{Update}
Then, updating this approximation when in $\occ_\depth$ requires
solving the (infinite) local game defined by
$\upb{Q}_\depth(\occ_\depth,\beta^1_\depth,\beta^2_\depth) \eqdef r(\occ_\depth,\beta^1_\depth,\beta^2_\depth) + \gamma \upb{V}_\depth(\occ_\depth,\beta^1_\depth,\beta^2_\depth))$,
what is enabled by the following property.

\begin{restatable}[Proof in \extCshref{proofLemQQLC}]{lemma}{lemQQLC}
  \labelT{lem|QQ|LC}
In any $\occ_\depth$, the induced local game
  $\upb Q_\depth(\occ_\depth,\beta^1_\depth,\beta^2_\depth) \eqdef r(\occ_\depth,\beta^1_\depth,\beta^2_\depth) + \gamma \upb
  V_{\depth+1}(T(\occ_\depth,\beta^1_\depth,\beta^2_\depth))$
  is \ifextended{Lipschitz continuous}{LC} in both $\beta^1_\depth$
  and $\beta^2_\depth$.
\end{restatable}

\label{proofLemQQLC}

\begin{proof}
  \label{proof|lem|QQ|LC}
In this local game's definition,
  \begin{align*}
    \upb Q_\depth(\occ_\depth,\beta^1_\depth,\beta^2_\depth)
    & \eqdef r(\occ_\depth, \beta^1_\depth, \beta^2_\depth)
      + \gamma \upb V_{\depth+1}(\nxt(\occ_\depth, \beta^1_\depth, \beta^2_\depth)),
  \end{align*}
  \begin{itemize}
  \item the first term (reward-based) is
$\l_r$-LC in each $\beta^i_\depth$ (\Cref{lem|occSufficient},
    p.~\pageref{lem|occSufficient}), with
    $\l_r= \frac{r_{\max} - r_{\min}}{2}$, and
\item the second term is $(\gamma \cdot \l_{\depth+1} \cdot 1)$-LC in each
    $\beta^i_\depth$, 
due to (i) $\upb{V}_{\depth+1}$ being $\lt{\depth+1}$-LC (\Cref{cor|V|LC|occ},
    p.~\pageref{cor|V|LC|occ}), and (ii) $\nxt$ being linear in $\beta^i_\depth$
    (\Cref{lem|occSufficient}, p.~\pageref{lem|occSufficient}).
  \end{itemize}
$Q_\depth$ is thus $\l^{Q}_\depth$-LC with
  $\l^{Q}_\depth = \l_r + \gamma \cdot \lt{\depth+1}$.
\end{proof}

We provide an algorithm for solving such a bi-level optimization
problem given the known Lipschitz constant in \Cref{sec|DOO}.
It returns $\upb\beta^1_\depth$ to guide the trajectory, and the
associated value $\upb{v}$ to create a new cone at $\occ_\depth$.

\paragraph{Execution}

But let us point out that $\upb{v}$ is a worst expected value for $2$
if (i) applying $\upb\beta^2_\depth$, solution of the dual bi-level
optimization problem, and (ii) then acting so as to obtain at most
$\upb{V}_{\depth+1}(\nxt(\occ_\depth, \upb\beta^1_\depth,
\upb\beta^2_\depth))$.
We see here appearing again a recursive definition of strategies for
$2$ with guaranteed worst-case value.
This leads to storing along with each upper-bounding cone both (i) $\upb\beta^2_\depth$ and (ii) a pointer to the (single) cone at $\depth+1$ (or the initial
value approximation) involved in computing $\upb{v}$, which will allow
deducing decision rules to apply from $\depth+1$ on.

Then, once the algorithm has converged, $2$ (resp. $1$) can retrieve a
solution strategy to perform using the (recursive) strategy attached
to $\upb{V}_0(\occ_0)$ (resp.  $\lob{V}_0(\occ_0)$).

\paragraph{Pruning}

In this setting where all cones have the same slope, a cone
$\upb{c} = \langle \upb{\occ}_\depth, \upb{v}_\depth \rangle $ of
$\upb{V}_\depth$ can be pruned if and only if it is dominated by
another cone at $\upb\occ_\depth$.

\subsection{Bi-level Optimization with DOO}
\label{sec|DOO}

The discussion below explains how to compute an optimal
strategy profile $\langle \beta^1_\depth, \beta^2_\depth \rangle$ at
any time step $\depth$ for any occupancy state $\occ_\depth$ for the
local game
$\upb Q_\depth(\occ_\depth,\beta^1_\depth,\beta^2_\depth) \eqdef
r(\occ_\depth, \beta^1_\depth, \beta^2_\depth) + \gamma \upb
V_{\depth+1}(\nxt(\occ_\depth, \beta^1_\depth, \beta^2_\depth))$. 
Unfortunately, $\upb Q_\depth(\occ_\depth,\beta^1_\depth,\beta^2_\depth)$ is not convex or concave, so that a method relying on differentiability could return a local optimum instead of the optimal value.
Moreover, we require the algorithm to find an $\epsilon$-optimal solution in finite time, asymptotic convergence guarantees being insufficient.

\paragraph{Lipschitz Optimization}
\citeauthor{Munos-ftml14}' DOO (Deterministic Optimistic Optimization) algorithm (\citeyear{Munos-ftml14}) is a finite-time $\epsilon$-optimal algorithm that computes $\max_{x} f(x)$ for a $\lambda$-Lipschitz function $f$ (see \Cref{alg|DOOBiDOO}, top).
It iteratively covers up a compact search space by other compact sets
(the current cover being here denoted $(R_i)_{i\in \cI}$) on which the
optimal value function can be upper bounded thanks to its Lipschitz
continuity.
Each set $R_i$ is attached a reference point $x_i$, its value
$f(x_i)$, and its radius $r_i$,\footnote{$R_i$ is contained in the
  ball of center $x_i$ and radius $r_i$.} so that the value of $f$ on
$R_i$ is upper bounded by $f(x_i) + \l r_{i}$.
With this, the algorithm starts with a few compact sets and
iteratively subdivides the one whose upper bound is the largest.
Repeating this, the algorithm converges towards an $\epsilon$-optimal solution in finite time.
\Cref{fig|DOOTree} \citep{Munos-ftml14} represents a subdivision tree for the optimisation of $f(x) = [sin(13x)sin(27x)+1]/2$. Note that, except for the subdividing process that will be discussed later, this algorithm is generic.

\begin{algorithm}
  \caption{DOO \& BiDOO}
\label{alg|DOOBiDOO}
\DontPrintSemicolon
  \SetKwInOut{Input}{input}
  \SetKwInOut{Output}{output}
  \SetKwFunction{DOO}{{\bf DOO}}
  \SetKwFunction{BiDOO}{{\bf BiDOO}}
  \SetKwFunction{ScoBiDOO}{{\bf ScoBiDOO}}
\Fct{\DOO{$\Big[ \mathcal{D} \to \reals \,;\, x\mapsto f(x) \Big], \epsilon, \l$}}{ \Input{$f: \reals \to \reals$ is a $\lambda$-Lipschitz function} Initialize $\cI$ and $(R_i)_{i \in \cI}$ s.t. $\mathcal{D}  \subseteq \cup_{i \in \cI} R_i$ \;
    \While{ $\max_{i \in \cI} \left( f(x_i) + \l r_{i} \right) - \max_{i \in \cI} f(x_i) > \epsilon$}{
      $i^* \gets \argmax_{i \in \cI} f(x_i) + \l r_i$\;
      Subdivide $R_{i^*}$ into $\cup_{j \in \cI^*} R_j \quad (\supseteq R_{i^*})$  \tcp*{\scriptsize \Cref{theo|IntersectionHypercubeSimplexe} allows building $\cI^*$ in the case of simplexes.}
$\forall j \in \cI^*, x_j \gets Center(R_j)$\;
      $\cI \gets [\cI \setminus i^*] \cup \cI^{*}$\;
    }
    \Return{$\langle \arg \& \max_{x_i} f(x_i) \rangle $}\;
  }
\medskip
\Fct{\BiDOO{$\Big[  \mathcal{D}_x\times\mathcal{D}_y \to \reals \,;\, x, y \mapsto f(x,y) \Big], \epsilon_1, \epsilon_2, \l)$}}{ \Input{ $f: \mathcal{D}_x\times\mathcal{D}_y \to \reals$ a $\lambda$-Lipschitz function (in both $x$ and $y$)} $\langle x_{\max}, v_{\max} \rangle \gets$ \DOO{ $\Big[ \mathcal{D}_x \to \reals \,;\, x \mapsto -getMax\Big( $ { \DOO{$ \Big[ \mathcal{D}_y \to \reals \,;\, y \mapsto -f(x,y) \Big], \epsilon_2, \l$} } $ \Big) \Big]$, \linebreak \phantom{$\langle x_{\max}, v_{\max} \rangle \gets$ \DOO(}
      $ \epsilon_1, \l $}\; \tcc{\scriptsize $getMax(\cdot,\cdot)$ here returns its second argument, \ie, the maximum of the inner DOO computation.}
  }
  
  \Return{$\langle x_{\max}, v_{\max} \rangle$}\;

\end{algorithm}

\begin{figure}
    \centering
\includegraphics[width=0.45\textwidth]{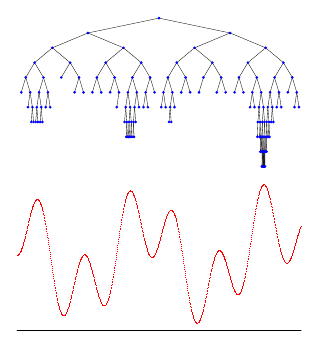}
    \caption{Example tree obtained while maximizing $f(x) = [sin(13x)sin(27x)+1]/2$ (Fig.~3.7 from \citet{Munos-ftml14})}
    \label{fig|DOOTree}
\end{figure}

\paragraph{Lipschitz bi-level Optimization}
We solve our bi-level optimization problem (finding $\max_{x} \min_{y} f(x,y)$ for a $\lambda$-Lipschitz function $f$) by using two nested DOO processes, \ie,
\begin{itemize}
\item an outer $\epsilon_1$-optimal DOO maximizing the function
  $x \mapsto \min_y f(x,y)$, using the solution of 
\item an inner $\epsilon_2$-optimal DOO minimizing the function
  $y \mapsto f(x,y)$ for fixed $x$
\end{itemize}
(see \Cref{alg|DOOBiDOO}, bottom).
\citeauthor{Munos-ftml14}' proof straightforwardly adapts to this
case.
The final error is then $\epsilon=\epsilon_1+\epsilon_2$.

\begin{figure}
    \centering
    \includegraphics[width=0.75\textwidth]{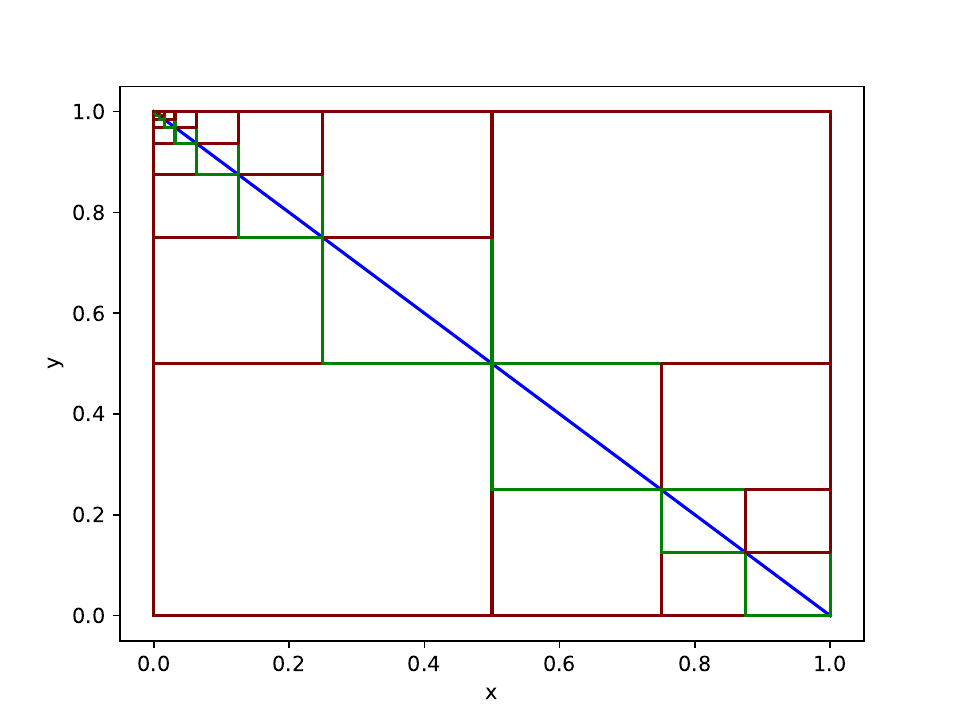}
    \caption{Example of iterative subdivision of the 2D simplex using squares. The optimized function is $(x,y) \to x + 1.01y^2$. The blue diagonal represents valid probability distributions (\ie, $x+y=1$). Ignored hypercubes are colored in brown and retained ones in green.}
    \label{fig|IntersectionHypercubeSimplex}
\end{figure}

\paragraph{Subdividing Probability Simplexes}
In our setting, we need to subdivide an $n$-dimensional probability
simplex (in fact, one $|\cA|$-dimensional simplex per individual
action-observation history).
To that end, we propose starting with the $n$-dimensional hypercube (\ie a closed ball $B(x_i,r_i) = \{x \in \mathbb{R}^n \mid \norm{ x - x_i }_{\infty} \leq r_i\}$ of radius $r_i$ and whose center is $x_i$)
that contains this simplex (with lowest corner $(0, \dots, 0)$ and
highest corner $(1, \dots, 1)$) and subdividing it in $2^n$ smaller
hypercubes, but only keeping the ones that intersect with the simplex.
Then, for a given hypercube, one will evaluate the function to
optimize at the center of the intersection between the hypercube and
the simplex.
The following theorem shows how to determine whether an hypercube intersects the simplex or not, and how to compute its reference point.

\begin{theorem}[Intersection bewteen the $n$-dimensional unit simplex and an $n$-dimensional hypercube]
  \labelT{theo|IntersectionHypercubeSimplexe}
Let $n \in \mathbb{N}^* \setminus \{1\}$.
Let $H$ be an $n$-dimensional cube (\ie a closed ball for $\norm{\cdot}_{\infty}$) of radius $\eta$ whose center is called $m$, and
let $\mathcal{S}(1)$ be the unit simplex in dimension $n$.
Then, $H$ and $S(1)$ intersect (\ie, $(H \cap S(1)) \neq \emptyset$) if and only if $\exists(x_i,x_j)$ such that $\sum_{k=1}^n x_i^k \leq 1$ and $\sum_{k=1}^n x_j^k \geq 1$, and the unique intersection point can be computed.
\end{theorem}

\begin{proof}
  \label{proof|theo|IntersectionHypercubeSimplexe}
  Let us consider the diagonal from the lowest point ($x_{inf} = (m_1 - \eta,\dots,m_n - \eta)$) to the highest point ($x_{sup} = (m_1 + \eta, \dots, m_n + \eta)$), and use the Intermediate Value Theorem on it.
There is no intersection point if $\sum_k x_{inf}^k > 1$ or $\sum_k x_{sup}^k < 1$.
Else, the intersection point is $x = x_{inf} + t (x_{sup}-x_{inf})$, where $t =  \frac{1 - \sum_{k=1}^n x_{inf}^k}{n \cdot \norm{ x_{inf} - x_{sup} }_{\infty}}$.
Indeed, 
\begin{align}
  \sum_{k=1}^n x^k = 1
  & \iff \sum_{k=1}^n \left[ x_{inf}^k + t(x_{sup}^k - x_{inf}^k) \right] = 1 \\
  & \iff t = \frac{1 - \sum_{k=1}^n x_{inf}^k}{\sum_k x_{sup}^k - x_{inf}^k} \\
  & \iff t = \frac{1 - \sum_{k=1}^n x_{inf}^k}{n \cdot \norm{ x_{sup}-x_{inf} }_{\infty}}.
\end{align}

\end{proof}

\Cref{fig|IntersectionHypercubeSimplex} illustrates the iterative subdivision of the $2$-dimensional unit simplex by $2$-dimensional hypercubes (\ie, squares).
Let us point out that the subdivision operation is here concentrated around the optimum: $(0,1)$.

\message{>>> Leaving \currfilename >>>}

 }{}

\end{document}